\documentclass[journal]{IEEEtran}

\usepackage{blindtext}
\usepackage{multicol}
\usepackage[cmex10]{amsmath}
\usepackage{environ}         
\usepackage{etoolbox}        

\usepackage{graphicx}        
\graphicspath{{figures/}}

\usepackage{rotating}
\usepackage{adjustbox}

\usepackage{xcolor}
\usepackage{mathtools}

\DeclarePairedDelimiterX\Basics[1](){ #1}

\usepackage{dsfont}
\usepackage{amsmath,amssymb}

\usepackage{makecell}
\setcellgapes{1.5pt}
\usepackage{amsmath}
\usepackage{stackengine}

\usepackage{bm}
\usepackage{array}
\usepackage{tabu}
\usepackage{multirow}
\usepackage{graphicx}

\usepackage[colorlinks=true, linkcolor    = blue, urlcolor     = blue, citecolor=blue, pdfborder={0 0 0}]{hyperref}
\usepackage{algorithmic} 
\usepackage{algorithm,float}
\usepackage{subcaption}
\usepackage{bm}
\usepackage{graphicx}
\usepackage{amsmath}
\usepackage{soul}
\usepackage{multirow}
\usepackage{amsthm,amssymb,amsmath}
\usepackage{mathrsfs}
\DeclareMathAlphabet{\mathpzc}{OT1}{pzc}{m}{it}
\usepackage{mathtools}

\usepackage{stackengine}

\theoremstyle{remark}
\newtheorem{definition}{Definition}

\newtheorem{example}{Example}

\newtheorem{proposition}{Proposition}
\newtheorem{lemma}{Lemma}
\newtheorem{remark}{Remark}

\newtheorem{requirement}{Criterion}

\usepackage{fontawesome5}
\usepackage{academicons}
\usepackage{xcolor}
\usepackage{amsmath}
\definecolor{dukeblue}{rgb}{0.0, 0.0, 0.61}
\definecolor{harvardcrimson}{rgb}{0.79, 0.0, 0.09}

\newcommand{\txbl}[1]{\textcolor{black}{#1}}

\newcommand{\orcid}[1]{\href{https://orcid.org/#1}{\textcolor[HTML]{A6CE39}{\faOrcid}}}

\usepackage{amsmath}
\usepackage{cite}
\usepackage{xfrac}
\usepackage{graphicx}
\usepackage[font={small}]{caption}
\usepackage{ptext}
\usepackage{float}
\usepackage{setspace}
\DeclareSymbolFont{matha}{OML}{txmi}{m}{it}
\DeclareMathSymbol{\varv}{\mathord}{matha}{118}
\usepackage{amsmath}
\usepackage{graphicx}
\newcommand{\tc}{t_{\mathrm{c}}}
\newcommand{\lamc}{\lambda_{\mathrm{c}}}

\newcommand{\Es}{E_{\mathrm{s}}}

\newcommand{\alabel}[1]{\stepcounter{equation}\tag{\theequation}\label{#1}}

\newcommand{\matr}[1]{\mathbf{#1}}
\newcommand{\matrH}{\matr{H}}
\newcommand{\matrQ}{\matr{Q}}

\newcommand{\vect}[1]{\textbfit{#1}}
\newcommand{\vb}{\vect{b}}
\newcommand{\vB}{\vect{B}}

\newcommand{\vX}{\vect{X}}
\newcommand{\vx}{\vect{x}}

\newcommand{\vY}{\vect{Y}}
\newcommand{\vy}{\vect{y}}

\newcommand{\vZ}{\vect{Z}}
\newcommand{\vD}{\vect{D}}
\newcommand{\vd}{\vect{d}}

\newcommand{\vu}{\vect{u}}
\newcommand{\vv}{\vect{v}}
\newcommand{\vz}{\vect{z}}
\newcommand{\vm}{\vect{m}}
\newcommand{\vms}{\vect{m}_\mathrm{s}}

\newcommand{\vmr}{\vect{m}_\mathrm{r}}

\newcommand{\set}[1]{\mathcal{#1}}
\newcommand{\setB}{\set{B}}

\newcommand{\setM}{\set{M}}
\newcommand{\setS}{\set{S}}

\newcommand{\setT}{\set{T}}
\newcommand{\setQ}{\set{Q}}
\newcommand{\setX}{\set{X}}
\newcommand{\setY}{\set{Y}}

\newcommand{\setC}{\set{C}}

\newcommand{\setBF}{\set{B}_{\mathpzc{f}}}
\newcommand{\setBE}{\set{B}_{\mathpzc{e}}}
\newcommand{\setBT}{\set{B}_{\mathpzc{c}}}

\newcommand{\setSF}{\set{S}_{\mathpzc{f}}}
\newcommand{\setSE}{\set{S}_{\mathpzc{e}}}
\newcommand{\setST}{\set{S}_{\mathpzc{c}}}

\newcommand{\BF}{{B}_{\mathpzc{f}}}
\newcommand{\BE}{{B}_{\mathpzc{e}}}
\newcommand{\BT}{{B}_{\mathpzc{c}}}

\newcommand{\vBF}{\vect{B}_{\mathpzc{f}}}
\newcommand{\vBE}{\vect{B}_{\mathpzc{e}}}

\newcommand{\SF}{{S}_{\mathpzc{f}}}
\newcommand{\SE}{{S}_{\mathpzc{e}}}
\newcommand{\ST}{{S}_{\mathpzc{c}}}

\newcommand{\vSF}{{S}'}

\newcommand{\bF}{{b}_{\mathpzc{f}}}
\newcommand{\bE}{{b}_{\mathpzc{e}}}
\newcommand{\bT}{{b}_{\mathpzc{c}}}

\newcommand{\vbF}{\vect{b}_{\mathpzc{f}}}
\newcommand{\vbE}{\vect{b}_{\mathpzc{e}}}

\newcommand{\sF}{{s}_{\mathpzc{f}}}
\newcommand{\sE}{{s}_{\mathpzc{e}}}
\newcommand{\sT}{{s}_{\mathpzc{c}}}	

\newcommand{\vsF}{{s}'}

\newcommand{\betaF}{{\beta}_{\mathpzc{f}}}
\newcommand{\betaE}{{\beta}_{\mathpzc{e}}}
\newcommand{\betaT}{{\beta}_{\mathpzc{c}}}

\newcommand{\LE}{{L}_{\mathpzc{e}}{(i,j)}}
\newcommand{\LT}{{L}_{\mathpzc{c}}{(i,j)}}

\newcommand{\vgB}{\vect{g}_{\mathrm{B}}}
\newcommand{\vgE}{\vect{g}_{\userE}}

\newcommand{\mB}{\nu^{\mathrm{B}}}
\newcommand{\mE}{\nu^{\userE}}

\newcommand{\sigmaB}{\sigma_{\mathrm{B}}}
\newcommand{\sigmaE}{\sigma_{\userE}}

\newcommand{\NBidx}[1]{W_{\mathrm{B}}(#1)}

\newcommand{\NEidx}[1]{W_{\userE}(#1)}

\newcommand{\Mr}{M_{\mathrm{r}}}
\newcommand{\Ms}{M_{\mathrm{s}}}

\newcommand{\mr}{m_{\mathrm{r}}}
\newcommand{\ms}{m_{\mathrm{s}}}

\newcommand{\kr}{{K}_{\mathrm{r}}}
\newcommand{\ks}{{K}_{\mathrm{s}}}

\newcommand{\Rr}{R_{\mathrm{r}}}
\newcommand{\Rs}{R_{\mathrm{s}}}

\newcommand{\Riud}{R_{\mathrm{iud}}}
\newcommand{\Rd}{R_{\mathrm{d}}}
\newcommand{\RSC}{R_{\mathrm{iud}}^{\mathrm{SC}}}

\newcommand{\Cs}{C_{\mathrm{s}}}
\newcommand{\Rcsc}{R_{\mathrm{csc}}}

\newcommand{\Rin}{R_{\mathrm{in}}}
\newcommand{\Rout}{R_{\mathrm{out}}}

\newcommand{\epsE}{\epsilon^{\mathrm{F}}}
\newcommand{\epsB}{\epsilon^{\mathrm{B}}}

\newcommand{\userB}{\mathrm{B}}
\newcommand{\userE}{\mathrm{E}}

\newcommand{\SNRdB}{\mathrm{SNR}_{\textrm{dB}}}
\newcommand{\SNRdBu}{\mathrm{SNR}_{\textrm{dB}}^u}
\newcommand{\SNRdBB}{\mathrm{SNR}_{\textrm{dB}}^{\userB}}
\newcommand{\SNRdBE}{\mathrm{SNR}_{\textrm{dB}}^{\userE}}

\newcommand{\defeq}{\triangleq}

\newcommand{\Z}{\mathbb{Z}}
\newcommand{\R}{\mathbb{R}}
\newcommand{\C}{\mathbb{C}}

\newcommand{\hell}{{\imath}}

\newcommand{\Prob}{\operatorname{Pr}}

\newcommand{\defend}{\mbox{}\hfill$\square$}

\newcommand{\exampleend}{\mbox{}\hfill$\square$}
\newcommand{\remarkend}{\mbox{}\hfill$\square$}

\newcounter{mytempeqcounter}

\newcommand{\indep}{\perp \!\!\!\! \perp}

\makeatletter
\DeclareRobustCommand\bfseriesitshape{%
	\not@math@alphabet\itshapebfseries\relax
	\fontseries\bfdefault
	\fontshape\itdefault
	\selectfont
}
\makeatother

\DeclareTextFontCommand{\textbfit}{\bfseriesitshape}

\begin{document}
	
	
	\title{Design and Analysis of a Concatenated Code for Intersymbol Interference Wiretap Channels}
	
	
	\author{
		Aria~Nouri\textsuperscript{\,\orcid{0000-0001-5548-184X}},~\IEEEmembership{Graduate Student Member,~IEEE,}
		Reza~Asvadi\textsuperscript{\,\orcid{0000-0001-9898-7744}},~\IEEEmembership{Senior Member,~IEEE}, and\\ Jun~Chen\textsuperscript{\,\orcid{0000-0002-8084-9332}},~\IEEEmembership{Senior Member,~IEEE}
		\thanks{Aria Nouri and Reza Asvadi are with the Department of Telecommunications, Faculty of Electrical Engineering, Shahid Beheshti University, Tehran 1983963113, Iran (e-mails:\ \href{mailto:ariya@ieee.org}{\txbl{ariya@ieee.org}}; \href{mailto:r_asvadi@sbu.ac.ir}{\txbl{r\_asvadi@sbu.ac.ir}}).}%
		\thanks{Jun Chen is with the Department of Electrical and Computer Engineering, McMaster University, Hamilton, ON L8S 4K1, Canada (e-mail:\ \href{mailto:chenjun@mcmaster.ca}{\txbl{chenjun@mcmaster.ca}}).}%
		\thanks{This paper was presented in part at the IEEE International Symposium on Information Theory (ISIT), Espoo, Finland, in June 2022~\cite{9834578}.}%
	}%
	\markboth{Accepted for publication in IEEE Transactions on Communications}{Nouri, Asvadi, and Chen}
	\maketitle
	
	
	\begin{abstract}
		We propose a two-stage concatenated coding scheme for reliable and secure communication over intersymbol interference wiretap channels. We first establish the secrecy capacity. Then, motivated by the theoretical codes that achieve the secrecy capacity, our scheme integrates low-density parity-check (LDPC) codes in the outer stage, forming a nested structure of wiretap codes, with trellis codes in the inner stage to improve achievable secure rates. The trellis code is specifically designed to transform the uniformly distributed codewords produced by the LDPC code stage into a Markov process, achieving tight lower bounds on the secrecy capacity. We further estimate the information leakage rate of the proposed scheme using an upper bound. To meet the weak secrecy criterion, we optimize degree distributions of the irregular LDPC codes at the outer stage, essentially driving the estimated upper bound on the information leakage rate to zero.
	\end{abstract}

	\IEEEpeerreviewmaketitle
	
	
	\section{Introduction}
	\label{sec::Intro}
	\subsection{Motivation}
	
	\IEEEPARstart{C}{lassical} cryptography addresses secrecy requirements by assuming that eavesdroppers have limited computational capabilities. This assumption prevents eavesdroppers from solving computationally complex problems, thereby safeguarding the cryptosystems from cryptanalysis. While breaking a robust cryptosystem within a reasonable time is currently infeasible, the evolution of prospective quantum computers poses a serious threat to the security of existing cryptographic protocols~\cite{Gidney_2021}. Moreover, the absence of security at the physical layer allows eavesdroppers to intercept and store the encrypted data, potentially breaking the encryption once the promised computational capabilities catch up. Alternatively, quantum cryptography aims to leverage the postulates of quantum mechanics to achieve loophole-free unconditional security~\cite{Pirandola:20}. Though this research area has brought revolutionary directions in near-future communication systems~\cite{Wehnereaam9288}, the proposed network protocol stack faces challenges in compatibility with existing infrastructures~\cite{9023997} and requires additional assumptions to achieve the claimed unconditional security~\cite{PhysRevLett.78.3414}.

	Information-theoretic physical-layer security~\cite{6772207,9380147} properly overcomes the aforementioned drawbacks: (i) by employing secrecy measures that are independent of the computational capabilities of eavesdroppers, security is assured even against adversaries with unlimited computational power; (ii) by operating at the physical layer of the existing infrastructure, it ensures that intercepting and analyzing the eavesdropper's observations do not reveal additional information about the secret data. Despite these advantages, a majority of information-theoretic proofs are developed under idealized channel conditions, making them largely infeasible to apply these results for securing non-ideal practical communication systems. This is because the criteria for ensuring information-theoretic security need to be fulfilled rigorously and measured explicitly~\cite{10336902}.\footnote{A limited amount of error caused by ignoring non-ideal properties of practical communication channels is tolerable when measuring the \emph{reliability}. However, neglecting these non-ideal properties when measuring the \emph{secrecy} can result in unintended leakage of private information, which is unacceptable.}
	
	Frequency selectivity is a non-ideal phenomenon that causes intersymbol interference (ISI), degrading the performance of wireless communication systems ranging from narrowband~\cite{8698792} to ultra-wideband~\cite{1569979}. In phased array antennas that are pivotal for millimeter-wave communication systems\textemdash due to their ability to support high-frequency bands, enabling the multi-gigabit-per-second data rates and low latency required by 5G and beyond networks\textemdash ISI arises from time delays across different antenna elements in both transmitted and received signals~\cite{9324767}. The existing equalization techniques are effective only when the delay exceeds a small fraction of the symbol interval, imposing several limitations on (i) the number of array elements, (ii) the scan angle,\footnote{The scan angle refers to the maximum angular range over which the antenna array can steer its main beam, typically measured from the array boresight (central axis) to the extreme steering angles in both directions.} and (iii) the symbol rate of the communication system~(see Appendix~\ref{apx:simsen}). Though multi-carrier methods can effectively mitigate ISI when the delay spread exceeds the symbol interval, several emerging applications continue to favor single-carrier approaches~\cite{10061469}. Battery-powered terminals with strict power or cost constraints require waveforms with low peak-to-average power ratio (PAPR), whereas multi-carrier methods typically exhibit high PAPR. Although various PAPR reduction techniques have been proposed~\cite{1421929,10530524}, they often involve trade-offs in rate and complexity. These limitations are particularly critical in uplink transmissions from low-cost user equipment~\cite{3GPP_TS_38_101_1} and downlink transmissions in battery-powered non-terrestrial/aerial platforms~\cite{3GPP_TR_38_821, 3gpp-ntn-overview}, which, despite being part of 5G ecosystems, operate under fundamentally different physical and energy constraints.\footnote{It is worth noting that the deployment of multi-carrier schemes does not conflict with our proposed inherently single-carrier approach. Different transmission scenarios within the same communication ecosystem impose distinct design requirements: single-carrier schemes are preferred when transmitter power efficiency is critical and processing complexity can be shifted to the receiver, whereas multi-carrier schemes are favored when receiver complexity must be minimized and transmitters operate under relaxed power constraints.} In addition, antenna size limitations in such aerial platforms generally result in phased arrays with wider main lobes, thereby increasing their susceptibility to eavesdropping and highlighting the need for enhanced physical-layer security.
	
	In this paper, we study the problem of code design for achieving reliable and information-theoretically secure communication over an ISI wiretap channel (ISI-WTC) model~\cite{10068266}. In this setup, the transmitter (Alice) is connected to a legitimate receiver (Bob, labeled by ``$\userB$'') and to an eavesdropper (Eve, labeled by ``$\userE$'') through two independent ISI channels. We consider a setting where the channel impulse responses (CIRs) of both Bob and Eve are fixed and known to all parties. Notably, the assumption that the eavesdropper's CIR is known at the transmitter corresponds to scenarios where Eve is another active user in the network, against whom Bob's secret message must be kept information-theoretically secure.
	
	\subsection{Background and Related Works}
	
	Designing practical methods for reliable transmission over point-to-point ISI channels is a well-established research topic in the literature of coding and signal processing societies~\cite{1054829,382016,4460060506,864167,911451,1003830,1207365,1431126,4282114,9279252}. For brevity, we only review a particular class of turbo-based error correction schemes~\cite{397441} related to this~work.
	
	The incorporation of maximum \emph{a posteriori} symbol detectors~\cite{10.1214/aoms/1177699147,welch,1055186} with linear block codes has led to the development of soft-estimation decoding architectures over ISI channels, commonly referred to as turbo-equalizers~\cite{4460060506,864167,911451,1003830,1207365,1431126}. Kavčić~\emph{et~al}.~\cite{1207365} extended the density evolution method~\cite{910578} to analyze the asymptotic performance of low-density parity-check (LDPC) codes within the turbo-equalizer framework. Varnica and Kavčić~\cite{1195499} further demonstrated that the \text{information} rate achieved by an independent and uniformly distributed (i.u.d.) input process can be closely approached using suitably optimized LDPC codes within the turbo-equalizer structure. More recently, it has been shown in~\cite{8840908} and \cite{9384303} that the complexity of  turbo-equalizers can be reduced by solving the optimization-based decoding problem~\cite{5571870} with the  approach presented  in~\cite{ADMM}. 
	
	Soriaga~\emph{et~al}.~\cite{4137898} addressed the design of a multi-level coding scheme paired with a multi-stage decoding algorithm for ISI channels. The decoder proposed in~\cite{4137898} operates without requiring interaction between all stages during each iteration and approaches the i.u.d.\ information rate as the number of decoding stages increases. To the best of our knowledge, the only coding scheme that surpasses the i.u.d.\ information rate and approaches the capacity of general point-to-point ISI channels is the one proposed in~\cite{1397934}, commonly referred to as matched information rate codes. Specifically, this code is a two-stage concatenated coding scheme comprising an inner-stage trellis code and an outer-stage LDPC code.
	
	Coding for secrecy over Gaussian wiretap channels~\cite{1055917} has received considerable attention over the past decade~\cite{5592833,5740591,6162586,6151133,7247218,9093879}. However, despite the practical importance of secure communication over channels with memory~\cite{7774989,8737786,10068266}, code design for such settings remains largely unresolved. To the best of our knowledge, this work is the first to explicitly address code design for Gaussian wiretap channels with ISI. By leveraging optimized Markov sources for ISI-WTCs~\cite{10068266} and extending the design principles established in~\cite{1397934}, we propose a practical concatenated coding scheme that satisfies both reliability and information-theoretic secrecy constraints.
	
	\subsection{Contributions and Organization}
	The main contributions of this work are outlined as follows.
	\begin{itemize}
		\item We derive the secrecy capacity of ISI-WTCs. Drawing on insights from the theoretical coding strategy that achieves the secrecy capacity, we develop a two-stage concatenated code. The inner coding stage is a trellis code that matches the transmitted codewords to a Markov process, achieving a numerical lower bound on the secrecy capacity, known as the constrained secrecy capacity~\cite{10068266}. We show that the designed inner-stage trellis code effectively shapes the spectrum of the transmitted codewords by concentrating the available power in frequency ranges where Bob's channel has a higher gain-to-noise power spectrum ratio than Eve's channel. Accordingly, positive secure rates remain achievable even in scenarios where Eve's channel has a higher point-to-point capacity than Bob's channel.
		\item We implement the nested structure of wiretap codes~\cite{4276938} using punctured LDPC codes at the outer stage. To evaluate the secrecy performance, we estimate a numerical upper bound on the rate of information leakage achieved by the proposed code over ISI-WTCs. Subsequently, for a fixed secure rate, we optimize degree distributions of the outer LDPC code stage to reduce the estimated upper bound essentially to zero, thereby satisfying the weak secrecy criterion~\cite{9380147}. Additionally, we extend the density evolution method from~\cite{1207365} to characterize the asymptotic performance of the proposed two-stage code and optimize degree distributions of the LDPC codes at the outer stage.
		\item In the asymptotic regime, the proposed coding scheme satisfies the reliability and weak secrecy criteria at secure rates within $0.42\!$~dB of the constrained secrecy capacity. In the finite blocklength regime, with a blocklength of $10^5$, the code achieves a bit-error rate of $10^{-5}$ at Bob's decoder and meets the weak secrecy criterion while remaining~within $3.6\!$~dB of the constrained secrecy capacity.
	\end{itemize}
	
	
	The remainder of this paper is organized as follows.
	Section~\ref{sec::PRELIM} introduces the channel model and key preliminary concepts.
	Section~\ref{sec::THEOBOUND} establishes the secrecy capacity of ISI-WTCs and elucidates the connection between the theoretical capacity-achieving code and the proposed concatenated code.
	Section~\ref{sec::FG} presents the normal factor graph representation of the code and generalizes the density evolution analysis.
	Section~\ref{sec::DESIGN} describes the code design procedure.
	Section~\ref{sec::SIMFL} reports the performance of the proposed code in the finite blocklength regime.
	Finally, Section~\ref{sec::CONC} concludes the paper.
	
	
	\subsection{Notation}\label{sec:intro:not}
	
	The sets of integers, reals, and complex numbers are denoted by $\Z$, $\R$, and $\C$, respectively. Other than that, sets are denoted by calligraphic letters, e.g., $\setS$. The Cartesian product of two sets $\setX$ and $\setY$ is written as $\setX \times \setY$, and the $n$-fold Cartesian product of $\setX$ with itself is written as $\setX^n$. If $\setX$ is a finite set, then its cardinality is denoted by $|\setX|$. Also, for a~finite vector space $\set{V}$, its dimension is represented by $\dim(\set{V})$. Moreover, we use the notation $\subsetneq$ to denote a \emph{finite} proper subset, e.g., $\setX \subsetneq \C$ indicates that $\setX$ is a \emph{finite} set whose elements lie in $\C$, but $\setX$ does not span the entire complex field.
	
	
	Random variables are denoted by upper-case italic letters, e.g., $X$, their realizations by the corresponding lower-case letters, e.g., $x$, and the set of possible values by the corresponding calligraphic letters, e.g., $\set{X}$. Random vectors are denoted by upper-case boldface italic letters, e.g., $\vX$, and their realizations by the corresponding lower-case letters, e.g., $\vx$. For a positive value of $n\in\Z$, a time-indexed ($t\in\Z$) vector of random variables is denoted by $\vX^{n}(t) \defeq \big(X{({n(t-1)+1})},\ldots,X({nt})\big) $,  and its realization is denoted by $\vx^{n}(t) \defeq \big(x{({n(t-1)+1})},\ldots,x({nt})\big) $. The time index $t$ is dropped from the vectors starting with $t=1$, e.g., $\vX^n$ is used instead of  $\vX^n(1)$. The convolution of two vectors $\vx$ and $\vy$ is denoted by $\vx\ast\vy$, and the $n$-time convolution of $\vx$ with itself is denoted by $\vx^{(\ast n)}$. Matrices and higher-dimensional arrays are denoted by upper-case boldface letters, e.g., $\matrH$.
	
	The probability of an event $\xi$ is represented by $\Prob(\xi)$. Additionally, $p_X(\,\cdot\,)$ denotes the probability mass function (PMF) of a discrete random variable $X$, and $p_{Y|X}(\,\cdot\, | x)$ represents the conditional PMF of a discrete random variable $Y$ given $X = x$. A similar notation is used for probability density functions (PDFs) when dealing with continuous random variables. The entropy of a random variable $X$ and the mutual information between two random variables $X$ and $Y$ are denoted by $H(X)$ and $I(X;Y)$, respectively. The Kullback–Leibler~(KL) divergence between two PMFs $p_X(\cdot)$ and $p_Y(\cdot)$ over the same finite alphabet $\setX$ is defined as
	$$
	\mathrm{D}_\mathrm{KL}(p_X||p_Y)\defeq\sum_{x\in\setX}p_X(x)\log \frac{p_X(x)}{p_Y(x)}.
	$$
	Additionally, the independence between two random variables $X$ and $Y$ is denoted by $X\indep Y$.
	
	For $x\in\R$, the expression $(x)^+$ stands for $\max\{x,0\}$.~Similarly, for a real-valued function $f(\cdot)$, the expression $f^+(\cdot)$ is used to represent $\big(f(\cdot)\big)^+$. Lastly, $\delta(\cdot)$ refers to the Kronecker delta function, and $\mathds{1}(\cdot)$ represents the indicator function.\\
	
	
	\section{Preliminaries and Definitions}\label{sec::PRELIM}
	
	In this section, we present finite-state machine models for intersymbol interference (ISI) channels and ISI wiretap channels (ISI-WTCs), along with their trellis representations. We then define trellis codes and introduce a joint finite-state model for trellis codes applied at the input of the ISI-WTCs.

	\subsection{Channel Model}
	\label{sec:pre:channel}
	
	
	An ISI channel considered in this work has a channel impulse response (CIR) specified by a complex vector $\vect{g} \defeq (g(0), \ldots, g(\nu)) \in \C^{\nu+1}$, where $\nu \in \Z$ denotes the memory length. The channel has an input process $\{ X(t) \}_{t \in \Z}$, a noiseless output process $\{ U(t) \}_{t \in \Z}$ given by $$U(t) \defeq \sum_{r=0}^{\nu} g(r)\cdot X(t-r),$$ and a noisy output process $\{ Y(t) \}_{t \in \Z}$ given by $$Y(t) \defeq U(t) + W(t),$$ where $W(t)\in\C$ is a circularly symmetric complex Gaussian random variable with mean zero, variance per dimension~$\sigma^2$, and $W(t)\indep X(t)$. For all $t\in\Z$, we have $X(t), U(t), Y(t)\in\C$. Any ISI channel is parameterized by $(\vect{g},\sigma^2)\in\C^{\nu+1}\times\R$.
	
	\begin{definition}[\textit{Finite-State Machine Channel (FSMC)}]\label{def:FSMC} A time-invariant FSMC consists of an input process $\{ X(t) \}_{t \in \Z}$, an output process $\{ Y(t) \}_{t \in \Z}$, and a state process $\{ \SF(t) \}_{t \in \Z}$, where $X(t) \in \setX$, $Y(t) \in \setY$, and $\SF(t) \in \setSF$, for all $t \in \Z$, and $\setSF$ and  $\setX$ are assumed to be finite. Let $\BF(t)\defeq\big(\SF(t-1),X(t),\SF(t)\big)$ referred to as a branch, and let $\setBF$ denote the set of all triples $\bF(t)=\big(\sF(t-1), x(t), \sF(t)\big)$ where $p_{\BF(t)}\big(\bF(t)\big)$ is allowed to be nonzero.\footnote{Because of the similarity of the notation used for various finite-state machines, we distinguish random variables and alphabet sets corresponding to states and branches of distinct finite-state machines using specific subscripts.} For any positive integer $N$, the joint PMF/PDF of $\vBF^N$ and $\vY^N$ conditioned on $\SF(0)=\sF(0)$ and $\vX^N=\vx^N$ is
		\begin{align*}
			&p_{\vBF^N, \vY^N | \SF(0), \vX^N}
			\big(\vbF^N, \vy^N | \sF(0), \vx^N\big)\\
			&\>\>= \prod_{t=1}^N
			p_{\BF(t), Y(t) | \SF(t-1), X(t)}
			\big(\bF(t), y(t) | \sF(t-1), x(t)\big), \alabel{equ:sdtr_decomp}
		\end{align*}
		where the factor on the right-hand side (RHS) is independent of $t$.
		\defend
	\end{definition}

	
	\begin{figure*}
		\centering
		\includegraphics[scale=.94]{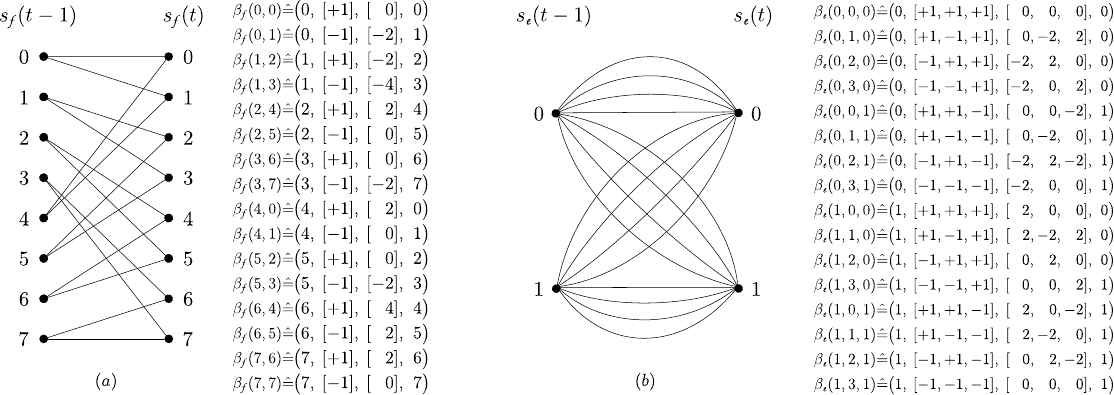}\vspace*{4pt}
		\caption{($a$) Trellis section of an FSMC, representing the EPR4 channel and the collection of branches $\betaF(i,j)=(i,x,u,j)$ for all $(i,j)\in\check{\setBF}$.
			($b$) Trellis section of a $3$-rd order E-FSMC, representing the DICODE channel and the collection of branches $\betaE(i,\ell,j)\defeq(i,\vx^3,\vu^3,j)$ for all $(i,\ell,j)\in\check{\setBE}$.\vspace{5pt}}\label{fig:CH_Trellis}
	\end{figure*}
	
	
	\begin{remark}\label{rem:FSMCrest}
		Any ISI channel with constant CIR $\vect{g}\in \C^{\nu+1}$, finite memory $\nu < \infty$, and finite input alphabet $\setX \subsetneq \C$ is a special case of an FSMC. By definition, the FSMC state must capture all information required to make the channel transition probability conditionally independent of the remaining past. In the present setting, this information comprises both the CIR $\vect{g}$ and the previous $\nu$ channel inputs. However, since the CIR $\vect{g}$ is assumed to be fixed and known over the entire transmission block, it is treated as part of the static channel model and we no longer include it explicitly in the state variable. Accordingly, the channel state at time $t$ is specified by $\sF(t)\defeq \vx^\nu\!\left(\tfrac{t}{\nu}\right)$ with finite state space $\setSF \defeq \setX^\nu$. Together with the fixed CIR $\vect{g}$, the state $\sF(t-1)$ fully determines the interference term as the weighted sum of past inputs $\sum_{r=1}^{\nu} g(r)X(t-r)$. This implies that the pair $\big( \sF(t-1), X(t) \big)$ contains all necessary dynamic information required to characterize the channel output distribution at time $t$. This also yields a one-to-one correspondence between $\big(\sF(t-1),x(t)\big)$ and $\bF(t)$ for all $t\in\Z$. Accordingly, the factor on the RHS of~\eqref{equ:sdtr_decomp} becomes
		\begin{align*}
			p_{\BF(t), Y(t) | \SF(t-1), X(t)}
			&\big(\bF(t), y(t) | \sF(t-1), x(t)\big)\\
			=p_{\BF(t) | \SF(t-1), X(t)}&\big(\bF(t) | \sF(t-1), x(t)\big)\\
			\cdot\,
			p_{Y(t) | \SF(t-1), X(t)}&\big(y(t) | \sF(t-1), x(t)\big),
		\end{align*}
		in which
			\begin{align*}
				p_{\BF(t) | \SF(t-1), X(t)}&\big(\bF(t) | \sF(t-1), x(t)\big)\\
				&\defeq \mathds{1}\Big(\bF(t)=\big(\sF(t-1), x(t),\sF(t)\big)\Big),
			\end{align*}
			and
		\begin{align*}
			p_{Y(t) | \SF(t-1), X(t)}&\big(y(t) | \sF(t-1), x(t)\big)\\
			&\defeq
			\frac{1}{{2 \pi \sigma^2}}
			\cdot
			\exp
			\left( 
			- \frac{| y(t) - u(t)|^2}{2\sigma^2}
			\right)\!, \alabel{equ:stntrsctr}
		\end{align*}
		where $u(t) \defeq \sum_{r=0}^{\nu} g(r)\cdot x(t-r)$.
		\remarkend
	\end{remark}

	\begin{definition}[\emph{Trellis Diagram}]\label{def:trellis}
		A trellis diagram is a two-dimensional directed acyclic graph that visualizes the time evolution of a finite-state machine. In this graph, the state alphabet of the finite-state machine is represented by ordered nodes grouped in consecutive vertical partitions, with each partition corresponding to a time index. State transitions are illustrated by directed branches connecting nodes in successive partitions. Since all finite-state machines considered in this work are time-invariant and the cardinality of the state alphabets is chosen such that the state processes have memory order one, a single trellis section comprising two successive partitions suffices to illustrate all valid state transitions. For notational convenience, realizations of the starting and ending states within a trellis section are denoted by $i$ and $j$, respectively, where $i,j\in\Z$.
		\defend
	\end{definition}

	According to $\sF(t)\defeq\vx^\nu(\frac{t}{\nu})$ and $|\setSF| \defeq |\setX|^\nu$ in Remark~\ref{rem:FSMCrest}, parallel branches do not appear in the trellis section of the FSMCs used to model the ISI channels. Therefore, each branch $\big(i,x(t),j\big)\in\setBF$ of the FSMC is uniquely specified by the pair $(i,j)\in\setSF^2$. Let $\check{\setBF}$ be a set of all pairs $(i,j)\in\setSF^2$ where $\big(i,x(t),j\big)\in\setBF$. Since one can define a bijective map between ${\setBF}$ and $\check{\setBF}$, we use $\check{\setBF}\subseteq\setSF^2$ as a simplified alternative instead of the branch alphabet $\setBF$. Accordingly, we associate a label $\betaF{(i,j)}\defeq(i,x,u,j)$ to $\bF(t)=\big(i,x(t),j\big)$, for all $\bF(t)\in\setBF$, where $x\in\setX$ is the input compatible with the realization of the branch $(i,j)\in\check{\setBF}$ (equivalently $(i,x,j)\in\setBF$) and $u\in\C$ is the associated noiseless output. We sketch the trellis section of an FSMC, representing an EPR4 channel\footnote{In this section, unnormalized magnetic recording channels are used as simplified models to introduce the general concept of trellises. While these channels may have limited relevance in wiretapping scenarios, the same trellises will be employed in later sections to model received signals from phased arrays (see Appendix~\ref{apx:simsen}).} with the unnormalized CIR $\vect{g}=(1,1,-1,-1)$ and the input alphabet $\setX\defeq\{-1,+1\}$ in~Fig.~\ref{fig:CH_Trellis}($a$).
	
	
	The following definition introduces an auxiliary model that extends Definition~\ref{def:FSMC} and provides the parameters required for designing the desired inner-stage trellis code in Section~\ref{sec:innertrellis}.
	
	\begin{definition}[\textit{Extended FSMC (E-FSMC)}]\label{def:EFSMC}
		An E-FSMC of order $n\in\Z$ is a time-invariant finite-state machine constructed by concatenating $n$ consecutive realizations of an FSMC, as in Definition~\ref{def:FSMC}. Indeed, the E-FSMC has an input process $\{ \vX^n(t) \}_{t \in \Z}$, an output process $\{ \vY^n(t) \}_{t \in \Z}$, and a state process $\{ \SE(t) \}_{t \in \Z}$, where $X(t) \in \setX$, $Y(t) \in \setY$, and $\SE(t) \in \setSE$, for all $t \in \Z$, and $\setSE$ and $\setX$ are assumed to be finite. 
			Let $\BE(t)\defeq\big(\SE(t-1),\vX^n(t),\SE(t)\big)$ and let $\setBE$ be a set of all triples  $\bE(t)=\big(\sE(t-1),\vx^n(t),\sE(t)\big)$ where $p_{\BE(t)}\big(\bE(t)\big)$ is allowed to be nonzero.
			Each realization of the branch $\bE(t)\in\setBE$ at the $n$-th order E-FSMC coincides with realizing $n$ legal consecutive branches\footnote{Legal consecutive realizations of branches are branch sequences where the ending state of each branch equals the starting state of the next. All other sequences are considered illegal.} at the corresponding FSMC, i.e., $\big(\sE(t-1),\bE(t)\big)\defeq\big(\sF\big(n(t-1)\big),\vbF^n(t)\big)$. Accordingly, the PMF of $\BE(t)$ conditioned on $\SE(t-1)=\sE(t-1)$ becomes
			\begin{align*}
				p_{\BE(t)| \SE(t-1)}&\big(\bE(t)|\sE(t-1)\big)\\
				&\defeq
				p_{\vBF^n(t)|\SF{(n(t-1))}}\Big(\vbF^n(t)|\sF\big(n(t-1)\big)\Big)\\
				&=\prod_{l=n(t-1)+1}^{nt} p_{\BF(l)| \SF(l-1)}\big(\bF(l)|\sF(l-1)\big).
			\end{align*}
			For any positive integer $N$, the joint PMF/PDF of $\vBE^N$ and $\vY^{nN}$ conditioned on $\SE(0)=\sE(0)$ and ${\vX}^{nN}={\vx}^{nN}$ is
			\begin{align*}
				&p_{\vBE^N, \vY^{nN} | \SE(0), \vX^{nN}}\big(\vbE^N, \vy^{nN} | \sE(0), \vx^{nN}\big)\\
				&= \prod_{t=1}^N
				p_{\BE(t), \vY^n(t) | \SE(t-1), \vX^n(t)}
				\big(\bE(t), \vy^n(t) | \sE(t-1), \vx^n(t)\big), \alabel{equ:efsmc_decomp}
			\end{align*}
			where the factor on the RHS is independent of $t$.
			\defend
	\end{definition}
	
	\begin{remark}\label{rem:E-FSMCISI}
		In the case of ISI channels, the factor on the RHS of \eqref{equ:efsmc_decomp} becomes
		\begin{align*}
			p_{\BE(t), \vY^n(t) | \SE(t-1), \vX^n(t)}
			&\big(\bE(t), \vy^n(t) | \sE(t-1), \vx^n(t)\big)\\
			=
			p_{\BE(t) | \SE(t-1), \vX^n(t)}
			&\big(\bE(t) | \sE(t-1), \vx^n(t)\big)\\
			\cdot\,
			p_{\vY^n(t) | \SE(t-1), \vX^n(t)}
			&\big(\vy^n(t) | \sE(t-1), \vx^n(t)\big),
		\end{align*}
		in which
			\begin{align*}
				p_{\BE(t) | \SE(t-1), \vX^n(t)}
				&\big(\bE(t) | \sE(t-1), \vx^n(t)\big)\\
				&\defeq \mathds{1}\Big(\bE(t)=\big(\sE(t-1), \vx^n(t),\sE(t)\big)\Big),
			\end{align*}
		and
		\begin{align*}
			&p_{\vY^n(t) | \SE(t-1), \vX^n(t)}\big(\vy^n(t) | \sE(t-1), \vx^n(t)\big)\\
			&\qquad\defeq\prod_{l=(t-1)n+1}^{nt} p_{Y(l) | \SF(l-1), X(l)}\big(y(l) | \sF(l-1), x(l)\big)\\
			&\qquad=
			\frac{1}{(2 \pi \sigma^2)^n}
			\cdot
			\exp
			\left(-\sum_{l=(t-1)n+1}^{nt}
			\frac{| y(l) - u(l)|^2}{2\sigma^2}
			\right), 
		\end{align*}
		where the last equality follows from (\ref{equ:stntrsctr}).
		\remarkend
	\end{remark}

	
	A set of all possible state transitions of the E-FSMC is also visualized by a single trellis section as in Definition~\ref{def:trellis}. Let $\LE\in\Z$ be the number of parallel branches between the pair of states $(i,j)\in\setSE^2$. Then, each branch $(i,\vx^n,j)\in\setBE$ of the E-FSMC is uniquely specified by the triple $(i,\ell,j)$, where $0\leq\ell\leq \LE-1$ identifies $\LE$ distinct branches between $(i,j)\in\setSE^2$. Let $\check{\setBE}$ be a set of triples $(i,\ell,j)$, specifying all valid branches $(i,\vx^n,j)$ in $\setBE$. Since one can define a bijective map between ${\setBE}$ and $\check{\setBE}$, we use $\check{\setBE}$ as a simplified alternative instead of the branch alphabet $\setBE$. Accordingly, we associate a label $\betaE(i,\ell,j)\defeq(i,\vx^n,\vu^n,j)$ to $\bE(t)=\big(i,\vx^n(t),j\big)$, for all $\bE(t)\in\setBE$, where $0\leq\ell\leq \LE-1$, $\vx^n\in\setX^n$ is the $n$-tuple input compatible with the realization of the branch $(i,\ell,j)\in\check{\setBE}$ (equivalently $(i,\vx^n,j)\in\setBE$), and $\vu^n\in\C^n$ is the associated noiseless $n$-tuple output. We sketch the trellis section of a third-order E-FSMC, representing a DICODE channel with the unnormalized CIR $\vect{g}=(1,-1)$ and the input alphabet $\setX\defeq\{-1,+1\}$ in~Fig.~\ref{fig:CH_Trellis}($b$). 	
	
	


	\begin{definition}[\textit{Intersymbol Interference Wiretap Channel (ISI-WTC)}]
		\label{def:prwt:channel:1}
		In an ISI-WTC, Alice transmits data symbols over Bob's and Eve's channels, characterized by the ISI channel parameter pairs $\bigl( \vgB, \sigmaB^2 \bigr)\in\C^{\mB+1}\times\R$ and $\bigl( \vgE, \sigmaE^2 \bigr)\in\C^{\mE+1}\times\R$, respectively. Specifically, Bob's channel has a noiseless output process $\{U(t)\}_{t\in\Z}$, a noise process $\{\NBidx{t}\}_{t\in\Z}$, and a noisy output process $\{Y(t)\}_{t\in\Z}$. Similarly, Eve's channel has a noiseless output process $\{V(t)\}_{t\in\Z}$, a noise process $\{\NEidx{t}\}_{t\in\Z}$, and a noisy output process $\{Z(t)\}_{t\in\Z}$. The random variables $\NBidx{t}$ and $\NEidx{t}$ are independent of each other and of the input process for all $t\in\Z$. (See Fig.~\ref{Fig:BLK_diag}.) \defend
	\end{definition}
	
	
	\begin{figure}
		\centering
		\includegraphics[scale=0.85]{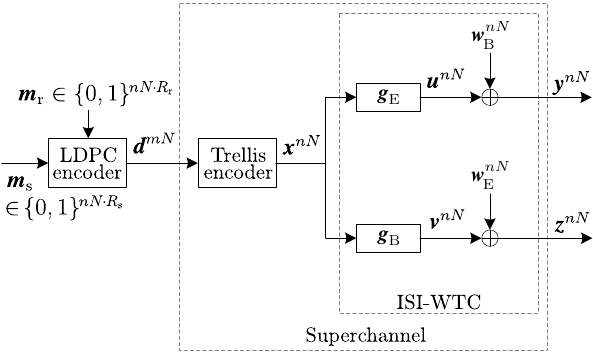}
		\caption{Block diagram of the proposed concatenated coding scheme, the ISI-WTC, and the corresponding superchannel. The trellis code has rate $m/n$.\\ \mbox{}}\label{Fig:BLK_diag}
	\end{figure}
	

	\begin{figure*}[t]
		\centering
		\includegraphics[scale=1.08]{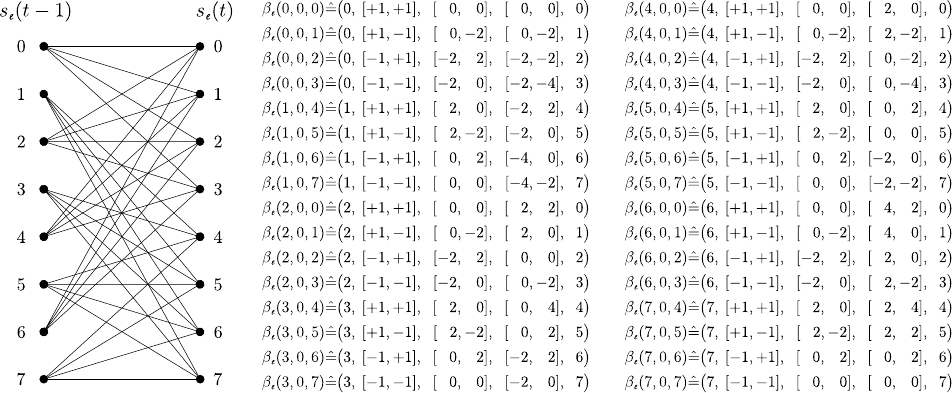}
		\caption{Trellis section of a $2$-nd order E-FSMC that represents an ISI-WTC, comprising a DICODE channel as Bob's channel and an EPR4 channel as Eve's channel. Each branch of the trellis section is labeled by $\betaE(i,\ell,j)\defeq(i,\vx^2,\vu^2,\vv^2,j)$ for all $(i,\ell,j)\in\check{\setBE}$.}\label{fig:PRWT_Trellis}
	\end{figure*}
	

	Notice that the FSMCs, the E-FSMCs, and their associated trellises are also utilized to represent ISI-WTCs. Since Bob's channel is assumed to be independent of Eve's channel, these generalizations are straightforward. We illustrate this using the following example. Consider an ISI-WTC with an input alphabet $\setX\defeq\{-1,+1\}$, where Bob's channel is a DICODE channel and Eve's channel is an EPR4 channel. The trellis section of the second-order E-FSMC, corresponding to this ISI-WTC is depicted in Fig.~\ref{fig:PRWT_Trellis}. Each branch of the trellis section is labeled by  $\betaE(i,\ell,j)\defeq(i,\vx^2,\vu^2,\vv^2,j)$, where $\vx^2\in\setX^2$ is the input compatible with the realization of the branch $(i,\ell,j)\in\check{\setBE}$ (equivalently $(i,\vx^2,j)\in\setBE$), and $\vu^2,\vv^2\in\R^2$ are, respectively, the associated noiseless outputs of Bob's and Eve's channels.
	
	
	
	\begin{figure}
		\centering
		\includegraphics[scale=1]{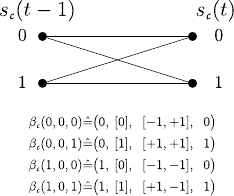}
		\caption{Wolf-Ungerboeck code of rate 1/2 and the collection of branches $\betaT(i,\ell,j)=(i,d,\vx^2,j)$ for all $(i,\ell,j)\in\check{{\setBT}}$.\\ \mbox{}}\label{fig:TRCD_Trellis}
	\end{figure}
	
	
	\subsection{Trellis Codes and Superchannel}
	\label{sec:pre:code}
	
	The trellis representation of error correction codes provides an effective framework for analyzing the statistical properties of codewords generated by a coding scheme without relying on its algebraic structure~\cite{TTC2015}. Without loss of generality, any~channel code that can be represented by a time-invariant finite-state machine is referred to as a trellis code in this paper.
	
	\begin{definition}[\emph{Trellis Code}]
		\label{def:trellis_code}
		A trellis code of rate $m/n$~(with $m\in\Z$, $n\in\Z$) is a time-invariant finite-state machine comprising an input process $\{ \vD^m(t) \}_{t \in \Z}$, an output process $\{ \vX^n(t) \}_{t \in \Z}$, and a state process $\{ \ST(t) \}_{t \in \Z}$, where $D(t)\in\{0,1\}$, $X(t)\in\setX$, $\ST(t)\in{\setST}$, for all $t\in\Z$, ${\setST}$ and $\setX$ are assumed to be finite, and $m\leq n\cdot\log|\setX|$. Let $\BT(t)\defeq\big(\ST(t-1),\vD^m(t),\ST(t)\big)$ and let ${\setBT}$ be a set of all triples $\bT(t)=\big(\sT(t-1),\vd^m(t),\sT(t)\big)$ where $p_{\BT(t)}\big(\bT(t)\big)$ is allowed to be nonzero.  For a trellis code of rate $m/n$, each realization of the branch $\bT(t)\in\setBT$ maps an $m$-bit input message $\vd^m(t)\in\{0,1\}^m$ to an $n$-symbol output codeword $\vx^n(t)\in\setX^n$, at time $t\in\Z$. Thus, there are totally $2^m$ branches $\bT(t)\in\setBT$ emanating from each state $\sT(t-1)\in\setST$ of the trellis code (one branch for each binary input of length $m$).
		
		\defend
	\end{definition}

	Due to the time-invariance assumption, a single trellis section is sufficient to demonstrate all valid branches between the successive states of the trellis~code. Let $\LT\in\Z$ be the number of parallel branches between the pair of states $(i,j)\in\setST^2$. Then, each branch $(i,\vd^m(t),j)\in\setBT$ of the trellis code is uniquely specified by a triple $(i,\ell,j)$, where $0\leq\ell\leq \LT-1$ identifies $\LT$ distinct branches between $(i,j)\in\setST^2$. Let $\check{\setBT}$ be a set of triples $(i,\ell,j)$, specifying all valid branches $(i,\vd^m(t),j)$ in $\setBT$. Since one can define a bijective map between $\setBT$ and $\check{\setBT}$, we use $\check{\setBT}$ as a simplified alternative instead of the branch alphabet $\setBT$. Accordingly, we associate a label $\betaT(i,\ell,j)\defeq(i,\vd^m,\vx^n,j)$ to $\bT(t)=\big(i,\vd^m(t),j\big)$, for all $\bT(t)\in\setBT$, where $0\leq\ell\leq \LT-1$, $\vd^m\in\{0,1\}^m$ is the $m$-bit input message compatible with the realization of the branch $(i,\ell,j)\in\check{\setBT}$ (equivalently $(i,\vd^m,j)\in\setBT$), and $\vx^n\in\setX^n$ is the associated $n$-tuple output codeword.
	The trellis representation of the Wolf-Ungerboeck code~\cite{1096620} of rate 1/2 is depicted in Fig.~\ref{fig:TRCD_Trellis}.
	
	
	\begin{figure}
		\centering
		\includegraphics[scale=0.90]{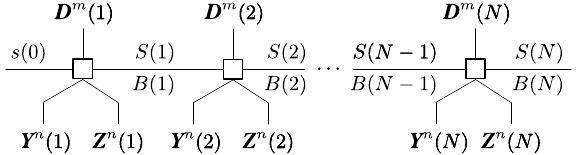}
		\caption{Normal factor graph representation of the superchannel.\\ \mbox{}}\label{Fig:SupCha_Factor}
	\end{figure}
	
	
	\begin{definition}[\emph{Superchannel}]\label{def:superchannel}
		A superchannel is a time-invariant finite-state machine, modeling the concatenation of a trellis code of rate $m/n$ (Definition~\ref{def:trellis_code}), and an $n$-th order E-FSMC (Definition~\ref{def:EFSMC}) representing an ISI-WTC (see Fig.~\ref{Fig:BLK_diag}). The superchannel comprises an input process $\{ \vD^m(t) \}_{t \in \Z}$ corresponding to the input messages of the trellis code, an output process $\{ \vY^n(t) \}_{t \in \Z}$ corresponding to Bob's observations, an output process $\{ \vZ^n(t) \}_{t \in \Z}$ corresponding to Eve's observations, and a state process $\{ S(t) \}_{t \in \Z}$. Indeed, $D(t)\in\{0,1\}$, $Y(t), Z(t) \in \C$, and $S(t) \in \setS$, for all $t \in \Z$, and ${\setS}$ is finite. Let $B(t)\defeq\big(S(t-1),\vD^m(t),S(t)\big)$ and let $\setB$ be a set of all triples $b(t)=\big(s(t-1),\vd^m(t),s(t)\big)$ where $p_{B(t)}\big(b(t)\big)$ is allowed to be nonzero. Each realization of the branch $b(t)\in\setB$ maps an $m$-bit input message $\vd^m(t)\in\{0,1\}^m$ to two $n$-tuple outputs, namely, Bob's observation $\vy^n(t)\in\C^n$ and Eve's observation $\vz^n(t)\in\C^n$. For any positive integer $N$, the joint PMF/PDF of ${\vB}^N\!\!$, $\vY^{nN}\!\!$, $\vZ^{nN}\!$ conditioned on $S(0)=s(0)$ and ${\vD}^{mN}={\vd}^{mN}$~is
		\begin{align*}
				&p_{{\vB}^N\!,\,\vY^{nN},\, \vZ^{nN} | S(0),\,\vD^{mN}}
				\big({\vb}^N, \vy^{nN}, \vz^{nN} | s(0),\,\vd^{mN}\big)\\
				&\qquad= \prod_{t=1}^N
				p_{B(t) | S(t-1),\,\vD^m(t)}
				\big(b(t) | s(t-1), \vd^m(t)\big)\\[-10pt]
				&\qquad\qquad\quad\cdot
				p_{\vY^n(t) | S(t-1),\,\vD^m(t)}
				\big(\vy^n(t) | s(t-1), \vd^m(t)\big)\\
				&\qquad\qquad\quad\cdot
				p_{\vZ^n(t) | S(t-1),\,\vD^m(t)}
				\big(\vz^n(t) | s(t-1), \vd^m(t)\big).
		\end{align*}
		This factorization is depicted as a normal factor graph~\cite{910573} in Fig.~\ref{Fig:SupCha_Factor}. \defend
	\end{definition}
	

	\begin{figure*}
		\centering
		\includegraphics[scale=.96]{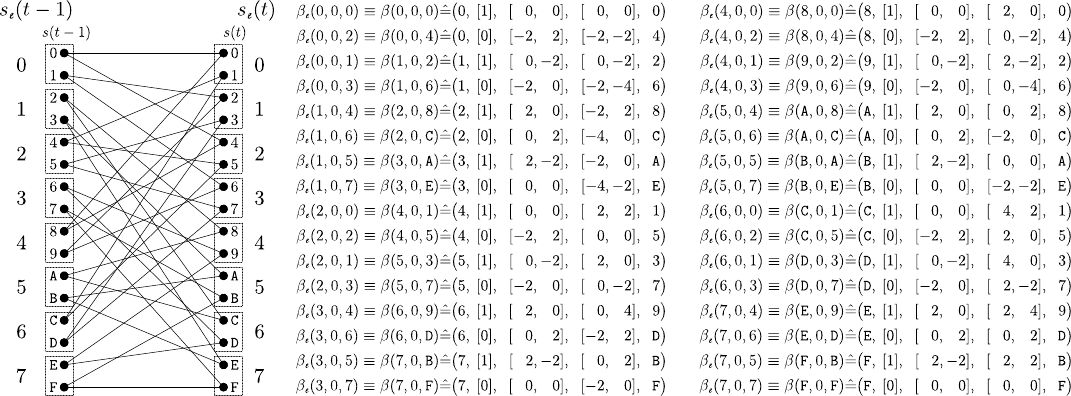}
		\caption{Trellis section of a superchannel, resulting from the concatenation of the trellis code depicted in Fig.~\ref{fig:TRCD_Trellis} and the $2$-nd order E-FSMC depicted in Fig.~\ref{fig:PRWT_Trellis}. Each branch of the trellis section is labeled by $\beta(i,\ell,j)\defeq(i,d,\vu^2,\vv^2,j)$ for all $(i,\ell,j)\in\check{\setB}$.}\label{fig:SupCha_Trellis}
	\end{figure*}
	
	
	Due to the time-invariance assumption in both the trellis code and the E-FSMC that represents the ISI-WTC, the finite-state machine modeling the superchannel can also be represented by a single trellis section. Let $L(i,j)\in\Z$ be the number of parallel branches between the pair of states $(i,j)\in\setS^2$. Then, each branch $\big(i,\vd^m(t),j\big)\in\setB$ of the superchannel is uniquely specified by a triple $(i,\ell,j)$, where $0\leq\ell\leq L(i,j)-1$ identifies $L(i,j)$ distinct branches between $(i,j)\in\setS^2$. Let $\check{\setB}$ be a set of triples $(i,\ell,j)$, specifying all valid branches $\big(i,\vd^m(t),j\big)$ in $\setB$. Since one can define a bijective map between ${\set{B}}$ and $\check{\set{B}}$, we use $\check{\set{B}}$ as a simplified alternative instead of the branch alphabet $\set{B}$. Accordingly, we associate a label $\beta(i,\ell,j)\defeq(i,\vd^m,\vu^n,\vv^n,j)$ to $b(t)=\big(i,\vd^m(t),j\big)$, for all ${b}(t)\in{\set{B}}$, where $0\leq\ell\leq L(i,j)-1$, $\vd^m\in\{0,1\}^m$ is the $m$-bit input message compatible with the realization of the branch $(i,\ell,j)\in\check{{\set{B}}}$ (equivalently $(i,\vd^m,j)\in{\setB}$), and $\vu^n, \vv^n\in\C^n$ are, respectively, the associated Bob's and Eve's noiseless $n$-tuple observations.
	
	In the following, we define additional notations that will be employed in~Section~\ref{sec:innertrellis} for optimizing the trellis code within the structure of the superchannel.
	
	\begin{definition}[\emph{Notation}]\label{def:superchannel:not}
		A set of all states of the superchannel compatible with the state $i\in\setSE$ in the constituent E-FSMC is denoted by $\setT_i\subseteq\setS$. Additionally, a set of all branches of the superchannel, starting from a state $s(t-1)\in\setT_i$, ending up with a state $s(t)\in\setT_j$, and having noiseless outputs $\vu^n\in\C^n$ and $\vv^n\in\C^n$ similar to the outputs on the branch $\betaE(i,\ell,j)\defeq(i,\vx^n,\vu^n,\vv^n,j)$ with $(i,\ell,j)\in\check{\setBE}$, is denoted by $\setT_{ij}^{(\ell)}\subseteq\setB$. It is important to note that each branch emanating from a particular state of the superchannel (e.g., $s(t-1)\in\setT_i$) is from a distinct $\setT_{ij}^{(\ell)}$ (for $\ell,j:(i,\ell,j)\in\check{\setBE}$). In Fig.~\ref{fig:SupCha_Trellis}, we sketch an exemplary trellis section of a superchannel,  obtained by the concatenation of~the Wolf-Ungerboeck code of rate $1/2$ depicted in Fig.~\ref{fig:TRCD_Trellis} and  the second-order E-FSMC depicted in Fig.~\ref{fig:PRWT_Trellis}. For instance, we have $\{\mathtt{A},\mathtt{B}\}\in\setT_5$ (with $5\in\setSE$) and $(\mathtt{A},0,8)\in\setT_{54}^{(0)}$ (with $(5,0,4)\in\setBE$). Also, we have $|\setT_i|=2$ for all $i\in\setSE$ and $|\setT_{ij}^{(\ell)}|=1$ for all $(i,\ell,j)\in\check{\setBE}$.~\defend
	\end{definition}
	
	
	\section{Theoretical Bounds}\label{sec::THEOBOUND}
	
	We first review a numerical lower bound on the secrecy capacity of ISI-WTCs from~\cite{10068266} and discuss the limitations of low-density parity-check (LDPC) codes in achieving this bound. We then derive the secrecy capacity of ISI-WTCs and use the theoretical capacity-achieving code to guide the design of a practical coding scheme, attaining the lower bound in~\cite{10068266}.
	
	Let $\Ms$ be a random variable representing a uniformly chosen secret message $\ms$ from an alphabet $\setM_\mathrm{s}$. Alice aims to ensure reliable and secure transmission by employing a binary code of rate $\Rs\defeq\frac{1}{N}\log|\setM_\mathrm{s}|$, mapping the secret message $\ms\in\setM_\mathrm{s}$ to a binary codeword $\vx^N\in\{0,1\}^N$. The reliability of Bob's decoder $\psi_N:\C^N\rightarrow\setM_\mathrm{s}$ is measured by the probability of a block error
	$
	\epsB_N\defeq\Prob\big(
	{M}_\mathrm{s}\neq\psi_N(\vY^N)\big),
	$
	and the secrecy performance of the code is measured by the rate of information leakage
	$
	\varepsilon_N\defeq {I(\Ms;\vZ^N)}/{N}.
	$
	With these measures, the reliability and secrecy criteria for the achievability of a secure rate $\Rs$ are, respectively, set to be
	\begin{align*}
		\lim_{N\to\infty}\epsB_N= 0,\qquad
		\lim_{N\to\infty}\varepsilon_N= 0.\alabel{equ:crit:sec}
	\end{align*}
	
	\subsection{Constrained Secrecy Capacity}
	
	Consider an ergodic Markov source that assigns a time-invariant transition probability $p_{ij}^{(\ell)}$ to each branch $\betaE(i,\ell,j)$ of an E-FSMC, i.e., $p_{ij}^{(\ell)}\defeq p_{\BE(t)|\SE(t-1)}\big(\betaE(i,\ell,j)|i\big),$ for all $(i,\ell,j)\in\check{\setBE}.$ Since the Markov source and the E-FSMC are assumed to be time-invariant, a steady-state probability $\mu_i\defeq p_{\SE(t)}(i)$ is assigned to each state $i\in\setSE$ of the E-FSMC. Let $Q_{ij}^{(\ell)}\defeq\mu_i\cdot p_{ij}^{(\ell)}$ represent the probability of realizing the branch $\betaE(i,\ell,j)$, for all $(i,\ell,j)\in\check{\setBE}$. Additionally, let $\mathbf{Q}\defeq\{Q_{ij}^{(\ell)}\}_{(i,\ell,j) \in \check{\setBE}}$ and let~$\setQ(\check{\setBE})$ be a polytope that includes all valid $\mathbf{Q}$ compatible with the branch connections specified by~the alphabet set $\check{\setBE}$. (See~\cite{10068266} and~\cite{4494705} for further details).
	
	Via a straightforward generalization of~\cite{10068266}, the constrained secrecy capacity of an $n$-th order E-FSMC modeling an ISI-WTC is the maximum achievable secure rate under the reliability and secrecy criteria in~\eqref{equ:crit:sec},\footnote{The secrecy criterion considered in~\cite[Prop.~1]{10068266} is stronger than the weak secrecy criterion presented in~\eqref{equ:crit:sec}. Consequently, all achievable secure rates in~\cite[Prop.~1]{10068266} are achievable under the weak secrecy criterion in~\eqref{equ:crit:sec} as well.} obtained by optimizing over finite-order Markov sources as characterized in~\cite[Prop.~1]{10068266}
	\begin{align*}
		\Rcsc^{(n)}
		\defeq
		\sup_{\mathbf{Q}\in\setQ(\check{\setBE})}
		\lim\limits_{N\to\infty}
		\frac{1}{nN}
		&\Big(
		I\big(\vBE^N; \vY^{nN} | \SE(0)\big)\\[-6pt]
		&-
		I\big(\vBE^N; \vZ^{nN} | \SE(0)\big)
		\Big)^+\!\!, \alabel{equ:csc}
	\end{align*}
	where the relevant maximization problem is efficiently solved using~\cite[Alg.~1]{10068266}. Note that since the finite-state machines considered in this work are indecomposable~\cite{Gallager:1968:ITR:578869} (i.e., the effect of an initial state vanishes over time), the mutual information rates are well-defined even when the initial state is unknown. Accordingly, for notational convenience, conditioning on the initial state will be omitted in the mathematical statements of the following sections.
	
	\subsection{Limitation of Single-Stage LDPC Codes}\label{subsec:codingforsec}
	
	LDPC codes are promising candidates for secure communication over wiretap channels due to their flexible coset-based nested structure~\cite{4276938}, powerful design tools such as density evolution~\cite{910578}, and reasonable encoding and decoding complexity~\cite{910579}. As a result, we incorporate them into our design. Nevertheless, achieving the constrained secrecy capacity requires the transmitted codewords to follow the channel input process that attains this bound. Conventional LDPC codes, however, lack the capability to precisely shape the distribution of the generated codewords to match this optimal profile.
	
	\begin{remark}\label{rem:ultbnd}
		Let $R_{\mathrm{iud}}^{(n)}$ denote the secure rate achieved by an i.u.d.\ input process $\{X(t)\}_{t\in\Z}$ over the considered ISI-WTC,
		\begin{equation*}
			R_{\mathrm{iud}}^{(n)} \defeq \lim_{N\to\infty}
			\frac{1}{N}
			\left(I(\vX^N;\vY^N)
			-
			I(\vX^N;\vZ^N)\right)^+\!\!.
		\end{equation*}
		According to~\cite[Prop.~1]{10068266}, $R_{\mathrm{iud}}^{(n)}$ is achievable. This secure rate provides a natural benchmark for LDPC codes without explicit probabilistic shaping. We thus define $\Rs^\ast = R_{\mathrm{iud}}^{(n)}$ as the target rate for single-stage LDPC codes over ISI-WTCs. Since our two-stage scheme can be viewed as a single-stage LDPC code over a superchannel, this remark directly provides a benchmark for the proposed code in subsequent sections. This secure rate can be numerically evaluated using variants of the sum-product/BCJR algorithm proposed in~\cite{1661831}. \remarkend
	\end{remark}
	
	Appendix~\ref{apx:simsen} considers an eavesdropping scenario in a downlink transmission from a typical 256-element ($16\times16$) phased array, modeled as an ISI-WTC. Example~\ref{ex:casestudy} applies \cite[Alg.~1]{10068266} to calculate the constrained secrecy capacity $\Rcsc^{(n)}$ and the secure rate achieved by an i.u.d.\ input process $\Riud^{(n)}$ of the E-FSMC representing this ISI-WTC. This channel model also serves as the basis for our simulations in subsequent sections.
	\begin{example}\label{ex:casestudy}
		Consider an ISI-WTC with input alphabet $\setX\defeq\{-1,+1\}$ and
			\begin{align*}
				&\vgB \!\!=\! (0.2360 \!-\! 0.7195i,  0.6230 \!+\! 0.1061i, -0.1314 \!-\! 0.0994i),\\
				&\vgE \!\!=\! (0.5211 \!-\! 0.4792i, -0.5791 \!+\! 0.4043i).
			\end{align*}
		Let $\SNRdBB=\SNRdBE=-4.0$, where $\SNRdBu \defeq 10\log_{10}\sfrac{1}{\sigma_u^2}$ for each user $u\in\{\userB,\userE\}$.\footnote{In the considered setup, (i) Bob operates in a low SNR regime, making it challenging to satisfy the reliability criterion (while such conditions are relevant for energy-efficient 5G NR base station deployments, particularly in mmWave Frequency Range 2 (FR2) scenarios involving non-terrestrial networks, where power constraints are critical~\cite{3GPP_TR_38_821}); (ii) Bob's channel has a lower point-to-point capacity than Eve's channel (see Appendix~\ref{apx:simsen}), making it challenging to satisfy the secrecy criterion; and (iii) binary phase-shift keying (BPSK) is used for inputs to simplify the presentation of the trellis code. It is important to note that neither of these considerations is strictly necessary.} The constrained secrecy capacity of the fourth-order E-FSMC representing this ISI-WTC is $\Rcsc^{(n)}\big|_{n=4}=0.121$~$(\sfrac{\text{bits}}{\text{channel use}})$ and the secure rate achieved by an i.u.d.\ input process is $\Riud^{(n)}\big|_{n=4}=0.013$~$(\sfrac{\text{bits}}{\text{channel use}})$; see the corresponding curves in Fig.~\ref{Fig:secrecy_gain} for a wide range of $-6.0\leq\SNRdBB\leq 0$.
		\exampleend
	\end{example}
	
	
	\begin{figure*}
		\centering
		\includegraphics[scale=1.15]{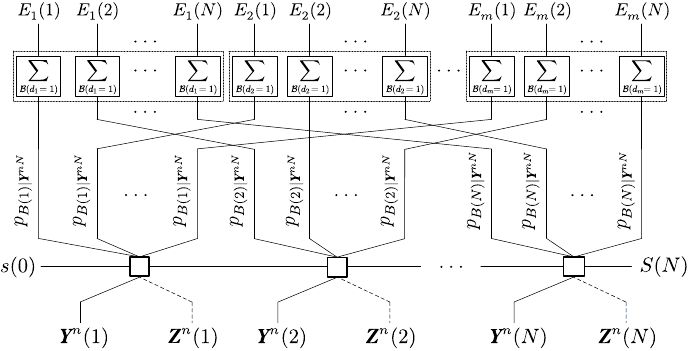}
		\caption{Soft-output messages from FNs to VNs at Bob's decoder. (See equations \eqref{equ:extrinsix} and \eqref{equ:extrinsix:el}.)}\label{Fig:FIG_Inner_Stage_EXT}
	\end{figure*}
	

	\subsection{Secrecy Capacity of ISI-WTCs}

	In the following, we describe the connection between the theoretical coding strategy that achieves the secrecy capacity and the proposed practical code that surpasses $\Riud^{(n)}$ and achieves $\Rcsc^{(n)}$. We begin with formulating the secrecy capacity.
	\begin{proposition}\label{prop:secap}
		The secrecy capacity of the ISI-WTC under the reliability criterion in~\eqref{equ:crit:sec} and the strong secrecy criterion\footnote{The strong secrecy constraint is applied solely to characterize the secrecy capacity of ISI-WTCs. Throughout the remainder of the paper, weak secrecy is used as the main metric for evaluating the performance of the designed wiretap codes. Note that the conversion from weak to strong secrecy can be realized as a modular post-processing step beyond the channel-coding layer, without any asymptotic rate loss, via privacy amplification as in~\cite[Thm.~2]{Maurer2000}.} $nN\cdot\varepsilon_{nN}\to0$ is given by
		\begin{equation*}
			\Cs
			=
			\lim_{n\to\infty}\frac{1}{n} \sup_{p_{D,\vX^n}} \Big(I(D;\vY^n) - I(D;\vZ^n)\Big)^+\!\!,
		\end{equation*}
		where the supremum is taken over all joint distributions induced by an auxiliary random variable $D\in\set{D}$ and the conditional input distribution $p_{\vX^n|D}$, complying with the Markov chains $D\to\vX^n\to\vY^n$ and $D\to\vX^n\to\vZ^n$.
	\end{proposition}
	\begin{IEEEproof}
		See Appendix~\ref{apx:secap}.
	\end{IEEEproof}
	As shown in Appendix~\ref{apx:secap}, the encoder of the theoretical capacity-achieving scheme consists of the concatenation of (i)~a nested-code encoder for wiretap channels $f_{N}\!:\setM_\mathrm{s}\to\set{D}^{N}$ and (ii)~a component-wise stochastic encoder $h_n\!:\set{D}\to\setX^n$, operating according to the stochastic map $p_{\vX^n\mid D}$. In the practical scheme, we implement the nested wiretap code structure using \emph{outer}-stage LDPC codes, while realizing the blockwise probabilistic mapping $p_{\vX^n\mid D}$ via an \emph{inner}-stage trellis code.
	
	In the proof of Proposition~\ref{prop:secap}, the auxiliary random variable $D\in\set{D}$ is defined as a scalar on a per-block basis, where the theoretical encoder $h_n$ maps each realization $d(t)$ to a length-$n$ channel input vector $\vX^n(t)$ according to the \emph{stochastic} map $p_{\vX^n\mid D}$. In contrast, in the proposed practical implementation, this mapping is realized by a \emph{deterministic} component-wise trellis encoder that maps length-$m$ vectors $\{\vd^m(t)\}_{t=1}^N$ to the channel input vectors $\{\vx^n(t)\}_{t=1}^N$. Under this framework, the inner trellis code transforms the i.u.d.\ sequence $\{\vD^m(t)\}_{t\in\Z}$ produced by the \emph{outer}-stage LDPC codes into a channel-input Markov process $\{\vX^n(t)\}_{t\in\Z}$ possessing the optimal distribution $\matr{Q}\in\setQ(\check{\setBE})$ that maximizes~\eqref{equ:csc}, thereby achieving the constrained secrecy capacity of the underlying ISI-WTC. Within this architecture (Fig.~\ref{Fig:BLK_diag}), the i.u.d.\ secure rate of the resulting superchannel equals the constrained secrecy capacity. Accordingly, since Remark~\ref{rem:ultbnd} implies that the i.u.d.\ secure rate of this induced superchannel is achievable by the \emph{outer}-stage LDPC code, the constrained secrecy capacity of the ISI-WTC becomes achievable by the proposed two-stage coding scheme.

	\section{Normal Factor Graph of the Two-Stage Code}\label{sec::FG}
	
	In this section, we concatenate the normal factor graph of the superchannel (Fig.~\ref{Fig:SupCha_Factor}) with an outer LDPC code stage. We then describe the employed punctured encoding procedure and outline the modified density evolution analysis~\cite{910578,1207365} for the resulting two-stage concatenated code.
	
	Let FN, VN, and CN be the abbreviations standing for the factor nodes of a superchannel (see Fig.~\ref{Fig:SupCha_Factor}), the variable nodes, and the check nodes of an LDPC code, respectively. With a slight abuse of notation, we denote the $l$-th element of the time-indexed vector $\vD^m(t)$ by $D_l(t)\defeq D\big((t-1)m+l\big)$ when the vector length is clear from the context. Given the output realizations $\vy^{nN}$, the soft-output information emanating from the FNs of the superchannel is the logarithmic \emph{a posterior} probability (APP) ratio of the symbol $d_l(t)\in\{0,1\}$, defined as
	\begin{equation}\label{equ:extrinsix}
		{E}_l(t)\defeq\log\left(\frac{p_{D_l(t)|\vY^{nN}}(1|\vy^{nN})}{1-p_{D_l(t)|\vY^{nN}}(1|\vy^{nN})}\right),
	\end{equation}
	for all $1\leq l\leq m$ and $1\leq t\leq N$. We also define the set $\set{B}(d_{l}=1)\defeq\{(i,\vd^m,j)\in\set{B}:d(l)=1\}$. Then, we have
	\begin{equation}\label{equ:extrinsix:el}
		p_{D_l(t)|\vY^{nN}}(1|\vy^{nN})=\sum_{(i,\vd^m,j)\in\set{B}(d_{l}=1)}p_{B(t)|\vY^{nN}}\big(i,\vd^m,j|\vy^{nN}\big),
	\end{equation}
	where $p_{B(t)|\vY^{nN}}\big(i,\vd^m,j|\vy^{nN}\big)$ is calculated by running the windowed version of the BCJR algorithm~\cite{490940}, given the observations from the channel's output and the extrinsic information available from the outer code stage (Fig.~\ref{Fig:FIG_Inner_Stage_EXT}). Since for $N\to\infty$ the soft-output messages $\big({E}_l(t)\big)_{t=1}^N$ have different statistics for each $1\leq l\leq m$, we employ $m$ separated LDPC subcodes as the outer code stage~\cite{1397934}. The normal factor graph of the two-stage code is depicted in Fig.~\ref{Fig:ConcGraph}.
	
	
	\begin{figure*}
		\centering
		\includegraphics[scale=0.76]{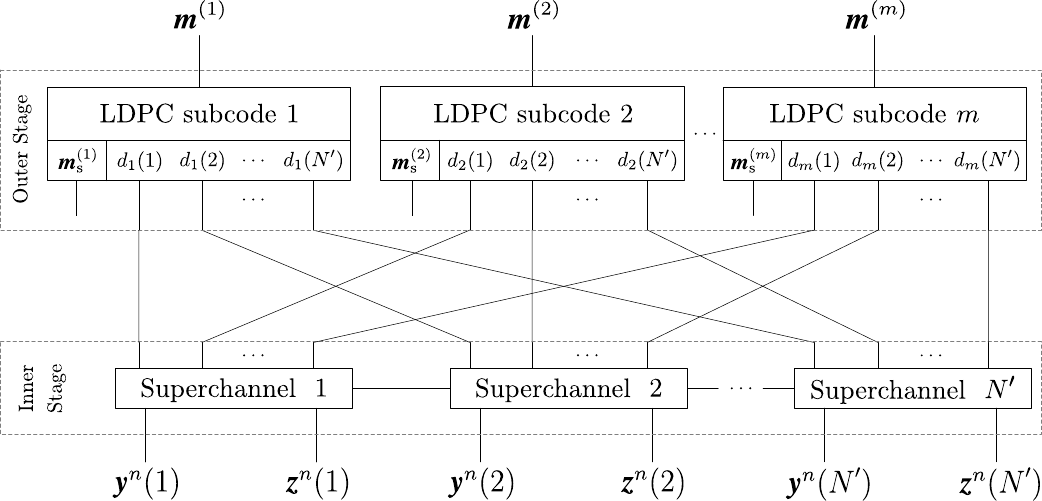}
		\caption{Normal factor graph of the concatenated code, comprising a superchannel ($\Rin=\sfrac{m}{n}$) and $m$ LDPC subcodes as the outer stage.}\label{Fig:ConcGraph}
	\end{figure*}
	
	
	\subsection{Nested Structure and Punctured Coding}
	\label{subsec:Enc}
	
	
	According to the equation $$N\cdot\varepsilon_N=I(\vX^N;\vZ^N)+H(\vX^N|\Ms,\vZ^N)-H(\vX^N|\Ms),$$ the employed coding scheme should avoid a bijection between $\Ms$ and $\vX^N$ in order to leverage the term $H(\vX^N|\Ms)>0$ for the cancellation of $I(\vX^N;\vZ^N)+H(\vX^N|\Ms,\vZ^N)$. To accomplish this, a nested code is adopted for wiretap channels, which comprises a mother code $\setC\subset\{0,1\}^N$ with a design rate $\Rd$ (i.e., $\dim(\setC)=N\cdot\Rd$) and a subspace $\setC'\subset\setC$ such that $\dim(\setC)-\dim(\setC')=N\cdot\Rs$. The subspace $\setC'$ partitions the mother code $\setC$ into $2^{N\cdot\Rs} (=|\setM_\mathrm{s}|)$ disjoint subcodes $\setC\defeq\bigcup_{i\in\setM_\mathrm{s}}\setC'_i$, with $\{\setC'_{i}\}_{i\in\setM_\mathrm{s}}$ being cosets of $\setC'$. Alice encodes her secret message $\ms\in\setM_\mathrm{s}$ by choosing a \emph{random} codeword from the coset $\setC'_i\big|_{i={\ms}}\!$. Accordingly, Alice also requires a dummy message $\mr\in\setM_\mathrm{r}$ with~$|\set{\Mr}|\defeq2^{\dim(\setC')}$ to capture the local randomness required for choosing the \emph{random} codeword from $\setC'_i$. Let $\Mr$ be a random~variable corresponding to a uniformly chosen dummy message $\mr$ from $\setM_\mathrm{r}$, independent of $\Ms$. In this setting, the non-zero term $H(\vX^N|\Ms)=H(\Mr)=N\cdot\Rr$, with $\Rr\defeq\frac{1}{N}\log|{\setM_\mathrm{r}}|$ representing the rate of the dummy message, ensures the \text{non-bijective} mapping required to meet the secrecy criterion.
	
	
	The mentioned nested structure is implemented in the \emph{outer} stage using LDPC codes. Let $\vms\in\{0,1\}^{\ks}$ and $\vmr\in\{0,1\}^{\kr}$ be the binary vectors corresponding to the secret message and the dummy message, respectively. Then, $\vm\defeq(\vms,\vmr)$ denotes the message vector, which is partitioned into $m\in\Z$ subvectors $\{\vm^{(l)}\}_{l=1}^m$ with $\vm^{(l)}\defeq(\vms^{(l)},\vmr^{(l)})$. For all $1\leq l\leq m$, each subvector $\vm^{(l)}$ is systematically encoded using the $l$-th LDPC subcode of rate $\Rout^{(l)}$ and blocklength $N^{(l)}\in\Z$. Accordingly, the associated lengths of $\vm^{(l)}$, $\vms^{(l)}$, and $\vmr^{(l)}$ are respectively~$N^{(l)}\cdot\Rout^{(l)}$ and
	\begin{align*}
		\ks\cdot\frac{N^{(l)}\cdot\Rout^{(l)}}{\sum_{\imath=1}^{m}N^{(\imath)}\cdot\Rout^{(\imath)}},\quad
		\kr\cdot\frac{N^{(l)}\cdot\Rout^{(l)}}{\sum_{\imath=1}^{m}N^{(\imath)}\cdot\Rout^{(\imath)}}.
	\end{align*}
	
	To prevent the eavesdropper from extracting secret information using the systematic part of codewords, we \emph{puncture} the secret message subvectors $\vms^{(l)}$ from output blocks of the $l$-th LDPC subcode, for all $1\leq l\leq m$. In our simulations, we choose the values of $N^{(l)}$ (for $1\leq l\leq m$) such that the blocklengths of all \emph{punctured} subcodes become a constant integer independent of~$l$ (see Remark~\ref{rem:fixedN}), denoted~by
	\begin{equation*}
		N'\defeq N^{(l)}\cdot\Bigg(1-\frac{\ks\cdot\Rout^{(l)}}{\sum_{\imath=1}^{m}N^{(\imath)}\cdot\Rout^{(\imath)}}\Bigg),\quad\text{for all $1\leq l\leq m$.}
	\end{equation*}	
	For $q^{(l)}\defeq(N^{(l)}-N')/N^{(l)}$ being the fraction of the \emph{punctured} VNs in the $l$-th LDPC subcode, the rate of the $l$-th \emph{punctured} subcode becomes $\Rout'^{(l)}\defeq\Rout^{(l)}/(1-q^{(l)})$ and the rate of the overall  \emph{punctured} outer-stage LDPC code becomes $$\Rout'\defeq\frac{1}{m}\sum_{l=1}^{m}\Rout'^{(l)}.$$
	
	Let $\big(d_l(t)\big)_{t=1}^{N'}$ denote the \emph{punctured} output block from the $l$-th LDPC encoder, for all $1\leq l\leq m$. These $m$ output blocks $\big\{\big(d_l(t)\big)_{t=1}^{N'}\big\}_{l=1}^m$ are interleaved into ${N'}$ vectors $\{\vd^m(t)\}_{t=1}^{N'}$, where each $\vd^m(t)$ is encoded with the superchannel, for all $1\!\leq\! t \leq {N'}$ (see Fig.~\ref{Fig:ConcGraph}). Finally, the channel observations $\vy^{nN'}$ and $\vz^{nN'}$ are received at Bob's and Eve's decoders, respectively.
	
	It is crucial to apply puncturing directly to the outputs of the outer LDPC code stage \emph{before} they are encoded by the inner trellis code (see Fig.~\ref{Fig:ConcGraph}). If puncturing is performed \emph{after} encoding with the inner code stage, the memory effect of the trellis code can cause private information to leak into adjacent positions, potentially exposing it to the eavesdropper.\footnote{We adopt systematic encoding followed by puncturing of the secret bits to strengthen security, although this is not required by the weak secrecy criterion. This introduces a small reliability penalty, but ensures that confidential bits do not directly appear in the transmitted signal and strengthens weak secrecy.} 
	
	\subsection{Asymptotic Analysis of the Two-Stage Turbo Decoder}\label{subsec:de}
	
	An ensemble of LDPC codes is specified by a blocklength and a fixed degree distribution of the corresponding bipartite graphs. Let $\lambda_\jmath^{(l)}$ ($\rho_\jmath^{(l)}$) denote the fraction of edges connected to the VNs (CNs) of degree $\jmath$ in the $l$-th subcode, for $1\leq l\leq m$. The degree distribution of the $l$-th subcode ensemble is represented by $(\bm{\lambda}^{(l)},\bm{\rho}^{(l)})$, where $\bm{\lambda}^{(l)}\defeq(\lambda_\jmath^{(l)})_\jmath$ and $\bm{\rho}^{(l)}\defeq(\rho_\jmath^{(l)})_\jmath$. The rate of the $l$-th \emph{unpunctured} subcode from this ensemble is calculated by\vspace{-6pt}
	$$
	\Rout^{(l)}\,=\,1
	-%
	\frac{\sum_\jmath (\rho_\jmath^{(l)}/\jmath) }{\sum_\jmath (\lambda_\jmath^{(l)}/\jmath) }.
	$$
	
	
	Density evolution is a method for tracking the evolution of the message PDFs passing through the nodes of a factor graph~\cite{910578}. Applying density evolution to the factor graph of a sum-product decoder enables one to asymptotically approximate the error probability of the corresponding code ensemble at each iteration of the decoder. In the following, we apply density evolution to the normal factor graph of the concatenated code involving the sum-product decoder of the $l$-th \emph{punctured} subcode.\footnote{The concentration of the decoder behavior and the zero-error threshold around their respective expected values over all possible graphs and cosets in the code ensemble follows directly from the results presented in~\cite[Sec.~III]{1207365}.} We omit the subcode index $l$ from message PDFs for brevity of notation.
	
	
	Let $f_\mathrm{vc}^{(\hell)}(\xi)$ be the density of messages from VNs to CNs at the $\hell$-th iteration of the decoder ($\hell\in\Z$). Accordingly, the block error probability at the $\hell$-th iteration is approximated by~\cite{910578}
	\begin{equation}\label{equ:pdf:er}
		\mathrm{e}_\hell\defeq\int_{-\infty}^{0}f_\mathrm{vc}^{(\hell+1)}(\xi)\mathrm{d}\xi.
	\end{equation}
	
	Let $f_\mathrm{fv}^{(\hell)}(\xi)$ be the density of messages from FNs to VNs and let $f_\mathrm{cv}^{(\hell)}(\xi)$ be the density of messages from CNs to VNs. Then, the density $f_\mathrm{vc}^{(\hell+1)}(\xi)$ in~\eqref{equ:pdf:er} is calculated by
	\begin{align*}
		f_\mathrm{vc}^{(\hell+1)}(\xi)
		= \big(
		q^{(l)}\cdot\delta(\xi)
		+&(1-q^{(l)})\cdot f_\mathrm{fv}^{(\hell)}(\xi)
		\big)\\
		&\ast
		\sum_\jmath
		\lambda_\jmath^{(l)}\big(f_\mathrm{cv}^{(\hell)}(\xi)\big)
		^{(\ast(\jmath-1))},
	\end{align*}
	where the impulse of weight $q^{(l)}$ at zero refers to the decoder's uncertainty (zero belief) about the \emph{punctured} VNs.
	The density $f_\mathrm{cv}^{(\hell+1)}(\xi)$ is calculated by evolving $f_\mathrm{vc}^{(\hell+1)}(\xi)$ through the CNs of the LDPC decoder. This can be accomplished by defining the convolution in a different domain or by quantizing the densities and using precomputed lookup~tables~as~in~\cite{910577,905935}.
	
	
	Let the fraction of VNs of degree $\jmath$ in the $l$-th subcode be denoted by
	$$
	\Lambda_\jmath^{(l)}
	\defeq
	\frac{\lambda_\jmath^{(l)}}%
	{\jmath\cdot\sum_d \frac{\lambda_d}{d}}.
	$$
	Then, the density of messages from VNs to FNs is calculated by
	\begin{equation*}
		f_\mathrm{vf}^{(\hell+1)}(\xi)
		=\sum_\jmath
		\Lambda_\jmath^{(l)}\big(f_\mathrm{cv}^{(\hell+1)}(\xi)\big)
		^{(\ast(\jmath-1))}.
	\end{equation*}
	Finally, the density $f_\mathrm{fv}^{(\hell+1)}(\xi)$ is derived by evolving $f_\mathrm{vf}^{(\hell+1)}(\xi)$ through the FNs of the superchannel using Monte Carlo simulations. This process involves generating a sample block $\vy^{nN'}$,\footnote{First, we generate $mN'$ i.u.d.\ symbols $\vd^{mN'}\in\{0,1\}^{mN'}$\!. Then, we encode $\big\{{\vd}^m(t)\big\}_{t=1}^{N'}$ using the superchannel to obtain ${\vy}^{nN'}$ at Bob's output.} and $mN'$ independent realizations from $f_\mathrm{vf}^{(\hell+1)}(\xi)$. After evolving the realizations from the density $f_\mathrm{vf}^{(\hell+1)}(\xi)$ as the extrinsic information and ${\vy}^{nN'}$ as the channel observations through a windowed BCJR decoder, the density $f_\mathrm{fv}^{(\hell+1)}(\xi)$ is estimated by a PMF using the histogram of the soft-output messages from the BCJR decoder. (See~\cite[Sec.~IV-B]{1207365},~\cite{866615}.) It is worth noting that the approximation error can be made arbitrarily small by choosing the value of $N'$ sufficiently large.
	
	
	\section{Code Design}\label{sec::DESIGN}
	
	\begin{figure*}
			\captionsetup{justification=centering}
			\captionof{table}{\textsc{Parameters of the Optimized Markov Source, \!$\{Q_{ij}^{\ast(\ell)}\}_{(i,\ell,j)\in\check{\setBE}}$ and \!$\{\mu_i^*\}_{i\in\setSE}$, Obtained Using~\cite[Alg.~1]{10068266}. Optimal Values of $\{k_i\in\Z\}_{i\in\setSE}$, $\{\kappa_i\}_{i\in\setSE}$, $\{\varv_{ij}^{(\ell)}\in\Z\}_{(i,\ell,j)\in\check{\setBE}}$, and $\{\check{Q}_{ij}^{(\ell)}\}_{(i,\ell,j)\in\check{\setBE}}$, Obtained by Solving \eqref{opt:rule1} and \eqref{opt:rule2} (Example~\ref{ex:innerate} and Example~\ref{ex:req1}).}}
			\label{tbl:trellis1}
			\scalebox{0.755}{
				\parbox{1.5\linewidth}{
					{\begin{tabular}{||c|c|c|c||}
							\hline\hline
							&&&\\[-2.25ex]
							$\betaE(i,\ell,j)$ &  $Q_{ij}^{\ast(\ell)}$  &  $\varv_{ij}^{(\ell)}$  &  $\check{Q}_{ij}^{(\ell)}$  \\[0.5ex]
							\hline\hline
							&&&\\[-2.25ex]
							$\betaE(0,0,0)$ &  0.2004 &   80  & 0.2000 \\
							$\betaE(0,0,1)$ &  0.0051 &   2   & 0.0050 \\
							$\betaE(0,0,2)$ &  0.0082 &   3   & 0.0075 \\
							$\betaE(0,0,3)$ &  0.0028 &   1   & 0.0025 \\\hline &&&\\[-2.25ex]
							$\betaE(0,1,0)$ &  0.0196 &   8   & 0.0200 \\
							$\betaE(0,1,1)$ &  0.0005 &        0   &      0 \\
							$\betaE(0,1,2)$ &  0.0008 &        0   &      0 \\
							$\betaE(0,1,3)$ &  0.0006 &        0   &      0 \\\hline &&&\\[-2.25ex]
							$\betaE(0,2,0)$ &  0.0087 &   4   & 0.0100 \\
							$\betaE(0,2,1)$ &  0.0089 &   4   & 0.0100 \\
							$\betaE(0,2,2)$ &  0.0001 &        0   &      0 \\
							$\betaE(0,2,3)$ &  0.0077 &   3   & 0.0075 \\\hline &&&\\[-2.25ex]
							$\betaE(0,3,0)$ &  0.0492 &  20   & 0.0500 \\
							$\betaE(0,3,1)$ &  0.0389 &  16   & 0.0400 \\
							$\betaE(0,3,2)$ &  0.0003 &        0   &      0 \\
							$\betaE(0,3,3)$ &  0.0687 &  27   & 0.0675 \\[.05ex]\hline
							\multicolumn{4}{|c|}{}\\[-2.5ex]
							\multicolumn{4}{|c|}{$\mu_0^\ast=0.4204$}\\
							\multicolumn{4}{|c|}{$k_0=84$, $\kappa_0=0.4200$}\\\hline\hline
						\end{tabular}\,%
						~\begin{tabular}{||c|c|c|c||}
							\hline\hline
							&&&\\[-2.25ex]
							$\betaE(i,\ell,j)$ &  $Q_{ij}^{\ast(\ell)}$  &  $\varv_{ij}^{(\ell)}$  & $\check{Q}_{ij}^{(\ell)}$ \\[0.5ex]
							\hline\hline
							&&&\\[-2.25ex]
							$\betaE(1,0,0)$ &  0.0040 &   2  &  0.0050 \\
							$\betaE(1,0,1)$ &  0.0001 &        0  &       0 \\
							$\betaE(1,0,2)$ &  0.0069 &   3  &  0.0075 \\
							$\betaE(1,0,3)$ &  0.0040 &   2  &  0.0050 \\\hline &&&\\[-2.25ex]
							$\betaE(1,1,0)$ &  0.0007 &        0  &       0 \\
							$\betaE(1,1,1)$ &  0.0000 &        0  &       0 \\
							$\betaE(1,1,2)$ &  0.0010 &        0  &       0 \\
							$\betaE(1,1,3)$ &  0.0011 &        0  &       0 \\\hline &&&\\[-2.25ex]
							$\betaE(1,2,0)$ &  0.0000 &        0  &       0 \\
							$\betaE(1,2,1)$ &  0.0000 &        0  &       0 \\
							$\betaE(1,2,2)$ &  0.0001 &        0  &       0 \\
							$\betaE(1,2,3)$ &  0.0069 &   3  &  0.0075 \\\hline &&&\\[-2.25ex]
							$\betaE(1,3,0)$ &  0.0009 &        0  &       0 \\
							$\betaE(1,3,1)$ &  0.0011 &        0  &       0 \\
							$\betaE(1,3,2)$ &  0.0004 &        0  &       0 \\
							$\betaE(1,3,3)$ &  0.0099 &   4  &  0.0100 \\[.05ex]\hline
							\multicolumn{4}{|c|}{}\\[-2.5ex]
							\multicolumn{4}{|c|}{$\mu_1^\ast=0.0369$}\\
							\multicolumn{4}{|c|}{$k_1=7$, $\kappa_1=0.0350$}\\\hline\hline
						\end{tabular}\,%
						~\begin{tabular}{||c|c|c|c||}
							\hline\hline
							&&&\\[-2.25ex]
							$\betaE(i,\ell,j)$ &  $Q_{ij}^{\ast(\ell)}$  &  $\varv_{ij}^{(\ell)}$  & $\check{Q}_{ij}^{(\ell)}$ \\[0.5ex]
							\hline\hline
							&&&\\[-2.25ex]
							$\betaE(2,0,0)$ &  0.0472 &  19 &   0.0475 \\
							$\betaE(2,0,1)$ &  0.0005 &        0 &        0 \\
							$\betaE(2,0,2)$ &  0.0011 &   1 &   0.0025 \\
							$\betaE(2,0,3)$ &  0.0009 &        0 &        0 \\\hline &&&\\[-2.25ex]
							$\betaE(2,1,0)$ &  0.0038 &   2 &   0.0050 \\
							$\betaE(2,1,1)$ &  0.0000 &        0 &        0 \\
							$\betaE(2,1,2)$ &  0.0002 &        0 &        0 \\
							$\betaE(2,1,3)$ &  0.0000 &        0 &        0 \\\hline &&&\\[-2.25ex]
							$\betaE(2,2,0)$ &  0.0026 &   1 &   0.0025 \\
							$\betaE(2,2,1)$ &  0.0008 &        0 &        0 \\
							$\betaE(2,2,2)$ &  0.0000 &        0 &        0 \\
							$\betaE(2,2,3)$ &  0.0022 &   1 &   0.0025 \\\hline &&&\\[-2.25ex]
							$\betaE(2,3,0)$ &  0.0263 &  11 &   0.0275 \\
							$\betaE(2,3,1)$ &  0.0047 &   2 &   0.0050 \\
							$\betaE(2,3,2)$ &  0.0002 &        0 &        0 \\
							$\betaE(2,3,3)$ &  0.0065 &   3 &   0.0075 \\[.05ex]\hline
							\multicolumn{4}{|c|}{$\mu_2^\ast=0.0972$}\\
							\multicolumn{4}{|c|}{$k_2=20$, $\kappa_2=0.1000$}\\\hline\hline
						\end{tabular}\,%
						~\begin{tabular}{||c|c|c|c||}
							\hline\hline
							&&&\\[-2.25ex]
							$\betaE(i,\ell,j)$ &  $Q_{ij}^{\ast(\ell)}$  &  $\varv_{ij}^{(\ell)}$  & $\check{Q}_{ij}^{(\ell)}$ \\[0.5ex]
							\hline\hline
							&&&\\[-2.25ex]
							$\betaE(3,0,0)$ &  0.0606 &  24  &  0.0600 \\
							$\betaE(3,0,1)$ &  0.0006 &   0  & 		 0 \\
							$\betaE(3,0,2)$ &  0.0383 &  15  &  0.0375 \\
							$\betaE(3,0,3)$ &  0.0397 &  16  &  0.0400 \\\hline &&&\\[-2.25ex]
							$\betaE(3,1,0)$ &  0.0047 &   2  &  0.0050 \\
							$\betaE(3,1,1)$ &  0.0000 &   0  &       0 \\
							$\betaE(3,1,2)$ &  0.0025 &   1  &  0.0025 \\
							$\betaE(3,1,3)$ &  0.0014 &   1  &  0.0025 \\\hline &&&\\[-2.25ex]
							$\betaE(3,2,0)$ &  0.0003 &   0  &       0 \\
							$\betaE(3,2,1)$ &  0.0005 &   0  &       0 \\
							$\betaE(3,2,2)$ &  0.0019 &   1  &  0.0025 \\
							$\betaE(3,2,3)$ &  0.0564 &  23  &  0.0575 \\\hline &&&\\[-2.25ex]
							$\betaE(3,3,0)$ &  0.0044 &   2  &  0.0050 \\
							$\betaE(3,3,1)$ &  0.0038 &   1  &  0.0025 \\
							$\betaE(3,3,2)$ &  0.0091 &   4  &  0.0100 \\
							$\betaE(3,3,3)$ &  0.2213 &  88  &  0.2200 \\[.05ex]\hline
							\multicolumn{4}{|c|}{}\\[-2.5ex]
							\multicolumn{4}{|c|}{$\mu_3^\ast=0.4456$}\\
							\multicolumn{4}{|c|}{$k_3=89$, $\kappa_3=0.4450$}\\\hline\hline
			\end{tabular}}}}\vspace*{10pt}
	\end{figure*}
	
	\subsection{Inner Trellis Code Stage}\label{sec:innertrellis}
	
	In this section, we apply the design principles outlined in~\cite{1397934} to design a trellis code that maps an input i.u.d.\ process to an output Markov process, achieving the constrained secrecy capacity of the ISI-WTC. We begin by determining the desired rate of the inner trellis code.
	
	Let $\matrQ^*\defeq\{Q_{ij}^{\ast(\ell)}\}_{(i,\ell,j)\in\check{\setBE}}$ characterize an optimized Markov source, achieving the constrained secrecy capacity~\eqref{equ:csc} of an $n$-th order E-FSMC specified by $\check\setBE$, and let $$p_{ij}^{\ast(\ell)}\defeq \frac{Q_{ij}^{\ast(\ell)}}{\sum_{\ell,j:(i,\ell,j)\in\check{\setBE}} Q_{ij}^{\ast(\ell)}},\quad \text{for all } (i,\ell,j)\in\check{\setBE},$$ represent the corresponding state transition probabilities. Additionally, let $\Rd\defeq\Rin\cdot\Rout'$ be the design rate of the concatenated code and let $\Rin\triangleq m/n$ be the rate of the inner trellis code. Since the BCJR decoding complexity grows exponentially with the trellis code's input vector length~$m$ (see \cite{1397934} and Remark~\ref{rem:comp}),
	we designate the smallest values of $m$ and $n$, satisfying
	\begin{equation}\label{ex:rule1}
		\Rd<\Rin
		\defeq\frac{m}{n}
		\leq\frac{1}{n}\min_{(i,\ell,j)\in\check{\setBE}}\bigg(
		\log\frac{1}{p_{ij}^{\ast(\ell)}}
		\bigg),
	\end{equation}
	where obeying the strict lower bound $\Rin>\Rd$ allows for the use of an outer code stage of rate $\Rout'={\Rd}/{\Rin}$ to implement the nested structure of the wiretap codes (Section~\ref{subsec:Enc}). Additionally, the upper bound ensures the unique decodability of the trellis code. Specifically, as will be demonstrated later in this section, the main idea in designing the trellis code is to find a superchannel where $\Pr\big(b(t)\in\setT_{ij}^{(\ell)}|s(t-1)\in\setT_i\big)$ approaches $p_{ij}^{\ast(\ell)}$ for all $(i,\ell,j)\in\check{\setBE}$. Nevertheless, since each branch emanating from a state $s(t-1)\in\setT_i$ of the superchannel must be from a distinct element of the set $\{\setT_{ij}^{(\ell)}\}_{\ell,j:(i,\ell,j)\in\check{\setBE}}$ and because the input process of the superchannel is i.u.d., one can verify $$\Pr\big(b(t)\in\setT_{ij}^{(\ell)}|s(t-1)\in\setT_i\big)\leq 2^{-m},\quad\!\! \text{for all } (i,\ell,j)\in\check{\setBE},$$  implying the upper bound~\eqref{ex:rule1}.
	
	\begin{remark}\label{rem:optdesignrt}
		To ensure that the mother code ($\setC$ in Section~\ref{subsec:Enc}) is achievable over Bob's channel, the target design rate $\Rd^*$ of the concatenated code is set to be the achievable information rate of the Markov source specified by $\matrQ^*\in\setQ(\check{\setBE})$ over Bob's channel, satisfying only the reliability criterion in~\eqref{equ:crit:sec}. Therefore, according to~\cite{965976}, we set
		\begin{align*}
			\Rd^*\defeq\lim_{N\to\infty}\frac{I(\vX^{nN};\vY^{nN})}{nN}=\lim_{N\to\infty}\frac{I(\vBE^N; \vY^{nN})}{nN},
		\end{align*}
		where $\{\BE(t)\}_{t\in\Z}$ is a branch sequence of the E-FSMC generated according to the Markov process specified by $\matrQ^*\in\setQ(\check{\setBE})$. This mutual information rate term can be efficiently calculated using variants of the sum-product/BCJR algorithm proposed in~\cite{1661831} due to the stationarity and ergodicity of both the E-FSMC and the input Markov process. \remarkend
	\end{remark}
	
	
	\begin{example}\label{ex:innerate}
		The optimal branch probabilities of the second-order and third-order E-FSMCs, corresponding to the ISI-WTC in {Example}~\ref{ex:casestudy}, violate the upper bound constraint presented in~\eqref{ex:rule1}, even for $m=1$ (i.e., we have $p_{ij}^{\ast(\ell)}>2^{-m}$ for some $(i,\ell,j)\in\check{\setBE}$). By setting $n=4$, the optimized branch probabilities  of the fourth-order E-FSMC satisfy the upper bound with $m=1$ and achieve $$\lim_{N\to\infty}\frac{I(\vBE^N;\vY^{nN})}{nN}=0.195\quad(\sfrac{\text{bits}}{\text{channel use}}),$$ over Bob's channel. Therefore, the target design rate of the two-stage code and the rate of the inner trellis code are set to $\Rd^*=0.195$ and $\Rin=1/4$, respectively. The optimal values of $\{Q_{ij}^{\ast(\ell)}\}_{(i,\ell,j)\in\check{\setBE}}$ and $\{\mu_i^\ast\defeq\sum_{\ell,j:(i,\ell,j)\in\check{\setBE}}Q_{ij}^{\ast(\ell)}\}_{i\in\setSE}$ are shown in Table~\ref{tbl:trellis1}.
		\exampleend
	\end{example}
	
	
	During the design of the desired trellis code, we optimize the parameters of the superchannel while assuming that the constituent E-FSMC is fixed. Following Definition~\ref{def:superchannel:not}, let $k_i\defeq|\setT_i|$ and $\varv_{ij}^{(\ell)}\defeq|\setT_{ij}^{(\ell)}|$. Accordingly, the total number of states in the superchannel is $|\setS|=\sum_{i\in\setSE}k_i$, and the PMF $\kappa_i\defeq {k_i}/|\setS|$ is defined to be the fraction of states in $\setT_i$. Additionally, the total number of branches emanating from the states in $\setT_i$ is $k_i\cdot 2^m=\sum_{\ell,j:(i,\ell,j)\in\check{\setBE}}\varv_{ij}^{(\ell)}$, and the PMF $\upsilon_{ij}^{(\ell)}\defeq  \varv_{ij}^{(\ell)}/(k_i\cdot 2^m)$ is defined to be the fraction of branches in $\setT_{ij}^{(\ell)}$ over the total number of branches emanating from the states in $\setT_i$. Note that the main idea in designing the desired trellis code within the superchannel is to split every state $i\in\setSE$ (branch $(i,\ell,j)\in\check{\setBE}$) of the E-FSMC into $k_i\geq 1$ states ($\varv_{ij}^{(\ell)}\geq 1$~branches) of the superchannel so that the following two criteria are fulfilled.
	
	\begin{requirement}\label{req:rule12}
		Given an i.u.d.\ input process $\{D(t)\}_{t\in\Z}$, the output process $\{X(t)\}_{t\in\Z}$ of the trellis code  should emulate the output process of the optimized Markov source characterized by $\mathbf{Q}^\ast\in\setQ(\check{\setBE})$. To satisfy this criterion, the probability of ``being in a state from $\setT_i$ and taking a branch from $\setT_{ij}^{(\ell)}$'' is attempted to be matched on $Q_{ij}^{\ast(\ell)}$, for all $(i,\ell,j)\in\check{\setBE}$. Specifically, for $\check{Q}_{ij}^{(\ell)}\defeq\kappa_i\cdot\upsilon_{ij}^{(\ell)}=\varv_{ij}^{(\ell)}/(|\setS|\cdot 2^m)$ and $\check{\mathbf{Q}}\defeq\{\check{Q}_{ij}^{(\ell)}\}_{(i,\ell,j) \in \setBE}$, we attempt to match $\check{\mathbf{Q}}$ on the optimal PMF $\mathbf{Q}^\ast$ by properly choosing the integer values of  $\{k_i\}_{i\in\setSE}$ and $\{\varv_{ij}^{(\ell)}\}_{(i,\ell,j)\in\check{\setBE}}$. Following~\cite{1397934}, we break this criterion into the following optimization problems;
			
			(i) Let $J\in\Z$ denote the maximum number of states allowed for the designed superchannel. Also, let $\bm{\mu}^\ast\defeq \{\mu_i^\ast\}_{i\in\setSE}$ (Table~\ref{tbl:trellis1}) and $\bm{\kappa}\defeq\{\kappa_i\defeq{k_i}/{|\setS|}\}_{i\in\setSE}$. The probability of realizing a state from $\setT_i$ is attempted to be matched on $\mu_i^\ast$ by solving
			\begin{align*}
					\min_{\{k_i\in\Z\}_{i\in\setSE}}\mathrm{D}_\mathrm{KL}(\bm{\kappa}||\bm{\mu}^\ast), \qquad
					\mathrm{s.t.}\quad|\setSE|\leq \sum_{i\in\setSE}k_i \leq J;\alabel{opt:rule1}
			\end{align*}
			
			(ii) Let $\bm{\pi}_i^\ast\defeq\{p_{ij}^{\ast(\ell)}\}_{\ell,j:(i,\ell,j)\in\check{\setBE}}$ (Table~\ref{tbl:trellis1}) and $\bm{\upsilon}_i\defeq\{\upsilon_{ij}^{(\ell)}\defeq\varv_{ij}^{(\ell)}/(k_i\cdot 2^m)\}_{\ell,j:(i,\ell,j)\in\check{\setBE}}$ for all $i\in\setSE$. Given the PMF vector $\bm{\kappa}$ derived from solving \eqref{opt:rule1}, the probability of realizing a branch within $\setT_{ij}^{(\ell)}$ from a state in $\setT_i$ is attempted to be matched on $p_{ij}^{\ast(\ell)}$ by solving
			\begin{align*}
					\min_{\{\varv_{ij}^{(\ell)}\in\Z\}_{(i,\ell,j)\in\check{\setBE}}}\sum_{i\in\setSE}\kappa_i\cdot
					&\mathrm{D}_\mathrm{KL}(\bm{\upsilon}_i||\bm{\pi}_i^\ast), \alabel{opt:rule2}\\ 
					\mathrm{s.t.}\sum_{\ell,j:(i,\ell,j)\in\check{\setBE}}\!\!&\varv_{ij}^{(\ell)}=k_i\cdot 2^m\quad(\text{for all $i\in\setSE$}),\\
					\sum_{\ell,j:(i,\ell,j)\in\check{\setBE}}\!\!&\varv_{ij}^{(\ell)}=\!\!\sum_{s,\ell:(s,\ell,i)\in\check{\setBE}}\!\!\varv_{si}^{(\ell)}\quad(\text{for all $i\in\setSE$}).
			\end{align*}
		A superchannel with the optimal values of $\{k_i\}_{i\in\setSE}$ and $\{\varv_{ij}^{(\ell)}\}_{(i,\ell,j)\in\check{\setBE}}$, vanishing the objective functions of \eqref{opt:rule1} and \eqref{opt:rule2} guarantees Criterion~\ref{req:rule12}.\footnote{The inaccuracy in approximating $\mathbf{Q}^\ast\in\setQ(\check{\setBE})$ with $\check{\mathbf{Q}}$ can be made arbitrarily small by choosing the number of allowed states for the superchannel ($J\in\Z$) sufficiently large.}
	\end{requirement}
	
	
	\begin{example}\label{ex:req1}
		Following Example~\ref{ex:innerate}, the optimal values of $k_i$, $\kappa_i$ for all $i\in\setSE$, and $\varv_{ij}^{(\ell)}$\!,\, $\check{Q}_{ij}^{(\ell)}$ for all $(i,\ell,j)\in\check{\setBE}$, are shown in Table~\ref{tbl:trellis1}, for $J=200$.
		\exampleend
	\end{example}
	
	
	\begin{requirement}\label{req:rule4}
		The i.u.d.\ secure rate of the superchannel,
		\begin{equation*}
			\RSC
			\defeq\lim_{N\rightarrow\infty}
			\frac{1}{nN}\bigg(
			I(\vD^{mN};\vY^{nN})
			-
			I(\vD^{mN};\vZ^{nN})\bigg)^+\!\!,\alabel{equ:RSC}
		\end{equation*}
		should approach the constrained secrecy capacity of the constituent $n$-th order E-FSMC~$\Rcsc^{(n)}$ at the required SNR value. To fulfill this criterion, among all branch connections of the superchannel that meet the optimal values of $\{k_i\}_{i\in\setSE}$ and $\{\varv_{ij}^{(\ell)}\}_{(i,\ell,j)\in\check{\setBE}}$, we choose the one that achieves the maximum i.u.d.\ secure rate $\RSC$ using ordinal optimization~\cite{HO1999169}.
	\end{requirement}
	
	
	\begin{figure*}
		\centering
		\includegraphics[scale=0.86]{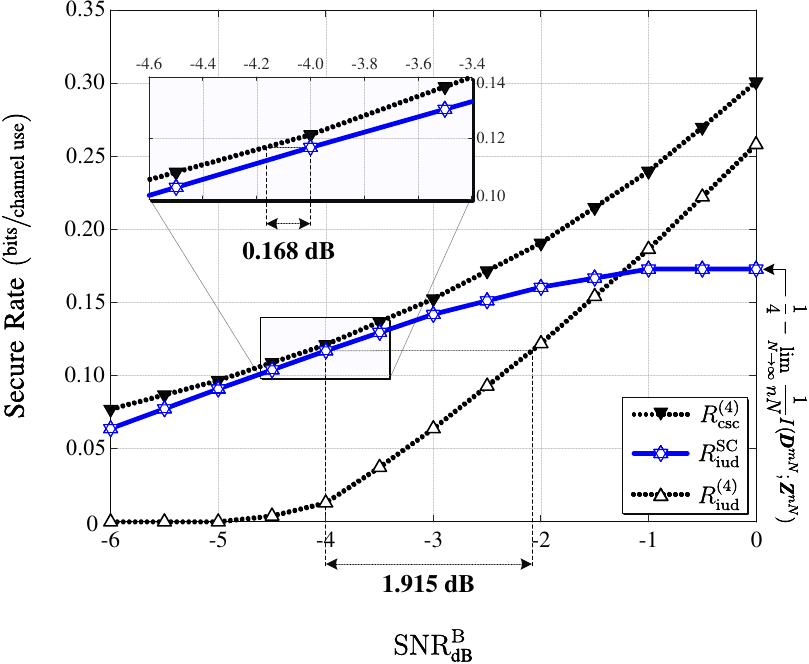}
		\caption{The i.u.d.\ secure rate $\Riud^{(4)}$, the constrained secrecy capacity $\Rcsc^{(4)}$ of the $4$-th order E-FSMC modeling the ISI-WTC in Example~\ref{ex:casestudy}, and the i.u.d.\ secure rate of the superchannel $\RSC$ in Table~\ref{tbl:trellis2} (Example~\ref{ex:supchatr}) for $\SNRdBE=-4.0$. Note that the superchannel is specifically designed for operating at $\SNRdBB=\SNRdBE=-4.0$.}\label{Fig:secrecy_gain}
	\end{figure*}
	
	\begin{example}\label{ex:supchatr}
		Table~\ref{tbl:trellis2} in Appendix~\ref{apx:tb} depicts the branch connections of a superchannel, satisfying the	optimal values of $\{k_i\}_{i\in\setSE}$ and $\{\varv_{ij}^{(\ell)}\}_{(i,\ell,j)\in\check{\setBE}}$ as given in Table~\ref{tbl:trellis1}. These connections were determined using a randomized search with ordinal optimization. For the SNR values of Example~\ref{ex:casestudy}, the i.u.d.\ secure rate of the resulting superchannel is $\RSC=0.117$~$(\sfrac{\text{bits}}{\text{channel use}})$ while the reference Markov source specified by $\matr{Q}^*$ achieves $\Rcsc^{(n)}|_{n=4}=0.121$~$(\sfrac{\text{bits}}{\text{channel use}})$. Fig.~\ref{Fig:secrecy_gain} illustrates the gain of the trellis code at $\SNRdBE=-4.0$ and $-6.0\leq\SNRdBB\leq 0$. Note the superchannel achieves its best performance at the required SNR value, i.e., $\SNRdBB=-4.0$, compared with the constrained secrecy capacity. As shown in Fig.~\ref{Fig:secrecy_gain}, the i.u.d.\ secure rate of the superchannel converges to $$\RSC = 0.25 - \lim_{N\to\infty}\frac{I(\vD^{mN};\vZ^{nN})}{nN},$$ as $\SNRdBB$ exceeds a certain threshold where $$\lim_{N\to\infty}\frac{I(\vD^{mN};\vY^{nN})}{nN}=0.25\quad(\sfrac{\text{bits}}{\text{channel use}}).$$ This occurs because the mutual information rate of the users over the superchannel is limited by the rate of the constituent trellis code ($\Rin=1/4$). To overcome this limitation, one could increase the rate of the trellis code by choosing a larger value of $m\in\Z$ while ensuring that the constraint in~\eqref{ex:rule1} is satisfied by properly choosing $n\in\Z$. However, the trade-off for this improvement is a larger outer code stage, a larger superchannel, and a more complex decoder.
		\exampleend
	\end{example}
	
	
	\begin{remark}
		The choice of bit assignment from the set $\{0,1\}^m$ to $\vd^m$ on each branch $(s(t-1),\vd^m,s(t))\in\set{B}$ of the superchannel significantly influences the block error probability of the outer code. Therefore, the final step in establishing the connection between the superchannel and the outer code stage is to assign optimal values to $\vd^m\in\{0,1\}^m$ such that the block error probability at Bob's decoder is minimized. By adopting a multi-stage decoder as in~\cite{1055718} (see Fig.~\ref{Fig:decoder}), determining the values of $\vd^{l-1}$ ($1 \leq l \leq m$) ensures that the performance of the $l$-th LDPC decoder depends only on the bit values assigned to $d(l)$~\cite{1397934}. Accordingly, starting with $l=1$, we search for the optimal values of $d(l)$ that minimize the error probability of the $l$-th BCJR decoder in Fig.~\ref{Fig:decoder} at the first iteration, assuming the previous messages $(\vm^{(1)},\ldots,\vm^{(l-1)})$ have been successfully decoded. The resulting bit assignments to branches of the superchannel in Example~\ref{ex:supchatr} are listed in Table~\ref{tbl:trellis2} (Appendix~\ref{apx:tb}). \remarkend
	\end{remark}

	\subsection{Outer LDPC Code Stage}
	In this section, we estimate the rate of information leakage  using a numerical upper bound. Then, we optimize the degree distributions of the outer LDPC subcodes to reduce the estimated upper bound. We begin with the following remark, which considers an auxiliary receiver terminal (cf.~\cite{6162586}) whose decoding block error probability is used to derive this bound.
	\begin{remark}\label{rem:fict}
		Consider a fictitious receiver terminal that observes through Eve's channel and has perfect knowledge about the secret message. This receiver is referred to as ``Frank''. Let $\epsE_N$ be the block error probability of Frank's decoder, trying to recover $\Mr$ through observing $\Ms$ and $\vZ^N$. Indeed, Fano's inequality implies $H(\Mr|\Ms,\vZ^N)\leq1+\epsE_N\cdot\kr.$
		\remarkend
	\end{remark}
	
	\begin{lemma}\label{lem:leakage_bd}
		Consider a concatenated code as in Fig.~\ref{Fig:ConcGraph} with the blocklength $nN'$, the dummy message rate $\Rr$, and the block error probability $\epsE_{nN'}$ at Frank's decoder. The rate of information leakage achieved by this code over an ISI-WTC is upper-bounded by
		\begin{equation}\label{equ:leakage_num}
			{\varepsilon_{nN'}}\leq \frac{I(\vD^{mN'};\vZ^{nN'})}{nN'}-(1-\epsE_{nN'})\cdot\Rr+\frac{1}{nN'}.
		\end{equation}
	\end{lemma}
	
	\begin{proof} We have
		\begin{align*}
			&nN'\cdot\varepsilon_{nN'} \triangleq H(\Ms)-H(\Ms|\vZ^{nN'})\\
			&=\ks-H(\vD^{mN'}|\vZ^{nN'})-H(\Ms|\vD^{mN'},\vZ^{nN'})\\&\qquad+H(\vD^{mN'}|\Ms,\vZ^{nN'})\\
			&\stackrel{(a)}{\leq} \ks-H(\vD^{mN'})+I(\vD^{mN'}\!;\vZ^{nN'})-H(\Ms|\vD^{mN'}\!\!,\vZ^{nN'})\\&\qquad+\epsE_{nN'}\cdot\kr+1\\
			&\stackrel{(b)}{=} \ks - H(\Ms,\Mr) - H(\vD^{mN'}|\Ms,\Mr)\\&\qquad + H(\Ms,\Mr|\vD^{mN'}) + I(\vD^{mN'}\!\!;\vZ^{nN'})-H(\Ms|\vD^{mN'})\\&\qquad+\epsE_{nN'}\cdot\kr+1\\
			&\stackrel{(c)}{=} H(\Ms|\vD^{mN'}) + H(\Mr|\vD^{mN'},\Ms) + I(\vD^{mN'};\vZ^{nN'})\\&\qquad-H(\Ms|\vD^{mN'})-(1-\epsE_{nN'})\cdot\kr+1\\
			&\stackrel{(d)}{=} I(\vD^{mN'};\vZ^{nN'})-(1-\epsE_{nN'})\cdot\kr+1,
		\end{align*}
		where $(a)$ is deduced from $H(\vD^{mN'}|\Ms,\vZ^{nN'})=H(\Mr|\Ms,\vZ^{nN'})\leq1+\kr\cdot\epsE_{nN'}$ in {Remark}~\ref{rem:fict}, $(b)$ relies on	 $H(\Ms|\vD^{mN'},\vZ^{nN'})=H(\Ms|\vD^{mN'})$, $(c)$ results from $H(\Ms,\Mr)=\ks+\kr$ and $H(\vD^{mN'}|\Ms,\Mr)=0$, and $(d)$ follows from the fact that the dummy message $\Mr$ is systematically available in the punctured codewords $\vD^{mN'}$\!\!, implying $H(\Mr|\vD^{mN'}\!\!\!,\Ms) \leq H(\Mr|\vD^{mN'}) = 0$.
	\end{proof}
	
	
	\begin{figure*}
		\centering
		\includegraphics[scale=0.99]{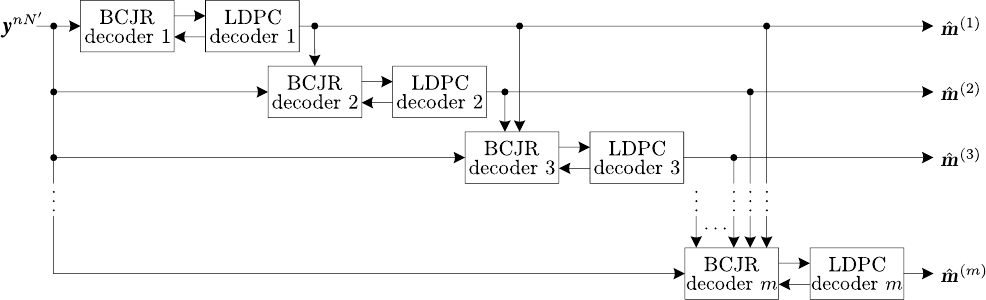}
		\caption{Multi-stage decoder for estimating $\vm^{(l)}$, assuming that the symbols $(\vm^{(1)},\ldots,\vm^{(l-1)})$ are successfully decoded ($1\leq l\leq m$).}\label{Fig:decoder}
	\end{figure*}
	
	The term $I(\vD^{mN'};\vZ^{nN'})/{(nN')}$ in~\eqref{equ:leakage_num} can be efficiently calculated using variants of the sum-product/BCJR algorithm~\cite{1661831} as $N'\!\!\to\!\infty$. Moreover, in the code design procedure, the asymptotic  error probability of Frank's decoder ($\epsE\defeq\lim_{N'\to\infty}\epsE_{nN'}$) is approximated by applying the following modifications to the density evolution algorithm.\footnote{Density evolution is used during the code design procedure to \emph{estimate} Frank's asymptotic error probability, while the actual finite-length secrecy performance is evaluated by \emph{measuring} Frank's bit error rate in simulations.} (i)~Since Frank's and Eve's observations are identical ({Remark}~\ref{rem:fict}), to evolve $f_\mathrm{vf}^{(\hell)}(\xi)$ through the FNs of Frank's decoder, we use Eve's observations ${\vz}^{nN'}$ in the Monte Carlo step instead of Bob's observations ${\vy}^{nN'}$\!. Let $f_\mathrm{fv}'^{(\hell)}(\xi)$ denote the resulting density. (ii)~Since Frank possesses perfect knowledge about the secret message ({Remark}~\ref{rem:fict}), we set the density of messages from FNs to VNs as a probabilistic mixture of the density $f_\mathrm{fv}'^{(\hell)}(\xi)$ and an impulse at $+\infty$ with the weight of $q^{(l)}$.\footnote{Recall that $q^{(l)}$ represents the fraction of VNs carrying the secret message at the $l$-th LDPC subcode. Consequently, the impulse of weight $q^{(l)}$ at $+\infty$ appropriately captures Frank's perfect knowledge about the secret message.} Thus, the density of messages from VNs to CNs of Frank's decoder is calculated by
	\begin{align*}
		f_\mathrm{vc}'^{(\hell+1)}(\xi)
		=\big(
		q^{(l)}\cdot\delta(\xi-\infty)
		+&(1-q^{(l)})\cdot f_\mathrm{fv}'^{(\hell)}(\xi)
		\big)\\
		&\ast
		\sum_\jmath
		\lambda_\jmath^{(l)}\big(f_\mathrm{cv}^{(\hell)}(\xi)\big)
		^{(\ast(\jmath-1))},
	\end{align*}
	and the block error probability at the $\hell$-th iteration of Frank's decoder is approximated by
	\begin{equation*}
		\mathrm{e}^{\mathrm{F}}_{\hell}=\int_{-\infty}^{0}f_\mathrm{vc}'^{(\hell+1)}(\xi)\mathrm{d}\xi.
	\end{equation*}
	The \emph{asymptotic} error probability of Frank's decoder ($\epsE$) is estimated by $\check{\epsilon}^{\mathrm{F}}\defeq\mathrm{e}^{\mathrm{F}}_{\hell_{\max}}$, where $\hell_{\max}\in\Z$ represents the maximum number of allowed iterations or the iteration at which a stopping criterion is met. Accordingly, we estimate the \emph{asymptotic} rate of information leakage by 
	\begin{equation}\label{equ:leakage_lstest}
		\check{\varepsilon}
		\defeq
		\lim_{N'\to\infty}\check{\varepsilon}_{nN'}=\lim_{N'\to\infty}\frac{I(\vD^{mN'};\vZ^{nN'})}{nN'}-\big(1-\check{\epsilon}^{\mathrm{F}}\big)\cdot\Rr.
	\end{equation}
		Since Frank is allowed to use a more powerful decoder, we have $\epsE\leq\check{\epsilon}^{\mathrm{F}}$; hence, due to~\eqref{equ:leakage_num}, the estimation in \eqref{equ:leakage_lstest} presents an upper bound $\lim_{N'\to\infty}{\varepsilon}_{nN'}\leq\check{\varepsilon}$. 
	
	Let the fraction of all \emph{punctured} VNs in the outer code stage be denoted by $$q\defeq\frac{\sum_{l=1}^{m}N^{(l)}-mN'}{\sum_{l=1}^{m}N^{(l)}}.$$ According to the employed puncturing policy (Section~\ref{subsec:Enc}), the secure rate of the concatenated code is determined by
	\begin{align*}
		\Rs&=\frac{\ks}{nN'}=\frac{\sum_{l=1}^{m}N^{(l)}-mN'}{nN'}
		=\frac{\sum_{l=1}^{m}N^{(l)}-mN'}{mN'}\cdot\frac{m}{n}\\
		&=\frac{q}{1-q}\cdot\Rin.\alabel{equ:fxsecrate}
	\end{align*}
	Throughout the design procedure, we fix the secure rate by applying a fixed value for $q$.
	
	We attempt to reduce the rate of information leakage by minimizing the estimated upper bound in~\eqref{equ:leakage_lstest}. Using $\Rd=\Rs+\Rr$, for a sufficiently small value of $\check{\epsilon}^{\mathrm{F}}>0$, we have
	\begin{alignat*}{2}
		\check{\varepsilon}
		&=\lim_{N\to\infty}\frac{I(\vD^{mN};\vZ^{nN})}{nN}-\Rr\\
		&=\lim_{N\to\infty}\frac{I(\vD^{mN};\vZ^{nN})}{nN}+\Rs-\Rd.
		\alabel{equ:desR_opt}
	\end{alignat*}
	Accordingly, for a fixed value of the secure rate and a fixed superchannel, minimizing the rate of information leakage is equivalent to the well-known problem of maximizing the design rate of the \emph{punctured} LDPC codes. In the following, we determine the target design rate of the outer code stage~and~allocate the resulting target rate among the constituent subcodes.

	
	By choosing the number of allowed states for the superchannel sufficiently large (Section~\ref{sec:innertrellis}), the information rate term $$\lim_{N\to\infty}\frac{I(\vD^{mN};\vY^{nN})}{nN},$$ over the superchannel (with $\{D(t)\}_{t\in\Z}$~being an i.u.d.\ sequence) approaches $$\lim_{N\to\infty}\frac{I(\vBE^N;\vY^{nN})}{nN}\bigg|_{\matr{Q}=\matr{Q}^*},$$ over the ISI-WTC.
	Therefore, following Remark~\ref{rem:optdesignrt}, we set
	\begin{equation}\label{equ:ultdesrate}
		\Rd^*
		=\lim_{N\rightarrow\infty}
		\frac{I(\vD^{mN};\vY^{nN})}{nN}%
		.
	\end{equation}
	Following Remark~\ref{rem:ultbnd}, by setting $\Rs=\Rs^\ast$ (where $\Rs^\ast=\RSC$ for an LDPC code operating over the superchannel) and using the expression for $\RSC$ in~\eqref{equ:RSC}, the asymptotic upper bound on the rate of information leakage in~\eqref{equ:desR_opt} becomes
		\begin{alignat*}{2}
			\check{\varepsilon}
			&=\lim_{N\to\infty}\frac{I(\vD^{mN};\vZ^{nN})}{nN}+\RSC-\Rd\\
			&=\lim_{N\to\infty}\frac{I(\vD^{mN};\vY^{nN})}{nN}-\Rd.
		\end{alignat*}
		Hence, achieving the design rate $\Rd=\Rd^{\ast}$ as presented in \eqref{equ:ultdesrate} is a sufficient (though not necessary) condition to guarantee the weak secrecy criterion for all secure rates $\Rs\leq\RSC$. 
	
	Due to \eqref{equ:ultdesrate}, the target design rate of the \emph{punctured} outer code stage is set to
	\begin{align*}
		\Rout'^*
		\defeq\frac{\Rd^*}{\Rin}
		=\lim_{N\rightarrow\infty}
		\frac{I(\vD^{mN};\vY^{nN})}{mN}.
	\end{align*}
	Also, according to $\Rout'^*\defeq\frac{1}{m}\sum_{l=1}^{m}\Rout'^{*(l)}$ and
	\begin{align*}
		I\Big(&\big(\vD^m(t)\big)_{t=1}^N;\vY^{nN}\Big)\\
		&=\sum_{l=1}^m
		I\Big(
		\big(D_l(t)\big)_{t=1}^N;\vY^{nN}
		\Big|\big(D_{1}(t)\cdots D_{l-1}(t)\big)_{t=1}^N
		\Big),
	\end{align*}
	we set the target design rate of the $l$-th \emph{punctured} subcode to be the achievable rate of Bob's decoder, trying to reliably estimate $\vm^{(l)}$ from observing $\vy^{nN}$ and $\big(d_{1}(t)\cdots d_{l-1}(t)\big)_{t=1}^N$
	\begin{equation*}
		\Rout'^{*(l)}
		\!\defeq\!\!
		\lim_{N\to\infty}\!
		\frac{1}{N}
		I\Big(\!
		\big(D_l(t)\big)_{t=1}^N;\vY^{nN}
		\Big|\big(D_{1}(t)\cdots D_{l-1}(t)\big)_{t=1}^N
		\Big),
	\end{equation*}
	where $\{D_l(t)\}_{t\in\Z}$ is an i.u.d.\ sequence for all $1\leq l\leq m$. Using the results from~\cite{965976} and the assumptions made in the multi-stage decoder shown in Fig.~\ref{Fig:decoder}, it is straightforward to verify that $\Rout'^{*(l)}$ is achievable for the $l$-th subcode under the reliability criterion specified in~\eqref{equ:crit:sec}, for all $1\leq l\leq m$. Consequently, we maximize the rate of the $l$-th \emph{punctured} subcode $\Rout'^{(l)}$ by optimizing the degree distributions $(\bm{\lambda}^{(l)},\bm{\rho}^{(l)})$ to approach $\Rout'^{\ast(l)}$\textemdash assuming that the codewords $\big(d_{1}(t)\cdots d_{l-1}(t)\big)_{t=1}^{N'}$ are successfully decoded (see Fig.~\ref{Fig:decoder}).
	
	Note that choosing the target rate for the $l$-th \emph{unpunctured} LDPC subcode ($\Rout^{*(l)} \defeq (1 - q^{(l)}) \cdot \Rout'^{*(l)}$) also requires determining the value of $q^{(l)}$, for $1 \leq l \leq m$. Given a fixed $\Rs$ (and thus a fixed $q$, see \eqref{equ:fxsecrate}), subject to the overall constraint $$\frac{1}{m}\sum_{l=1}^m \frac{1}{1-q^{(l)}} = \frac{1}{1-q},$$ the code designer is free to choose the individual $q^{(l)}$ values to balance the reliability performance across LDPC subcodes, e.g., by allocating smaller values~of~$q^{(l)}$~to subcodes with weaker reliability performance.
	
	
	\begin{table}[t]
		\captionsetup{justification=centering}
		\captionof{table}{ 
			\textsc{Optimized VN Side Degree Distributions of the Outer LDPC Code Stage.}}
		\label{tbl:opt:dd}
		\centering
		\begin{adjustbox}{angle=-270}
			{\scalebox{1}{
					\begin{tabular}{||r|c|r||}
						\hline\hline &&\\[-2.25ex]
						$\jmath$ & $\lambda_\jmath^{(1)}$ & $\rho_\jmath^{(1)}$\\&&\\[-2.25ex]
						\hline\hline&&\\[-2.5ex]
						$2$ & $0.18212$ & $0$\\&&\\[-3.0ex]
						$3$ & $0.00176$ & $0.266$\\&&\\[-3.0ex]
						$4$ & $0.34291$ & $0$\\&&\\[-3.0ex]						
						$5$ & $0.47049$ & $0$\\&&\\[-3.0ex]
						$6$ & $0.00073$ & $0$\\&&\\[-3.0ex]
						$7$ & $0.00027$ & $0$\\&&\\[-3.0ex]
						$8$ & $0.00015$ & $0$\\&&\\[-3.0ex]
						$9$ & $0.00011$ & $0$\\&&\\[-3.0ex]
						$10$ & $0.00008$ & $0$\\&&\\[-3.0ex]
						$11$ & $0.00007$ & $0.166$\\&&\\[-3.0ex]
						$12$ & $0.00006$ & $0$\\&&\\[-3.0ex]
						$13$ & $0.00005$ & $0.102$\\&&\\[-3.0ex]
						$14-19$ & $0.00004$ & $0.102$\\&&\\[-3.0ex]
						$20$ & $0.00003$ & $0.133$\\&&\\[-3.0ex]
						$21-32$ & $0.00003$ & $0$\\&&\\[-3.0ex]
						$33$ & $0.00003$ & $0.333$\\&&\\[-3.0ex]
						$34-44$ & $0.00003$ & $0$\\&&\\[-3.0ex]
						$45-55$ & $0.00002$ & $0$\\
						\multicolumn{3}{||c||}{}\\[-3ex]
						\hline\hline\multicolumn{3}{||c||}{}\\[-2.2ex]
						\multicolumn{3}{||c||}{$\Rout^{(1)}=0.5277$}\\
						\hline\hline
			\end{tabular}}}\vspace*{15pt}
		\end{adjustbox}
	\end{table}
	
	\begin{figure}
		\centering
		\includegraphics[scale=0.68]{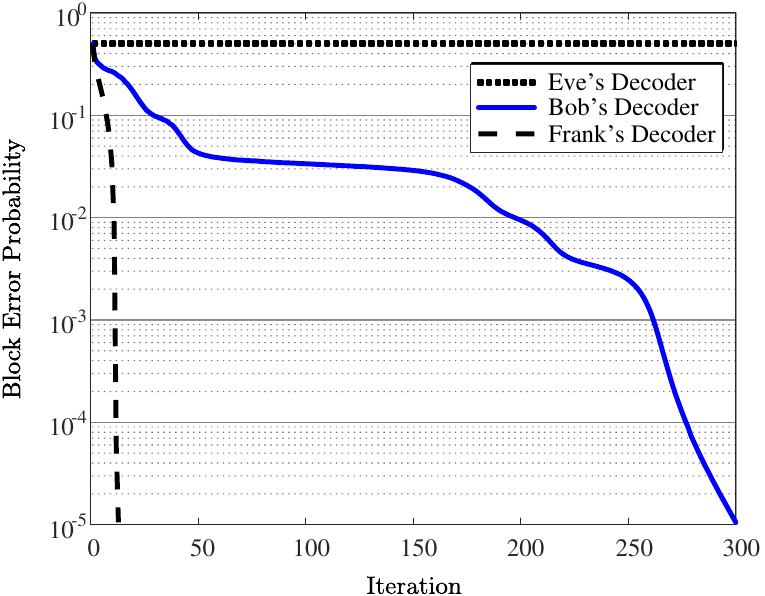}
		\caption{The average asymptotic performance of Bob's, Eve's, and Frank's decoders, as estimated using density evolution over $100$ randomly generated codewords of length $10^5$.\\ \mbox{}}
		\label{fig:DE_DECODER}
	\end{figure}

	\begin{example}
		Following Example~\ref{ex:innerate}, since the inner trellis code has $\Rin=1/4$ (i.e., $m=1$), a single LDPC code is employed as the outer stage. By using the superchannel in Table~\ref{tbl:trellis2} along with the optimized degree distributions of the outer LDPC code depicted in Table~\ref{tbl:opt:dd}, the weak secrecy criterion in~\eqref{equ:crit:sec} is satisfied at a secure rate of $\Rs=0.1100$, which is $0.42$~dB below the constrained secrecy capacity $\Rcsc^{(4)}=0.121$~$(\sfrac{\text{bits}}{\text{channel use}})$. According to~\eqref{equ:fxsecrate}, we fix the secure rate at $\Rs=0.1100$ by setting the fraction of the punctured VNs to $q=0.3056$. The optimized degree distributions in Table~\ref{tbl:opt:dd} have an \emph{unpunctured} rate $\Rout^{(1)}=0.5277$ and a \emph{punctured} rate $\Rout^{\prime}=\Rout^{\prime(1)}\defeq\Rout^{(1)}/(1-q)=0.7598$. Due to the upper bound in \eqref{equ:desR_opt}, $\Rd\defeq\Rin\cdot\Rout^{\prime}=0.1900$, and $$\lim_{N\rightarrow\infty}\frac{I(\vD^{mN};\vZ^{nN})}{nN}=0.080\quad (\sfrac{\text{bits}}{\text{channel use}}),$$ the rate of information leakage is upper bounded by $\check{\varepsilon}<10^{-5}$. We estimate the asymptotic error probabilities of sum-product decoders at Bob's, Eve's, and Frank's receivers by running the modified density evolution algorithm for $100$ randomly generated codewords of length $10^5$.\footnote{Since for channels with memory the concentration statement for the decoding error probability depends on the transmitted input sequence~\cite{1207365}, and the modified density evolution procedure involves an inherent Monte Carlo component, we estimate the asymptotic decoder performance at each iteration by averaging the asymptotic error probabilities over multiple input sequences.} The average error probabilities at each iteration are depicted in Fig.~\ref{fig:DE_DECODER}. As illustrated in the figure, Frank's block error probability vanishes, thereby validating the condition required by \eqref{equ:leakage_lstest} for the upper bound in \eqref{equ:desR_opt} to hold. Also, one should note that a high block error probability at Eve's decoder is not required by the secrecy criterion; we include Eve's curve merely to provide additional assurance that the proposed code remains secure under the decoding strategy matched to the considered system setup.
		\exampleend
	\end{example}
	
	\begin{figure*}
		\centering
		\includegraphics[scale=0.68]{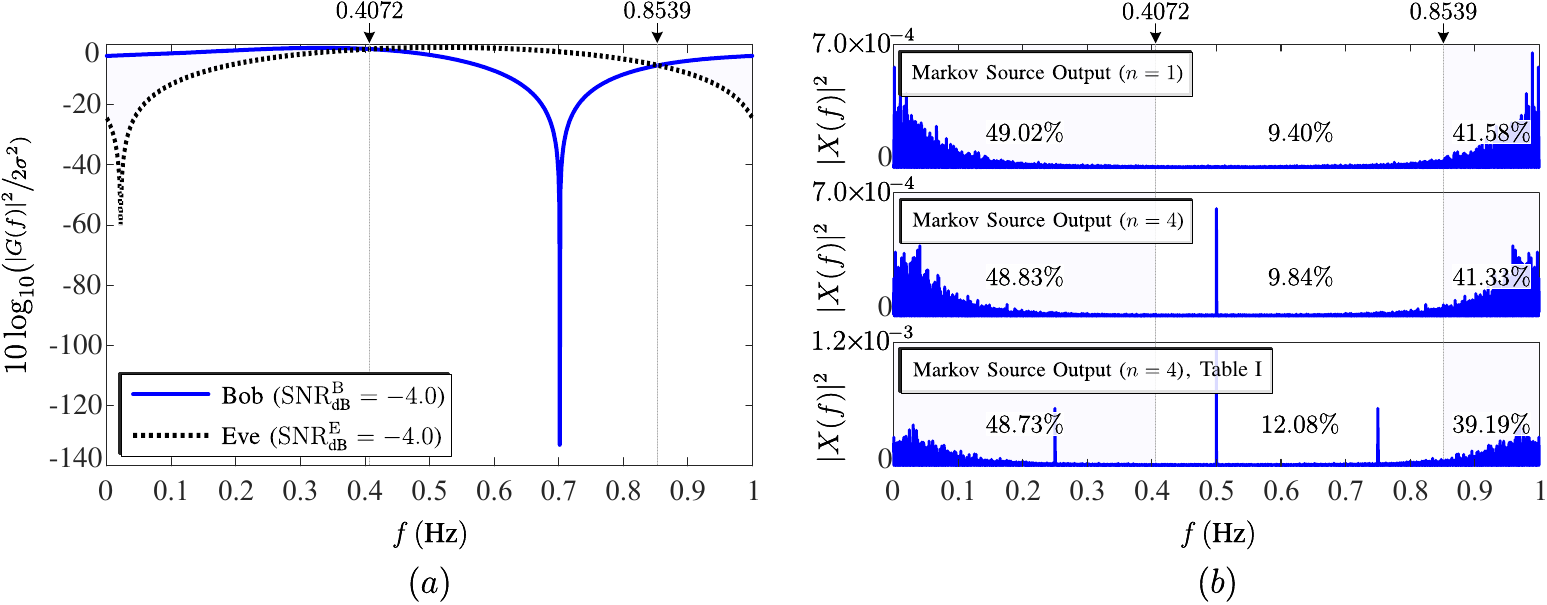}
		\caption{$(a)$: Gain-to-noise power spectrum ratios of Bob's and Eve's point-to-point channels in dB/Hz. $(b)$:~Power spectra of sequences produced by Markov sources optimized for the $1$-st order and $4$-th order E-FSMCs representing the considered ISI-WTC, and the power spectrum of a sequence generated by the Markov source specified in Table~\ref{tbl:trellis1}, satisfying the upper bound constraint in~\eqref{ex:rule1}.}
		\label{Fig.GNR}
	\end{figure*}
	
	\section{Finite Blocklength Regime}\label{sec::SIMFL}
	
	We first present remarks regarding the finite blocklength code construction procedure. Then, we analyze the spectral properties of the codewords generated by the designed trellis code. Finally, we evaluate the reliability and secrecy performance of the two-stage code in the finite blocklength regime.
	
	It should be noted that careful consideration is required when specifying the blocklengths $N^{(l)}$ ($1 \leq l \leq m$) to ensure that the output blocks of all \emph{punctured} LDPC subcodes have the same length $N'$. The following remark clarifies this point.
	\begin{remark}\label{rem:fixedN}
		For fixed values of $\Rs$, $\Rin$, and $\Rout'^{(l)}$ ($1\leq l\leq m$), we begin by choosing the value of $N'$, determining the length of the codewords transmitted over the channel ($nN'$). Next, we calculate $\ks\defeq nN'\cdot\Rs$ and $\kr\defeq nN'\cdot(\Rd-\Rs)$, where $\Rd\defeq\Rin\cdot\Rout^{\prime}$ and $\Rout'\defeq\frac{1}{m}\sum_{l=1}^{m}\Rout'^{(l)}$. As described in Section~\ref{subsec:Enc}, the length of the secret message block $\vms^{(l)}$ in the $l$-th subcode ($\defeq N^{(l)} - N'$) is calculated as $$\ks\cdot \frac{N^{(l)}\cdot\Rout^{(l)}}{\sum_{l=1}^{m}N^{(l)}\cdot\Rout^{(l)}}.$$ Using $N^{(l)}\cdot\Rout^{(l)} = N'\cdot\Rout^{\prime(l)}$ and $\sum_{l=1}^{m}N^{(l)}\cdot\Rout^{(l)} = \ks+\kr$, the blocklength of the $l$-th \emph{unpunctured} subcode is given by
		$$
		N^{(l)}
			= N' + \frac{\ks\cdot N'\cdot\Rout^{\prime(l)}}{\ks+\kr},
		$$ for all $1\leq l\leq m$.
		\remarkend
	\end{remark}
	
	The multi-stage decoder shown in Fig.~\ref{Fig:decoder} facilitates the code \emph{design} process. In practical implementations, however, the decoder consists of only two stages: a windowed BCJR decoder for the inner-stage trellis code and a sum-product decoder for the overall outer-stage LDPC code that encompasses all subcodes. In the following, we construct the outer code stage (comprising $m$ LDPC subcodes) using a single LDPC code that includes all $m$ subcodes, as detailed in \cite{1397934}.
	
	For $1\leq l\leq m$, let $\matr{H}^{(l)}=\big(\vect{h}_1^{(l)}\> \vect{h}_2^{(l)}\> \cdots\> \vect{h}_{N^{\!(l)}}^{(l)}\big)$ denote the parity-check matrix of the $l$-th LDPC subcode, where $\vect{h}_t^{(l)}$ represents the $t$-th column of $\matr{H}^{(l)}$. Let $N''\defeq\max_l N^{(l)}$. For every subcode $1\leq l\leq m$ such that $N^{(l)}<N''$, we set $\vect{h}_t^{(l)}\defeq\vect{0}$ for column indices $N^{(l)}< t\leq N''$; we then define
	\begin{equation*}
		\matrH_t \defeq
		\begin{pmatrix}
			\vect{h}_t^{(1)} & \vect{0} & \cdots & \vect{0}\\
			\vect{0} & \vect{h}_t^{(2)} & \cdots & \vect{0}\\
			\vdots	 & \vdots 			& \ddots & \vdots  \\
			\vect{0} & \vect{0} & \cdots & \vect{h}_t^{(m)}\\
		\end{pmatrix},\quad (\text{for $1\leq t\leq N''$}).
	\end{equation*}
	The parity-check matrix of the overall LDPC code is obtained by interleaving the columns of the subcode parity-check matrices as $(\matrH_1\> \matrH_2\> \cdots\> \matrH_{N''})$, and removing any all-zero columns. See also Fig.~\ref{Fig:ConcGraph} and refer to~\cite[Sec.~V-D]{1397934} for further details.
	
	\begin{remark}\label{rem:comp}
		In our simulations, instead of performing joint iterative decoding between the inner and outer code stages at every iteration, we apply BCJR decoding at the inner stage once every $5$--$10$ iterations of message passing within the outer stage. This scheduling improves the convergence characteristics of the turbo decoder for short blocklengths. It also provides substantial computational savings for large blocklengths, since one iteration of BCJR decoding for a trellis code with rate $m/n$ has a complexity of $\mathcal{O}(N'\cdot|\setS|\cdot 2^m)$. \remarkend
	\end{remark}
	
	\begin{figure}
		\centering
		\includegraphics[scale=0.7]{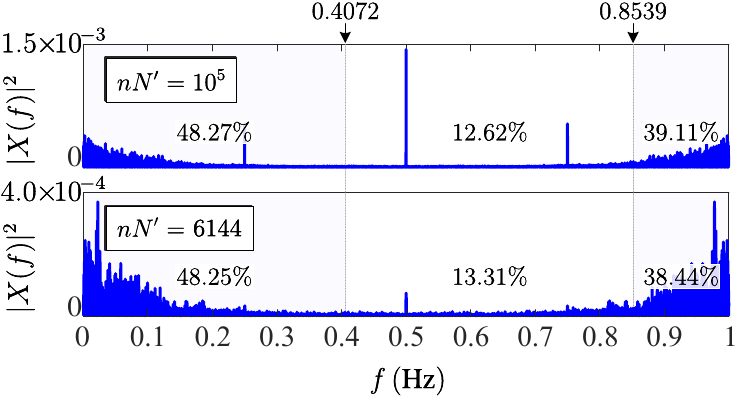}
		\caption{Power spectra of sample codewords generated by the inner trellis code specified in Table~\ref{tbl:trellis2} $(n=4)$ using LDPC codewords of lengths $N'\in\{25000, 1536\}$ as inputs.\\ \mbox{}}
		\label{Fig.GNR_code}
	\end{figure}
	
	\subsection{Spectral Properties}
	
	In the considered setup, Bob's channel has a lower point-to-point unconstrained capacity compared to Eve's channel, as detailed in Appendix~\ref{apx:simsen}. However, when examining the gain-to-noise power spectrum ratio $|G(f)|^2/(2\sigma^2)$ (see~\eqref{equ:Gf}) for both channels, we observe favorable frequencies for Bob where his channel has a higher gain-to-noise ratio than Eve's channel (Fig.~\ref{Fig.GNR}($a$)). By optimizing the input Markov sources, as demonstrated in~\cite{10068266}, these spectral discrepancies can be exploited to enable positive secure rates, as done in Example~\ref{ex:casestudy} (see Fig.~\ref{Fig:secrecy_gain}). Fig.~\ref{Fig.GNR}($b$) presents the normalized power spectra of sequences of length $10^5$ generated by Markov sources optimized for the first-order and fourth-order E-FSMCs representing the considered ISI-WTC. These Markov sources effectively allocate over $90\%$ of the available power to frequency ranges where Bob's gain-to-noise ratio surpasses Eve's. Notice that the Markov sources shown in the upper- and~middle-right subplots violate the upper bound constraint in~\eqref{ex:rule1}. In contrast, the lower-right subplot depicts the normalized power spectrum of a sequence generated by the optimized Markov source in~Table~\ref{tbl:trellis1}, having $n=4$ and satisfying~\eqref{ex:rule1}.
	
	As shown in the middle and lower subplots of Fig.~\ref{Fig.GNR}($b$), unintended impulses appear in the spectra of sequences generated by optimized Markov sources for $n=4$ (and $n>1$ in general). In an E-FSMC of order $n$, the branch process $\{\bE(t)\}_{t\in\Z}$ and, consequently, the block-output process $\{\vx^n(t)\}_{t\in\Z}$ are controlled by the input Markov source. However, the Markov source does not influence the internal structure within each block,~$\big(x(n(t-1)+1),\ldots,x(nt)\big)$. We conjecture that these observed impulses originate from this internal structure, which fails to conform to the desired spectral characteristics of the optimized Markov source. Note that the two additional impulses observed in the lower subplot of Fig.~\ref{Fig.GNR}($b$) result from a higher likelihood of the patterns $(0,0,1,0)$, $(0,1,0,0)$, $(0,1,1,0)$, $(1,0,0,1)$, $(1,0,1,1)$, and $(1,1,0,1)$ in the output. Indeed, the increased occurrence of these patterns arises from adhering to the upper bound constraint in~\eqref{ex:rule1}, which is violated by the Markov source used in the middle subplot.
	
	Fig.~\ref{Fig.GNR_code} illustrates the normalized power spectra of sample codewords generated by the inner trellis code specified in Table~\ref{tbl:trellis2}, using LDPC codewords of lengths $25000$ and $1536$ as inputs.\footnote{Choosing $N' = 1536$, resulting in a blocklength of $nN' = 6144$, is motivated by the LTE standard in 3GPP~TS~36.212~\cite{3GPP_TS_36_212}, where this blocklength is specified for turbo-coded transmissions. Moreover, since $6144$ is below the maximum LDPC blocklength of $8448$ defined in the 5G New Radio (NR) specifications in 3GPP~TS~38.212~\cite{3GPP_TS_38_212}, it remains a practically relevant size within the processing regimes of both standards.} As illustrated, even for short blocklengths, the codewords produced by the designed trellis code effectively replicate the spectral characteristics of sequences generated by the optimized Markov source, concentrating approximately $88\%$ of the available power in frequency ranges that favor Bob's channel over Eve's in terms of the gain-to-noise ratio.

	\subsection{Performance}
	
	We first evaluate the performance of the proposed code using BPSK-modulated sequences of blocklength $nN' = 10^5$. At a secure rate of $\Rs=0.0401$ ($q=0.1382$), the outer LDPC code is constructed with a punctured blocklength of $N' = 25000$ using the optimized degree distributions from Table~\ref{tbl:dd:finite} ($\Rout=0.4143$, $\Rout'=0.4808$). Combined with a trellis code of rate $\Rin=1/4$ (Table~\ref{tbl:trellis2}), this setup results in an overall design rate of $\Rd=0.1202$. Regarding the secrecy efficiency, the rate of information leakage is upper bounded by ${\varepsilon}_{nN'}<10^{-5}$, according to \eqref{equ:leakage_num} in Lemma~\ref{lem:leakage_bd}, $\epsE_{nN'}<10^{-5}$, $\Rr=0.0801$, and $$\lim_{N\rightarrow\infty}\frac{I(\vD^{mN};\vZ^{nN})}{nN}=0.080\quad(\sfrac{\text{bits}}{\text{channel use}}).$$ Regarding the reliability efficiency, the designed code achieves a bit-error rate~(BER) $\epsB_{nN'}\!<\!10^{-5}$ for $\SNRdBB=-4$ at Bob's decoder. In addition, for a shorter blocklength of $nN'=6144$, the proposed code demonstrates almost the same reliability and secrecy performance at a lower secure rate of $\Rs=0.0309$ ($q=0.1100$), using the optimized degree distributions of Table~\ref{tbl:dd:finite} ($\Rout=0.3955$, $\Rout'=0.4444$) and the same trellis code of rate $\Rin=1/4$ in Table~\ref{tbl:trellis2}.
	
	\begin{figure*}
		\centering
		\includegraphics[scale=0.76]{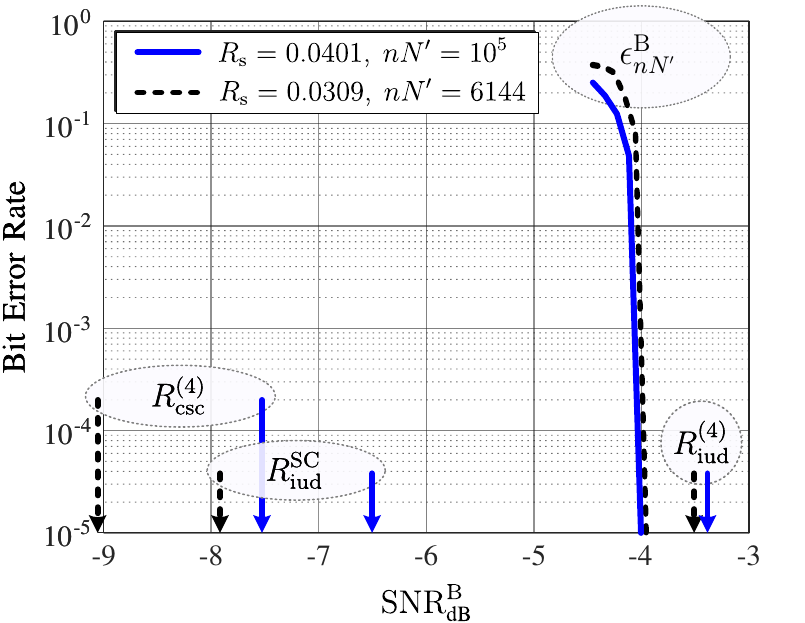}
		\caption{The bit error rate of the constructed code $\epsB_{nN'}$, the i.u.d.\ secure rate $\Riud^{(4)}$, the constrained secrecy capacity $\Rcsc^{(4)}$ of the $4$-th order E-FSMC modeling the ISI-WTC in Example~{\ref{ex:casestudy}}, and the i.u.d.\ secure rate of the superchannel $\RSC$ in Table~\ref{tbl:trellis2} at $\SNRdBE = -4.0$.\vspace{10pt}}
		\label{fig:BER}
	\end{figure*}
	
	The BER simulation curve is shown in Fig.~\ref{fig:BER}. For comparison, this figure also includes the i.u.d.\ secure rate $\Riud^{(4)}$, the constrained secrecy capacity $\Rcsc^{(4)}$ of the fourth-order E-FSMC modeling the ISI-WTC in Example~\ref{ex:casestudy}, and the i.u.d.\ secure rate of the superchannel $\RSC$ in Table~\ref{tbl:trellis2}, all evaluated~at~$\SNRdBE = -4.0$. As illustrated, the designed concatenated code achieves a reliability threshold that is $2.5$~dB away from $\RSC$, $3.6$~dB away from $\Rcsc^{(4)}$, and surpasses $\Riud^{(4)}$ by $0.6$~dB for a blocklength of $nN' = 10^5$. In addition, the code simulation results show that for $nN' = 6144$, a BER of $10^{-5}$ is achieved at an SNR that surpasses $\Riud^{(4)}$ by $0.4$~dB.
	(All codes and simulation files are publicly available online~\cite{ISIWTC_CD}.)
	
	The gap to the constrained secrecy capacity, even for large blocklengths, is primarily due to the puncturing strategy used to enforce secrecy. Near-optimal reliability performance under puncturing requires carefully optimized puncturing patterns, as in~\cite{Ha2004RCP}. In contrast, our design directly punctures the LDPC codeword bits carrying the secret message (see Section~\ref{subsec:Enc}). This design choice incurs a reliability loss to enforce security.
	
	\section{Conclusion}\label{sec::CONC}
	
	In this paper, we have first demonstrated the inherent limitation of single-stage LDPC codes in exceeding the i.u.d.\ secure rates over ISI-WTCs. To address this limitation, we have then incorporated an inner-stage trellis code and defined the concept of a superchannel as a joint model for trellis codes at the input of the ISI-WTCs. By appropriately designing the inner-stage trellis code, the i.u.d.\ sequence produced by the outer-stage LDPC code is transformed into a Markov process that approaches the constrained secrecy capacity. Consequently, the i.u.d.\ secure rate of the resulting superchannel asymptotically approaches the constrained secrecy capacity of the underlying ISI-WTC. This makes the constrained secrecy capacity achievable by the proposed two-stage code. Finally, numerical results show that, even in finite blocklength regimes, the proposed concatenated code can exceed the i.u.d.\ secure rates and leverage spectral differences between Bob's and Eve's channels to achieve positive secure rates despite Eve's channel having a higher point-to-point capacity than Bob's.
	
	\begin{table}[t]
	\captionsetup{justification=centering}
	\captionof{table}{ 
		\textsc{Optimized VN Side Degree Distributions of the Outer LDPC Code Stage for $\Rs=0.0401$ and $\Rs=0.0309$.}}
	\label{tbl:dd:finite}
	\centering
	\begin{adjustbox}{angle=-270}
		{\scalebox{.9}{
				\begin{tabular}{||r|c|c|r||}
					\hline\hline &&\\[-2.25ex]
					& $\Rout=$ & $\Rout=$ &\\
					& $0.4143$ & $0.3955$ &\\
					$\jmath$ & $\lambda_\jmath^{(1)}$ & $\lambda_\jmath^{(1)}$ & $\rho_\jmath^{(1)}$\\&&\\[-2.25ex]
					\hline\hline&&\\[-2.5ex]
					$2$ & $0.483$ & $0.461$ & $0$\\&&\\[-3.0ex]
					$3$ & $0.155$ & $0.153$ & $0.406$\\&&\\[-3.0ex]
					$4$ & $0.012$ & $0.012$ & $0$\\&&\\[-3.0ex]						
					$5$ & $0.010$ & $0.010$ & $0.070$\\&&\\[-3.0ex]
					$6$ & $0.009$ & $0.009$ & $0$\\&&\\[-3.0ex]
					$7$ & $0.009$ & $0.010$ & $0.033$\\&&\\[-3.0ex]
					$8$ & $0.011$ & $0.012$ & $0$\\&&\\[-3.0ex]
					$9$ & $0.014$ & $0.017$ & $0$\\&&\\[-3.0ex]
					$10$ & $0.015$ & $0.018$ & $0$\\&&\\[-3.0ex]
					$11$ & $0.012$ & $0.014$ & $0.122$\\&&\\[-3.0ex]
					$12$ & $0.009$ & $0.011$ & $0$\\&&\\[-3.0ex]
					$13$ & $0.008$ & $0.009$ & $0.103$\\&&\\[-3.0ex]
					$14$ & $0.007$ & $0.008$ & $0$\\&&\\[-3.0ex]
					$15-19$ & $0.007$ & $0.007$ & $0$\\&&\\[-3.0ex]
					$20$ & $0.006$ & $0.006$ & $0.133$\\&&\\[-3.0ex]
					$21-23$ & $0.006$ & $0.006$ & $0$\\&&\\[-3.0ex]
					$24$ & $0.006$ & $0.006$ & $0.133$\\&&\\[-3.0ex]
					$25-52$ & $0.006$ & $0.006$ & $0$\\&&\\[-3.0ex]
					$53-55$ & $0.007$ & $0.007$ & $0$\\
					\hline\hline	
		\end{tabular}}}
	\end{adjustbox}\\[6pt] \mbox{}
\end{table}
	
	\section{Acknowledgment}
	The authors are grateful to Prof. Pascal O. Vontobel for his invaluable comments on the earlier versions of this manuscript.
	
	
	\appendices
	
	\section{Simulation Scenario}\label{apx:simsen}
	
	In phased array antennas, when the steering beam angle deviates from the reference beam angle, different antenna elements experience different propagation-path lengths. This produces element-dependent time delays in the signals at both the transmitter and the receiver (see Fig.~\ref{Fig:phased_array}). When these signals are combined using a Wilkinson combiner, the resulting temporal misalignment causes signal overlap and~interference.
	
	Let $\lamc \in \R$ denote the wavelength of the carrier signal. We consider a phased array, as described in~\cite{9324767}, consisting of a $16\times16$ antenna array with elements spaced $d_{\mathrm{el}} \approx 0.5\lamc$ apart. Our model is based on an elevation-plane scan angle of $-\pi/3 \leq \theta_0 \leq \pi/3$, with $n_{\mathrm{el}} = 16$ channel branches to simulate the varying delays across the antenna elements in the elevation dimension (Fig.~\ref{Fig:phased_array}). For $0 \le r \le n_{\mathrm{el}} - 1$, let $\tau_r \in \R$ denote the time delay associated with the $r$-th antenna element in the elevation plane. We then have
	\begin{equation}\label{equ:PAISIbd}
		\tau_r = \frac{r\cdot d_{\mathrm{el}}\cdot\sin(\theta_0)}{c}, \qquad 0 \le r \le n_{\mathrm{el}} - 1,
	\end{equation}
	where $d_{\mathrm{el}}$ ($=5.3\,\text{mm}$~\cite{9324767}) is the spacing between antenna elements, and $c$ is the speed of light.
	
	\begin{figure}
		\centering
		\includegraphics[scale=1.3]{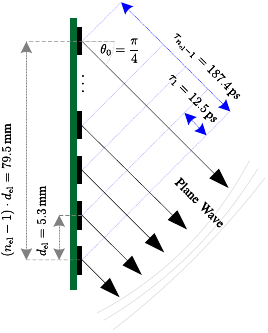}
		\caption{Time delays between signals transmitted by the antenna elements of a $16\times 16$ phased array in the elevation plane.\\ \mbox{}}\label{Fig:phased_array}
	\end{figure}

	\subsection{Estimating the Downlink Channel}
	
	Let ${X}(\tc), {Y}(\tc), {N}(\tc)\in\C$ (with $\tc\in\R$) denote continuous-time random signals\footnote{The variable $\tc \in \R$ represents continuous time, distinguishing it from the discrete-time variable $t \in \Z$ used throughout the paper. With a slight abuse of notation, we consider a signal to be continuous-time when its argument is~continuous~$(\in\R)$ and discrete-time when its argument is discrete~$(\in\Z)$.} corresponding to the channel input, the channel output, and additive noise, respectively. When local oscillators at the transmitter and receiver terminals are synchronized, the received signal is expressed as
	$$
	Y(\tc) \defeq \sum_{r=0}^{n_{\mathrm{el}}-1} |g_{\mathrm{el}}| e^{-i \frac{2\pi c }{\lamc} \tau_r} \cdot X(\tc-\tau_r) + N(\tc),
	$$
	where $i$ represents the imaginary unit and $|g_{\mathrm{el}}|$ denotes the element gain of the phased array. Fig.~\ref{fig:power_delay} illustrates a typical power-delay profile of a multipath channel, resulting from the considered phased array in Fig.~\ref{Fig:phased_array}, assuming a normalized element gain $|g_{\mathrm{el}}|=1$, a carrier frequency of $28\,\text{GHz}$ ($\lamc=10.71\,\text{mm}$), and a steering beam angle of $\theta_0={\pi}/{4}$, leading to $\tau_r = r\cdot 12.5~\text{ps}$ and $\tau_{n_{\mathrm{el}}-1} = 187.4~\text{ps}$~according~to~\eqref{equ:PAISIbd}.
	
	The continuous-time signal $Y(\tc)$ is sampled at the receiver with symbol interval $T\in\R$ and represented by $Y(t)\defeq Y(t\cdot T)$ for $t\in\Z$. In phased arrays, the interference among delayed signals is mitigated using linear multi-tap equalizers~\cite{9324767}. Nevertheless, this approach is effective as long as ${T^{-1}\cdot\tau_{n_\mathrm{el}-1}}<\eta$, where $0<\eta<1$ is the equalizer's operating threshold, which is inversely proportional to the noise level.\footnote{The scaled version of the term $T^{-1}\cdot\tau_{n_\mathrm{el}-1}$ is referred to as the ISI percentage, denoted by $\text{ISI}_\%\defeq 100\cdot T^{-1}\cdot\tau_{n_\mathrm{el}-1}$ in the literature.} Exceeding the ultimate bound $T^{-1}\cdot\tau_{n_\mathrm{el}-1}<1$ results in noise amplification and residual ISI in multi-tap equalizers~\cite[Ch.~9]{proakis2008digital}. Due to~\eqref{equ:PAISIbd}, this upper bound imposes the constraint $(n_\mathrm{el}-1)\cdot d_{\mathrm{el}}\cdot\sin(\theta_0) \cdot T^{-1} < \eta\cdot c$, limiting the number of array elements~($n_\mathrm{el}$), the size of the phased array ($\propto d_{\mathrm{el}}$), the range of the supported scan angles~($\theta_0$), \text{and the symbol rate of the communication system~($T^{-1}$).}

	\begin{figure}
		\includegraphics[scale=0.62]{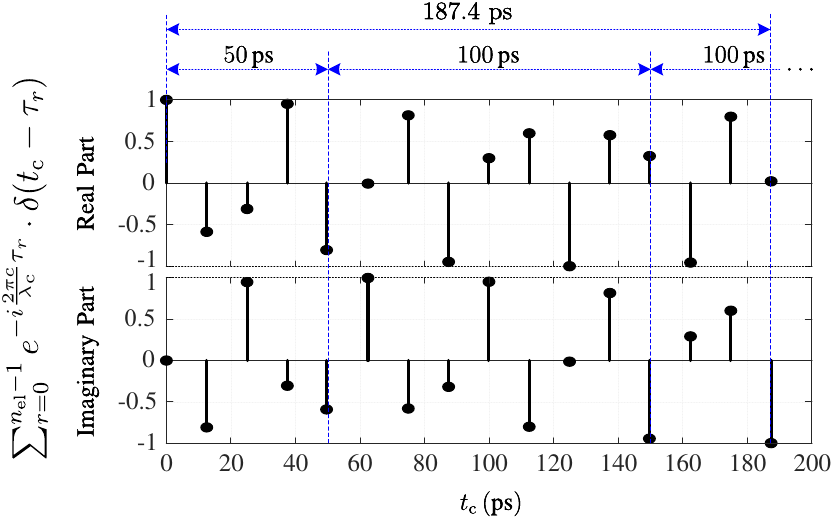}
		\caption{Power-delay profile of the multipath channel associated with the phased array in Fig.~\ref{Fig:phased_array}, for a carrier frequency of $\sfrac{c}{\lamc}=28$~GHz.\\ \mbox{}}
		\label{fig:power_delay}
	\end{figure}
	
	\begin{table}
		\captionsetup{justification=centering}
		\captionof{table}{ 
			\textsc{Unnormalized CIR of the ISI Channel Corresponding to the Power-Delay Profile in Fig.~\ref{fig:power_delay}, for a Symbol Interval of $T=100\,\mathrm{ps}$.}}
		\label{tbl:power_delay}
		\centering
		{\scalebox{1}{
				\begin{tabular}{crrr}
					\hline \\[-10pt]
					$\displaystyle r$ \!\! & \!\!\! Period ($\displaystyle$ps) 	   
					& \multicolumn{2}{c}{Unnormalized $\displaystyle {g}_r$ } \\[0pt]
					\hline\\[-10pt]
					0 	& $0-50$ 	& $0.2482$  \!\!\!\! & \!\!\!$-\>0.7569i$ \\[0pt] 
					1 	& $50-150$  & $0.6554$	\!\!\!\! & \!\!\!$+\>0.1116i$ \\[0pt] 
					2 	& $150-250$ & $-0.1383$ \!\!\!\! & \!\!\!$-\>0.1046i$ \\[0pt] 
					\hline\vspace*{16pt}
		\end{tabular}}}
	\end{table}
	
	To verify the performance of the proposed coding scheme, we consider BPSK signaling with a symbol rate of $10\,\text{Gbaud}$ ($T=100\,\text{ps}$ and ${T}^{-1}\cdot\tau_{n_\mathrm{el}-1}=1.874$) under a high level of Gaussian noise ($\SNRdB=-4.0$). As shown in Fig.~\ref{fig:power_delay}, the unnormalized ISI tap coefficients are captured by sampling at $\frac{T}{2} + t\cdot T$ for $t\in\{0,1,2\}$ (see Table~\ref{tbl:power_delay}). Accordingly, the normalized CIR\footnote{A normalized CIR $\vect{g} \defeq (g(0) \ldots g(\nu))$ has to satisfy $\sum_{r=0}^{\nu} |g(r)|^2 = 1$.} of the downlink channel from the considered phased array, obtained by sampling the output of a filter matched to the shaping pulse at the receiver, becomes
	\begin{equation*}
		\vect{g} = (0.2360 \!-\! 0.7195i, 0.6230 \!+\! 0.1061i, -0.1314 \!-\! 0.0994i).
	\end{equation*}
	
	
	\subsection{Eavesdropping Scenario}
	
	We consider a setup where Alice employs a phased array (as described in the previous section) to transmit a signal to Bob's receiver, located at an angle of $\theta^{\userB}=\pi/4$ relative to the array's boresight. The phased array steers its main lobe toward $\theta_0=\pi/4$ to maximize the signal strength at Bob's location. Meanwhile, Eve, positioned at an angle of $\theta^{\userE}=\pi/6$, intercepts the transmitted signals either through the side lobes or, if the main lobe is sufficiently broad to cover $\theta^{\userE}=\pi/6$, through the main lobe itself.\footnote{As noted in the introduction, transmitters with stringent power constraints (e.g., aerial or low-power ground units) often exhibit wider main lobes due to antenna size limitations. This increases their susceptibility to eavesdropping.} In this scenario, Eve benefits from a smaller angular deviation from the array boresight, resulting in smaller delays between antenna elements. Bob, on the other hand, receives the full strength of the steered main lobe, typically resulting in a higher SNR. To create a more challenging adversarial scenario, we assume $\SNRdBB=\SNRdBE=-4.0$, thereby keeping the setting advantageous for Eve. Using the calculations from the previous section, we have
	\begin{align*}
		&\vgB \!\!=\! (0.2360 \!-\! 0.7195i,  0.6230 \!+\! 0.1061i, -0.1314 \!-\! 0.0994i),\alabel{eq:BOB-ISI}\\
		&\vgE \!\!=\! (0.5211 \!-\! 0.4792i, -0.5791 \!+\! 0.4043i),\alabel{eq:EVE-ISI}
	\end{align*}
	for the considered setup.
	
	As described in~\cite{10068266}, to verify Eve's channel advantage over Bob's channel, we employ the well-known water-pouring formulation to analyze the capacities of the point-to-point channels $\bigl( \vgB, \sigmaB^2 \bigr)$ and $\bigl( \vgE, \sigmaE^2 \bigr)$. Let $\Es$ denote the average energy per input symbol, and let $W\defeq\frac{1}{2T}$ represent the bandwidth of a perfect low-pass filter with sampling at the Nyquist frequency $1/T$ at the receiver. Then, the unconstrained capacity (subject only to an average-energy constraint) of an ISI channel, described by $(\vect{g},\sigma^2)\in\C^{\nu+1}\times\R$, is given by
	$$
	C(\vect{g},\sigma^2,W) = \frac{1}{2} \cdot \int\limits_{-\infty}^{\infty} \log^+ \left(
	\frac{\alpha} {2\sigma^2/|G(f)|^2} \right) \, \mathrm{d} f,
	$$
	where
	\begin{align*}
		G(f)
		&= \begin{cases}
			{\displaystyle 
				\frac{\sum_{t=0}^{\nu} g({t}) e^{-i 2\pi t f T}}
				{\sqrt{\sum_{t=0}^{\nu} |g({t})|^2}}}
			& \text{(if $|f| \leq W$)} \\
			0 
			& \text{(otherwise)}
		\end{cases},\alabel{equ:Gf}
	\end{align*}
	and where $\alpha > 0$ is chosen to satisfy
	$$
	\Es = \int\limits_{-\infty}^{\infty}
	{\left(
		\alpha  -  \frac{2\sigma^2}{|G(f)|^2}
		\right)^+}
	\mathrm{d} f.
	$$
	As illustrated in Fig.~\ref{Fig:unconst_cap}, for a normalized energy constraint $\Es=1$ and $\SNRdBB=\SNRdBE=-4.0$, Eve's channel~\eqref{eq:EVE-ISI} has a higher unconstrained capacity than Bob's channel~\eqref{eq:BOB-ISI} for sufficiently large bandwidth, confirming Eve's advantage in channel quality.\footnote{Since the unconstrained capacity given by the water-pouring formulation is achieved by Gaussian inputs, caution should be exercised when leveraging insights from the presented results, which are used solely to compare the point-to-point channels in the ISI-WTC setup when used with BPSK inputs.}
	
	\begin{figure}
		\centering
		\includegraphics[scale=0.72]{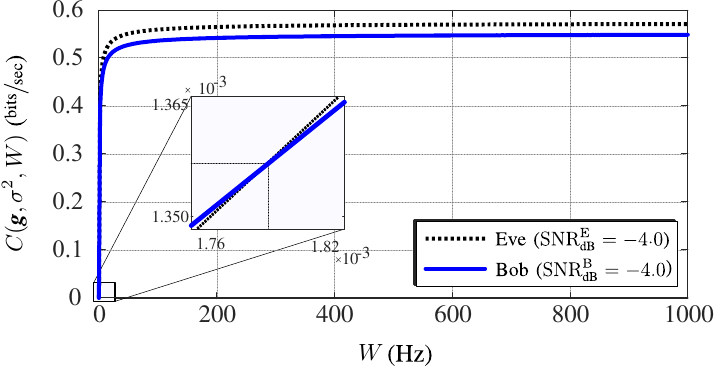}
		\caption{Unconstrained capacities of Bob's and Eve's point-to-point channels in $\sfrac{\mathrm{bits}\!}{\!\mathrm{sec}}$ with normalized average-energy constraint \text{$\Es\! =\! 1$}.\\ \mbox{}}\label{Fig:unconst_cap}
	\end{figure}
	
	\section{Proof of Proposition~\ref{prop:secap}}\label{apx:secap}
	
	We begin by introducing an input setup under which ISI-WTCs can be equivalently represented as vector-input/vector-output stationary memoryless wiretap channels~\cite{8017512}. This setting yields the achievability result via~\cite[Thm.~3]{8017512}. Finally, the optimality of the achievable secure rate is established through a converse proof, characterizing the secrecy capacity.
	
	\subsection{Achievability}
	The input process of the ISI-WTC $\{X(t)\}_{t\in\Z}$ is taken to be blockwise i.i.d., consisting of length-$n$ blocks $\{\vX^n(t)\}_{t\in\Z}$. To prevent interference between consecutive output blocks, for some arbitrary $\nu\geq\max\{\mB,\mE\}$ with $\nu<n$ and all $t\in\Z$, we introduce a guard interval by fixing $X(nt-\nu+1)\defeq 0,\ldots, X(nt)\defeq 0$, while allowing $X(n(t-1)+1),\ldots,X(nt-\nu)$ to follow an arbitrary identical distribution across blocks.
	
	Since the CIR vectors $\vect{g}_{\mathrm{B}}$ and $\vect{g}_{\mathrm{E}}$ remain constant throughout the transmission, the channel transition probabilities factorize identically across blocks as
	\begin{align}
		p_{\vY^{nN}|\vX^{nN}}\big(\vy^{nN}\big|\vx^{nN}\big)
		&=\prod_{t=1}^{N}
		p_{\vY^n(t)|\vX^n(t)}
		\big(\vy^n(t)\big|\vx^n(t)\big),
		\label{equ:blockfact_B}\\
		p_{\vZ^{nN}|\vX^{nN}}\big(\vz^{nN}\big|\vx^{nN}\big)
		&=\prod_{t=1}^{N}
		p_{\vZ^n(t)|\vX^n(t)}
		\big(\vz^n(t)\big|\vx^n(t)\big).
		\label{equ:blockfact_E}
	\end{align}
	By treating each block of $n$ symbols as a single channel use, the factorizations in \eqref{equ:blockfact_B} and \eqref{equ:blockfact_E} correspond to a standard stationary memoryless wiretap channel under the vector-input $\{\vX^n(t)\}_{t\in\Z},$ vector-output $\{\vY^n(t),\vZ^n(t)\}_{t\in\Z}$ representation.\footnote{Due to the imposed guard intervals, the inter-block memory is completely eliminated, making the vector channel memoryless across transmission blocks.}
	
	Let $D(t)\in\set{D}$ be an auxiliary random variable satisfying the Markov chains $D(t)\to\vX^n(t)\to\vY^n(t)$ and $D(t)\to\vX^n(t)\to\vZ^n(t)$ for all $t\in\Z$. These Markov chains induce an auxiliary concatenated channel at the input of the ISI-WTC, specified~by
	\begin{equation*}
		p_{\vX^{nN}|\vect{D}^{N}}\big(\vx^{nN}\big|\vect{d}^{N}\big)
		= \prod_{t=1}^{N} p_{\vX^n(t)|D(t)}\big(\vx^n(t)\big|d(t)\big).
	\end{equation*}
	The encoder for the capacity-achieving code, defined as a generalized version of the one-time stochastic encoder from~\cite{8017512}, is obtained by concatenating two components: (i)~a nested-code encoder for wiretap channels $f_N\!:\setM_\mathrm{s}\to\set{D}^N,$ and (ii)~a component-wise stochastic encoder $h_n\!:\set{D}\to\setX^n$, operating according to the stochastic map $p_{\vX^n|D}$. The resulting concatenated mapping yields the overall encoder $\phi_{nN}\!:\setM_\mathrm{s}\to\setX^{nN}$, defined as
	$$
	\phi_{nN}(\ms)
	\defeq\Big(h_n\big(d(t)\big)\Big)_{t=1}^N,
	$$
	where $\vect{d}^{N}\defeq f_N(\ms)$. See Fig.~\ref{Fig:onetime}.
	
	By~\cite[Thm.~3]{8017512}, all secure rates $R_{\mathrm{s}}$ satisfying
	\begin{equation}\label{equ:app:achieve}
		\Rs <
		\frac{1}{n}
		\Big(I(D;\vY^n) - I(D;\vZ^n)\Big)^+\!\!,
	\end{equation}
	are achievable over the equivalent vector-input/vector-output stationary memoryless wiretap channel specified by \eqref{equ:blockfact_B} and \eqref{equ:blockfact_E}. Importantly, the framework in~\cite[Thm.~3]{8017512} characterizes the secrecy capacity of stationary wiretap channels governed by a random state sequence available strictly as one-time state information. In our considered setting, the CIR vectors $\vect{g}_{\mathrm{B}}$ and $\vect{g}_{\mathrm{E}}$ are constant and deterministic. Consequently, the external channel state is fixed across all transmission blocks (i.e., CIRs are distributed by $\delta(\vect{g}-\vect{g}_{\mathrm{B}})$ and $\delta(\vect{g}-\vect{g}_{\mathrm{E}})$), trivially satisfying the one-time state information condition without requiring dynamic state tracking. This makes our static CIR setup a~simplified special case of the model considered in~\cite{8017512}.
	
	\begin{figure}
		\centering
		\includegraphics[scale=0.86]{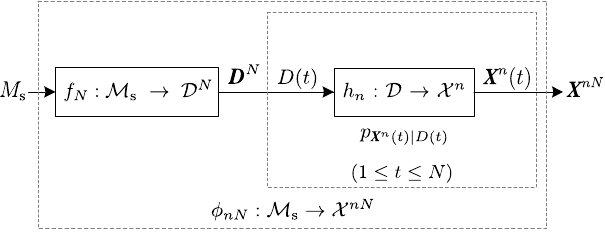}
		\caption{Theoretical encoder.\\ \mbox{}}\label{Fig:onetime}
	\end{figure}
	
	\subsection{Converse}
	
	In the converse, we verify the optimality of the achievable secure rate in \eqref{equ:app:achieve} by considering arbitrary wiretap codes over the ISI-WTC. Let $R_\mathrm{s}$ be an achievable secure rate. By definition, there exists a sequence of encoders $\phi_{nN}$ and decoders $\psi_{nN}$, satisfying the reliability criterion $\lim_{N\to\infty}\epsB_{nN}=0$ in~\eqref{equ:crit:sec} and the strong secrecy criterion $\lim_{N\to\infty} nN\cdot\varepsilon_{nN}=0$. Any achievable secure rate must satisfy $$\liminf_{N\to\infty}\frac{1}{nN}\log|\setM_\mathrm{s}| \ge R_\mathrm{s}.$$ Recalling $$\epsB_{nN}\defeq\Prob\big({M}_\mathrm{s}\neq\psi_{nN}(\vY^{nN})\big),$$ and $$nN\cdot\varepsilon_{nN} = I(M_\mathrm{s};\vZ^{nN}),$$ applying Fano's inequality, the secure rate is bounded by
	\begin{equation*}
		nN R_\mathrm{s} \le I(M_\mathrm{s};\vY^{nN}) - I(M_\mathrm{s};\vZ^{nN}) + nN(\epsB_{nN}+\varepsilon_{nN}) + 1.
	\end{equation*}
	
	Applying the chain rule with the blockwise histories $\underline{\vY}^{t-1}\triangleq(\vY^n(1),\ldots,\vY^n(t-1))$ and $\underline{\vZ}_{t+1}^N\triangleq(\vZ^n(t+1),\ldots,\vZ^n(N))$, and defining $F(t)\triangleq(\underline{\vY}^{t-1},\underline{\vZ}_{t+1}^N)$, yields
	\begin{align*}
		&nN R_\mathrm{s} \le \sum_{t=1}^N \Big(I(M_\mathrm{s};\vY^n(t)|\underline{\vY}^{t-1}) - I(M_\mathrm{s};\vZ^n(t)|\underline{\vZ}_{t+1}^N)\Big)\\[-4pt]
		&\qquad \qquad + nN(\epsB_{nN}+\varepsilon_{nN}) + 1 \\
		&\stackrel{(a)}{=} \sum_{t=1}^N \Big(\!I(M_\mathrm{s}, \underline{\vZ}_{t+1}^N;\vY^n(t)|\underline{\vY}^{t-1}) \! - \! I(M_\mathrm{s}, \underline{\vY}^{t-1};\vZ^n(t)|\underline{\vZ}_{t+1}^N)\!\Big) \\[-4pt]
		&\qquad \qquad + nN(\epsB_{nN}+\varepsilon_{nN}) + 1 \\
		&\stackrel{(b)}{=} \sum_{t=1}^N \Big(I(M_\mathrm{s};\vY^n(t)|F(t)) - I(M_\mathrm{s};\vZ^n(t)|F(t))\Big)\\[-4pt]
		&\qquad \qquad + nN(\epsB_{nN}+\varepsilon_{nN}) + 1 \\
		&\stackrel{(c)}{=}\sum_{t=1}^N \Big(I(D(t);\vY^n(t)|F(t)) - I(D(t);\vZ^n(t)|F(t))\Big)\\[-4pt]
		&\qquad \qquad + nN(\epsB_{nN}+\varepsilon_{nN}) + 1,
	\end{align*}
	where $(a)$ and $(b)$ follow from the Csisz\'{a}r sum identity and $(c)$ from defining the auxiliary random variable $D(t)\!\triangleq\!(M_\mathrm{s},\! F(t))$.
	
	To establish a single-block expression, we introduce a time-sharing variable $T$ uniformly distributed over $\{1,\ldots,N\}$ and independent of all other variables. We define the aggregated variables $F \triangleq (F(T),T)$, $D \triangleq (D(T),T)$, $\vX^n \triangleq \vX^n(T)$, $\vY^n \triangleq \vY^n(T)$, and $\vZ^n \triangleq \vZ^n(T)$. As $N\to\infty$, we obtain
	$$
	R_\mathrm{s} \leq \frac{1}{n}\Big(I(D; \vY^n | F) - I(D; \vZ^n | F)\Big)^+\!\!.
	$$
	
	To isolate the outputs in block $T$ from adjacent blocks in $F$, we define the boundary state $\vSF \triangleq \big(\SF\big(n(T-1)\big), \SF\big(nT\big)\big)$. This state captures the causal influence of past blocks and noncausal influence of future blocks on the current output blocks $(\vY^n,\vZ^n)$. Together with the codebook-induced dependencies captured by $\vX^n$, this implies the Markov chain $F \to (\vX^n, \vSF) \to (\vY^n, \vZ^n)$. Accordingly, we have
	\begin{align*}
		&I(D; \vY^n | F) - I(D; \vZ^n | F)\\[-4pt]
		&\qquad\stackrel{(a)}{=} I(D; \vY^n | F, \vSF) - I(D; \vZ^n | F, \vSF) + o(n) \\[-4pt]
		&\qquad\stackrel{(b)}{\le} \sup_{f, \vsF} \Big( I(D; \vY^n | F=f, \vSF=\vsF)\\[-8pt]
		&\qquad \qquad - I(D; \vZ^n | F=f, \vSF=\vsF) \Big) + o(n) \\[-4pt]
		&\qquad\stackrel{(c)}{\le} \sup_{p_{D,\vX^n}} \Big(I(D;\vY^n) - I(D;\vZ^n)\Big) + o(n), \alabel{equ:expinit}
	\end{align*}
	where $o(n)$ denotes a quantity that grows sublinearly in $n$. In particular, since $H(\vSF) \le 2\nu \log|\setX|$ and the effective channel memory length $\nu$ is fixed, the boundary-state entropy does not scale with $n$ and hence satisfies $H(\vSF)=o(n)$. Consequently, step~$(a)$ holds since conditioning on $\vSF$ alters the mutual information difference by at most $o(n)$, i.e., $$|I(D; \vY^n | F)-I(D; \vY^n | F, \vSF)|\leq H(\vSF)=o(n),$$ and similarly for Eve's side. Step $(b)$ follows by upper bounding the conditional mutual information with the supremum over all possible deterministic realizations. Finally, step $(c)$ holds because the outputs $(\vY^n, \vZ^n)$ depend on $D$ only through the inputs $\vX^n$ and the boundary state $\vSF$. Since $\vSF$ is fixed to a constant $\vsF$ in $(b)$, we obtain the clean Markov chains $D \to \vX^n \to \vY^n$ and $D \to \vX^n \to \vZ^n$. Evaluating $(b)$ under these Markov structures simplifies to taking the supremum over all joint distributions of the form $p_D\cdot p_{\vX^n| D}$.
	
	Since the penalty term $\frac{1}{n}o(n)$ in \eqref{equ:expinit} vanishes as $n\to\infty$, we have
	$$
	\Rs \le \lim_{n\to\infty}\frac{1}{n} \sup_{p_{D,\vX^n}} \Big(I(D;\vY^n) - I(D;\vZ^n)\Big)^+\!\!.
	$$
	Combining this upper bound with the achievable secure rate in~\eqref{equ:app:achieve} establishes the secrecy capacity.
	
	
	\noindent
	\begin{figure*}\centering
		\section{}\label{apx:tb}
		\captionsetup{justification=centering}
		\captionof{table}{\textsc{Designed Superchannel (SC) Based on the State Types and the Branch Types in Table~\ref{tbl:trellis1} (Example~\ref{ex:supchatr}).}}\vspace*{5pt}
		\label{tbl:trellis2}
		{\resizebox{15.65cm}{10.6cm}{%
				\begin{tabular}{||c c c | c||}
					\hline\hline
					SC's	& trellis	  	& SC's 		&	\\[-5pt]
					start 	& code		  	& end	 	&   $\setT_{ij}^{(\ell)}$ \\[-3pt]
					state 	& input		    & state 	&	\\[-1pt]\hline\hline
					0 & 0 & 45 & $\setT_{00}^{(0)}$ \\
					0 & 1 & 89 & $\setT_{01}^{(0)}$ \\\hline
					1 & 0 & 180 & $\setT_{03}^{(0)}$ \\
					1 & 1 & 17 & $\setT_{00}^{(0)}$ \\\hline
					2 & 0 & 22 & $\setT_{00}^{(0)}$ \\
					2 & 1 & 152 & $\setT_{03}^{(0)}$ \\\hline
					3 & 0 & 72 & $\setT_{00}^{(0)}$ \\
					3 & 1 & 182 & $\setT_{03}^{(0)}$ \\\hline
					4 & 0 & 105 & $\setT_{02}^{(1)}$ \\
					4 & 1 & 50 & $\setT_{00}^{(0)}$ \\\hline
					5 & 0 & 155 & $\setT_{03}^{(1)}$ \\
					5 & 1 & 68 & $\setT_{00}^{(0)}$ \\\hline
					6 & 0 & 177 & $\setT_{03}^{(1)}$ \\
					6 & 1 & 64 & $\setT_{00}^{(0)}$ \\\hline
					7 & 0 & 31 & $\setT_{00}^{(0)}$ \\
					7 & 1 & 136 & $\setT_{03}^{(1)}$ \\\hline
					8 & 0 & 77 & $\setT_{00}^{(0)}$ \\
					8 & 1 & 114 & $\setT_{03}^{(1)}$ \\\hline
					9 & 0 & 36 & $\setT_{00}^{(0)}$ \\
					9 & 1 & 171 & $\setT_{03}^{(0)}$ \\\hline
					10 & 0 & 128 & $\setT_{03}^{(3)}$ \\
					10 & 1 & 71 & $\setT_{00}^{(0)}$ \\\hline
					11 & 0 & 26 & $\setT_{00}^{(0)}$ \\
					11 & 1 & 179 & $\setT_{03}^{(3)}$ \\\hline
					12 & 0 & 100 & $\setT_{02}^{(1)}$ \\
					12 & 1 & 61 & $\setT_{00}^{(0)}$ \\\hline
					13 & 0 & 13 & $\setT_{00}^{(0)}$ \\
					13 & 1 & 161 & $\setT_{03}^{(0)}$ \\\hline
					14 & 0 & 64 & $\setT_{00}^{(0)}$ \\
					14 & 1 & 123 & $\setT_{03}^{(0)}$ \\\hline
					15 & 0 & 173 & $\setT_{03}^{(3)}$ \\
					15 & 1 & 15 & $\setT_{00}^{(0)}$ \\\hline
					16 & 0 & 145 & $\setT_{03}^{(3)}$ \\
					16 & 1 & 75 & $\setT_{00}^{(0)}$ \\\hline
					17 & 0 & 80 & $\setT_{00}^{(0)}$ \\
					17 & 1 & 91 & $\setT_{02}^{(1)}$ \\\hline
					18 & 0 & 166 & $\setT_{03}^{(3)}$ \\
					18 & 1 & 0 & $\setT_{00}^{(0)}$ \\\hline
					19 & 0 & 124 & $\setT_{03}^{(0)}$ \\
					19 & 1 & 60 & $\setT_{00}^{(0)}$ \\\hline
					20 & 0 & 88 & $\setT_{01}^{(0)}$ \\
					20 & 1 & 19 & $\setT_{00}^{(0)}$ \\\hline
					21 & 0 & 39 & $\setT_{00}^{(0)}$ \\
					21 & 1 & 102 & $\setT_{02}^{(3)}$ \\\hline
					22 & 0 & 137 & $\setT_{03}^{(1)}$ \\
					22 & 1 & 9 & $\setT_{00}^{(0)}$ \\\hline
					23 & 0 & 18 & $\setT_{00}^{(0)}$ \\
					23 & 1 & 194 & $\setT_{03}^{(3)}$ \\\hline
					24 & 0 & 22 & $\setT_{00}^{(0)}$ \\
					24 & 1 & 73 & $\setT_{00}^{(2)}$ \\\hline
					25 & 0 & 70 & $\setT_{00}^{(0)}$ \\
					25 & 1 & 172 & $\setT_{03}^{(1)}$ \\\hline
					26 & 0 & 78 & $\setT_{00}^{(0)}$ \\
					26 & 1 & 99 & $\setT_{02}^{(0)}$ \\\hline
					27 & 0 & 149 & $\setT_{03}^{(3)}$ \\
					27 & 1 & 52 & $\setT_{00}^{(0)}$ \\\hline
					28 & 0 & 52 & $\setT_{00}^{(0)}$ \\
					28 & 1 & 90 & $\setT_{01}^{(0)}$ \\\hline
					29 & 0 & 129 & $\setT_{03}^{(1)}$ \\
					29 & 1 & 67 & $\setT_{00}^{(0)}$ \\\hline
					30 & 0 & 45 & $\setT_{00}^{(0)}$ \\
					30 & 1 & 116 & $\setT_{03}^{(3)}$ \\\hline
					31 & 0 & 168 & $\setT_{03}^{(0)}$ \\
					31 & 1 & 43 & $\setT_{00}^{(0)}$ \\\hline
					32 & 0 & 176 & $\setT_{03}^{(3)}$ \\
					32 & 1 & 13 & $\setT_{00}^{(0)}$ \\\hline
					33 & 0 & 181 & $\setT_{03}^{(3)}$ \\
					33 & 1 & 27 & $\setT_{00}^{(0)}$ \\\hline
					34 & 0 & 139 & $\setT_{03}^{(1)}$ \\
					34 & 1 & 42 & $\setT_{00}^{(0)}$ \\\hline
					35 & 0 & 125 & $\setT_{03}^{(3)}$ \\
					35 & 1 & 51 & $\setT_{00}^{(0)}$ \\\hline
					36 & 0 & 4 & $\setT_{00}^{(0)}$ \\
					36 & 1 & 179 & $\setT_{03}^{(3)}$ \\\hline
					37 & 0 & 89 & $\setT_{01}^{(0)}$ \\
					37 & 1 & 197 & $\setT_{03}^{(0)}$ \\\hline
					38 & 0 & 31 & $\setT_{00}^{(0)}$ \\
					38 & 1 & 140 & $\setT_{03}^{(1)}$ \\\hline
					39 & 0 & 96 & $\setT_{02}^{(0)}$ \\
					39 & 1 & 47 & $\setT_{00}^{(0)}$ \\\hline
					\hline
				\end{tabular}\,\,
				\begin{tabular}{||c c c | c||}
					\hline\hline
					SC's	& trellis	  	& SC's 		&	\\[-5pt]
					start 	& code		  	& end	 	&   $\setT_{ij}^{(\ell)}$ \\[-3pt]
					state 	& input		    & state 	&	\\[-1pt]\hline\hline
					40 & 0 & 32 & $\setT_{00}^{(2)}$ \\
					40 & 1 & 77 & $\setT_{00}^{(0)}$ \\\hline
					41 & 0 & 159 & $\setT_{03}^{(1)}$ \\
					41 & 1 & 72 & $\setT_{00}^{(0)}$ \\\hline
					42 & 0 & 143 & $\setT_{03}^{(3)}$ \\
					42 & 1 & 19 & $\setT_{00}^{(0)}$ \\\hline
					43 & 0 & 198 & $\setT_{03}^{(1)}$ \\
					43 & 1 & 162 & $\setT_{03}^{(0)}$ \\\hline
					44 & 0 & 147 & $\setT_{03}^{(3)}$ \\
					44 & 1 & 17 & $\setT_{00}^{(0)}$ \\\hline
					45 & 0 & 188 & $\setT_{03}^{(0)}$ \\
					45 & 1 & 37 & $\setT_{00}^{(0)}$ \\\hline
					46 & 0 & 88 & $\setT_{01}^{(0)}$ \\
					46 & 1 & 67 & $\setT_{00}^{(0)}$ \\\hline
					47 & 0 & 58 & $\setT_{00}^{(0)}$ \\
					47 & 1 & 150 & $\setT_{03}^{(3)}$ \\\hline
					48 & 0 & 172 & $\setT_{03}^{(0)}$ \\
					48 & 1 & 63 & $\setT_{00}^{(0)}$ \\\hline
					49 & 0 & 112 & $\setT_{03}^{(3)}$ \\
					49 & 1 & 16 & $\setT_{00}^{(0)}$ \\\hline
					50 & 0 & 199 & $\setT_{03}^{(0)}$ \\
					50 & 1 & 2 & $\setT_{00}^{(0)}$ \\\hline
					51 & 0 & 84 & $\setT_{01}^{(0)}$ \\
					51 & 1 & 29 & $\setT_{00}^{(0)}$ \\\hline
					52 & 0 & 98 & $\setT_{02}^{(0)}$ \\
					52 & 1 & 73 & $\setT_{00}^{(0)}$ \\\hline
					53 & 0 & 196 & $\setT_{03}^{(0)}$ \\
					53 & 1 & 158 & $\setT_{03}^{(1)}$ \\\hline
					54 & 0 & 186 & $\setT_{03}^{(3)}$ \\
					54 & 1 & 175 & $\setT_{03}^{(1)}$ \\\hline
					55 & 0 & 115 & $\setT_{03}^{(0)}$ \\
					55 & 1 & 23 & $\setT_{00}^{(0)}$ \\\hline
					56 & 0 & 144 & $\setT_{03}^{(1)}$ \\
					56 & 1 & 53 & $\setT_{00}^{(0)}$ \\\hline
					57 & 0 & 171 & $\setT_{03}^{(0)}$ \\
					57 & 1 & 57 & $\setT_{00}^{(0)}$ \\\hline
					58 & 0 & 187 & $\setT_{03}^{(3)}$ \\
					58 & 1 & 54 & $\setT_{00}^{(0)}$ \\\hline
					59 & 0 & 95 & $\setT_{02}^{(3)}$ \\
					59 & 1 & 40 & $\setT_{00}^{(0)}$ \\\hline
					60 & 0 & 49 & $\setT_{00}^{(0)}$ \\
					60 & 1 & 87 & $\setT_{01}^{(0)}$ \\\hline
					61 & 0 & 189 & $\setT_{03}^{(3)}$ \\
					61 & 1 & 18 & $\setT_{00}^{(0)}$ \\\hline
					62 & 0 & 30 & $\setT_{00}^{(0)}$ \\
					62 & 1 & 122 & $\setT_{03}^{(3)}$ \\\hline
					63 & 0 & 103 & $\setT_{02}^{(3)}$ \\
					63 & 1 & 8 & $\setT_{00}^{(0)}$ \\\hline
					64 & 0 & 7 & $\setT_{00}^{(3)}$ \\
					64 & 1 & 12 & $\setT_{00}^{(0)}$ \\\hline
					65 & 0 & 199 & $\setT_{03}^{(3)}$ \\
					65 & 1 & 10 & $\setT_{00}^{(0)}$ \\\hline
					66 & 0 & 151 & $\setT_{03}^{(3)}$ \\
					66 & 1 & 58 & $\setT_{00}^{(0)}$ \\\hline
					67 & 0 & 76 & $\setT_{00}^{(0)}$ \\
					67 & 1 & 150 & $\setT_{03}^{(3)}$ \\\hline
					68 & 0 & 44 & $\setT_{00}^{(0)}$ \\
					68 & 1 & 74 & $\setT_{00}^{(1)}$ \\\hline
					69 & 0 & 108 & $\setT_{02}^{(0)}$ \\
					69 & 1 & 63 & $\setT_{00}^{(0)}$ \\\hline
					70 & 0 & 115 & $\setT_{03}^{(3)}$ \\
					70 & 1 & 24 & $\setT_{00}^{(0)}$ \\\hline
					71 & 0 & 79 & $\setT_{00}^{(1)}$ \\
					71 & 1 & 82 & $\setT_{00}^{(0)}$ \\\hline
					72 & 0 & 124 & $\setT_{03}^{(3)}$ \\
					72 & 1 & 28 & $\setT_{00}^{(0)}$ \\\hline
					73 & 0 & 68 & $\setT_{00}^{(2)}$ \\
					73 & 1 & 35 & $\setT_{00}^{(0)}$ \\\hline
					74 & 0 & 182 & $\setT_{03}^{(3)}$ \\
					74 & 1 & 74 & $\setT_{00}^{(0)}$ \\\hline
					75 & 0 & 161 & $\setT_{03}^{(1)}$ \\
					75 & 1 & 46 & $\setT_{00}^{(0)}$ \\\hline
					76 & 0 & 65 & $\setT_{00}^{(0)}$ \\
					76 & 1 & 160 & $\setT_{03}^{(0)}$ \\\hline
					77 & 0 & 36 & $\setT_{00}^{(0)}$ \\
					77 & 1 & 193 & $\setT_{03}^{(3)}$ \\\hline
					78 & 0 & 85 & $\setT_{01}^{(0)}$ \\
					78 & 1 & 81 & $\setT_{00}^{(0)}$ \\\hline
					79 & 0 & 177 & $\setT_{03}^{(0)}$ \\
					79 & 1 & 5 & $\setT_{00}^{(0)}$ \\\hline
					\hline
				\end{tabular}\,\,
				\begin{tabular}{||c c c | c||}
					\hline\hline
					SC's	& trellis	  	& SC's 		&	\\[-5pt]
					start 	& code		  	& end	 	&   $\setT_{ij}^{(\ell)}$ \\[-3pt]
					state 	& input		    & state 	&	\\[-1pt]\hline\hline
					80 & 0 & 42 & $\setT_{00}^{(0)}$ \\
					80 & 1 & 97 & $\setT_{02}^{(1)}$ \\\hline
					81 & 0 & 112 & $\setT_{03}^{(1)}$ \\
					81 & 1 & 8 & $\setT_{00}^{(0)}$ \\\hline
					82 & 0 & 198 & $\setT_{03}^{(0)}$ \\
					82 & 1 & 62 & $\setT_{00}^{(0)}$ \\\hline
					83 & 0 & 83 & $\setT_{00}^{(0)}$ \\
					83 & 1 & 174 & $\setT_{03}^{(0)}$ \\\hline
					84 & 0 & 94 & $\setT_{12}^{(3)}$ \\
					84 & 1 & 71 & $\setT_{10}^{(3)}$ \\\hline
					85 & 0 & 107 & $\setT_{12}^{(3)}$ \\
					85 & 1 & 76 & $\setT_{10}^{(0)}$ \\\hline
					86 & 0 & 28 & $\setT_{10}^{(0)}$ \\
					86 & 1 & 188 & $\setT_{13}^{(3)}$ \\\hline
					87 & 0 & 69 & $\setT_{10}^{(2)}$ \\
					87 & 1 & 142 & $\setT_{13}^{(3)}$ \\\hline
					88 & 0 & 16 & $\setT_{10}^{(2)}$ \\
					88 & 1 & 158 & $\setT_{13}^{(3)}$ \\\hline
					89 & 0 & 81 & $\setT_{10}^{(2)}$ \\
					89 & 1 & 57 & $\setT_{10}^{(3)}$ \\\hline
					90 & 0 & 181 & $\setT_{13}^{(3)}$ \\
					90 & 1 & 97 & $\setT_{12}^{(3)}$ \\\hline
					91 & 0 & 90 & $\setT_{21}^{(0)}$ \\
					91 & 1 & 48 & $\setT_{20}^{(0)}$ \\\hline
					92 & 0 & 62 & $\setT_{20}^{(0)}$ \\
					92 & 1 & 104 & $\setT_{22}^{(0)}$ \\\hline
					93 & 0 & 196 & $\setT_{23}^{(0)}$ \\
					93 & 1 & 61 & $\setT_{20}^{(0)}$ \\\hline
					94 & 0 & 41 & $\setT_{20}^{(0)}$ \\
					94 & 1 & 113 & $\setT_{23}^{(0)}$ \\\hline
					95 & 0 & 78 & $\setT_{20}^{(0)}$ \\
					95 & 1 & 168 & $\setT_{23}^{(0)}$ \\\hline
					96 & 0 & 51 & $\setT_{20}^{(0)}$ \\
					96 & 1 & 111 & $\setT_{23}^{(0)}$ \\\hline
					97 & 0 & 7 & $\setT_{20}^{(0)}$ \\
					97 & 1 & 113 & $\setT_{23}^{(3)}$ \\\hline
					98 & 0 & 56 & $\setT_{20}^{(0)}$ \\
					98 & 1 & 133 & $\setT_{23}^{(0)}$ \\\hline
					99 & 0 & 59 & $\setT_{20}^{(0)}$ \\
					99 & 1 & 166 & $\setT_{23}^{(1)}$ \\\hline
					100 & 0 & 41 & $\setT_{20}^{(0)}$ \\
					100 & 1 & 131 & $\setT_{23}^{(3)}$ \\\hline
					101 & 0 & 176 & $\setT_{23}^{(0)}$ \\
					101 & 1 & 43 & $\setT_{20}^{(0)}$ \\\hline
					102 & 0 & 183 & $\setT_{23}^{(0)}$ \\
					102 & 1 & 25 & $\setT_{20}^{(0)}$ \\\hline
					103 & 0 & 130 & $\setT_{23}^{(0)}$ \\
					103 & 1 & 86 & $\setT_{21}^{(0)}$ \\\hline
					104 & 0 & 143 & $\setT_{23}^{(0)}$ \\
					104 & 1 & 44 & $\setT_{20}^{(0)}$ \\\hline
					105 & 0 & 29 & $\setT_{20}^{(0)}$ \\
					105 & 1 & 195 & $\setT_{23}^{(0)}$ \\\hline
					106 & 0 & 24 & $\setT_{20}^{(0)}$ \\
					106 & 1 & 106 & $\setT_{22}^{(3)}$ \\\hline
					107 & 0 & 6 & $\setT_{20}^{(0)}$ \\
					107 & 1 & 46 & $\setT_{20}^{(2)}$ \\\hline
					108 & 0 & 153 & $\setT_{23}^{(1)}$ \\
					108 & 1 & 32 & $\setT_{20}^{(0)}$ \\\hline
					109 & 0 & 133 & $\setT_{23}^{(0)}$ \\
					109 & 1 & 33 & $\setT_{20}^{(0)}$ \\\hline
					110 & 0 & 118 & $\setT_{23}^{(3)}$ \\
					110 & 1 & 14 & $\setT_{20}^{(0)}$ \\\hline
					111 & 0 & 185 & $\setT_{33}^{(3)}$ \\
					111 & 1 & 102 & $\setT_{32}^{(3)}$ \\\hline
					112 & 0 & 55 & $\setT_{30}^{(3)}$ \\
					112 & 1 & 157 & $\setT_{33}^{(3)}$ \\\hline
					113 & 0 & 146 & $\setT_{33}^{(3)}$ \\
					113 & 1 & 101 & $\setT_{32}^{(3)}$ \\\hline
					114 & 0 & 110 & $\setT_{32}^{(2)}$ \\
					114 & 1 & 159 & $\setT_{33}^{(3)}$ \\\hline
					115 & 0 & 187 & $\setT_{33}^{(3)}$ \\
					115 & 1 & 1 & $\setT_{30}^{(0)}$ \\\hline
					116 & 0 & 105 & $\setT_{32}^{(3)}$ \\
					116 & 1 & 144 & $\setT_{33}^{(3)}$ \\\hline
					117 & 0 & 99 & $\setT_{32}^{(3)}$ \\
					117 & 1 & 121 & $\setT_{33}^{(3)}$ \\\hline
					118 & 0 & 11 & $\setT_{30}^{(2)}$ \\
					118 & 1 & 146 & $\setT_{33}^{(3)}$ \\\hline
					119 & 0 & 93 & $\setT_{32}^{(3)}$ \\
					119 & 1 & 178 & $\setT_{33}^{(3)}$ \\\hline
					\hline
				\end{tabular}\,\,
				\begin{tabular}{||c c c | c||}
					\hline\hline
					SC's	& trellis	  	& SC's 		&	\\[-5pt]
					start 	& code		  	& end	 	&   $\setT_{ij}^{(\ell)}$ \\[-3pt]
					state 	& input		    & state 	&	\\[-1pt]\hline\hline
					120 & 0 & 129 & $\setT_{33}^{(3)}$ \\
					120 & 1 & 30 & $\setT_{30}^{(0)}$ \\\hline
					121 & 0 & 103 & $\setT_{32}^{(3)}$ \\
					121 & 1 & 132 & $\setT_{33}^{(3)}$ \\\hline
					122 & 0 & 48 & $\setT_{30}^{(3)}$ \\
					122 & 1 & 164 & $\setT_{33}^{(3)}$ \\\hline
					123 & 0 & 173 & $\setT_{33}^{(3)}$ \\
					123 & 1 & 38 & $\setT_{30}^{(0)}$ \\\hline
					124 & 0 & 85 & $\setT_{31}^{(0)}$ \\
					124 & 1 & 154 & $\setT_{33}^{(3)}$ \\\hline
					125 & 0 & 135 & $\setT_{33}^{(3)}$ \\
					125 & 1 & 15 & $\setT_{30}^{(2)}$ \\\hline
					126 & 0 & 185 & $\setT_{33}^{(3)}$ \\
					126 & 1 & 33 & $\setT_{30}^{(2)}$ \\\hline
					127 & 0 & 197 & $\setT_{33}^{(3)}$ \\
					127 & 1 & 47 & $\setT_{30}^{(0)}$ \\\hline
					128 & 0 & 53 & $\setT_{30}^{(3)}$ \\
					128 & 1 & 141 & $\setT_{33}^{(3)}$ \\\hline
					129 & 0 & 49 & $\setT_{30}^{(0)}$ \\
					129 & 1 & 186 & $\setT_{33}^{(3)}$ \\\hline
					130 & 0 & 108 & $\setT_{32}^{(3)}$ \\
					130 & 1 & 114 & $\setT_{33}^{(3)}$ \\\hline
					131 & 0 & 25 & $\setT_{30}^{(3)}$ \\
					131 & 1 & 160 & $\setT_{33}^{(3)}$ \\\hline
					132 & 0 & 191 & $\setT_{33}^{(3)}$ \\
					132 & 1 & 125 & $\setT_{33}^{(2)}$ \\\hline
					133 & 0 & 93 & $\setT_{32}^{(3)}$ \\
					133 & 1 & 169 & $\setT_{33}^{(3)}$ \\\hline
					134 & 0 & 100 & $\setT_{32}^{(3)}$ \\
					134 & 1 & 190 & $\setT_{33}^{(3)}$ \\\hline
					135 & 0 & 134 & $\setT_{33}^{(3)}$ \\
					135 & 1 & 120 & $\setT_{33}^{(2)}$ \\\hline
					136 & 0 & 4 & $\setT_{30}^{(3)}$ \\
					136 & 1 & 180 & $\setT_{33}^{(3)}$ \\\hline
					137 & 0 & 56 & $\setT_{30}^{(2)}$ \\
					137 & 1 & 149 & $\setT_{33}^{(3)}$ \\\hline
					138 & 0 & 183 & $\setT_{33}^{(3)}$ \\
					138 & 1 & 65 & $\setT_{30}^{(3)}$ \\\hline
					139 & 0 & 60 & $\setT_{30}^{(2)}$ \\
					139 & 1 & 154 & $\setT_{33}^{(3)}$ \\\hline
					140 & 0 & 111 & $\setT_{33}^{(3)}$ \\
					140 & 1 & 94 & $\setT_{32}^{(3)}$ \\\hline
					141 & 0 & 87 & $\setT_{31}^{(3)}$ \\
					141 & 1 & 147 & $\setT_{33}^{(3)}$ \\\hline
					142 & 0 & 192 & $\setT_{33}^{(3)}$ \\
					142 & 1 & 109 & $\setT_{32}^{(3)}$ \\\hline
					143 & 0 & 69 & $\setT_{30}^{(3)}$ \\
					143 & 1 & 165 & $\setT_{33}^{(3)}$ \\\hline
					144 & 0 & 195 & $\setT_{33}^{(3)}$ \\
					144 & 1 & 34 & $\setT_{30}^{(0)}$ \\\hline
					145 & 0 & 132 & $\setT_{33}^{(3)}$ \\
					145 & 1 & 66 & $\setT_{30}^{(3)}$ \\\hline
					146 & 0 & 2 & $\setT_{30}^{(3)}$ \\
					146 & 1 & 155 & $\setT_{33}^{(3)}$ \\\hline
					147 & 0 & 107 & $\setT_{32}^{(3)}$ \\
					147 & 1 & 165 & $\setT_{33}^{(3)}$ \\\hline
					148 & 0 & 54 & $\setT_{30}^{(0)}$ \\
					148 & 1 & 117 & $\setT_{33}^{(3)}$ \\\hline
					149 & 0 & 75 & $\setT_{30}^{(2)}$ \\
					149 & 1 & 156 & $\setT_{33}^{(3)}$ \\\hline
					150 & 0 & 37 & $\setT_{30}^{(0)}$ \\
					150 & 1 & 170 & $\setT_{33}^{(3)}$ \\\hline
					151 & 0 & 34 & $\setT_{30}^{(0)}$ \\
					151 & 1 & 167 & $\setT_{33}^{(3)}$ \\\hline
					152 & 0 & 84 & $\setT_{31}^{(2)}$ \\
					152 & 1 & 126 & $\setT_{33}^{(3)}$ \\\hline
					153 & 0 & 130 & $\setT_{33}^{(0)}$ \\
					153 & 1 & 138 & $\setT_{33}^{(3)}$ \\\hline
					154 & 0 & 184 & $\setT_{33}^{(3)}$ \\
					154 & 1 & 35 & $\setT_{30}^{(2)}$ \\\hline
					155 & 0 & 21 & $\setT_{30}^{(2)}$ \\
					155 & 1 & 122 & $\setT_{33}^{(3)}$ \\\hline
					156 & 0 & 92 & $\setT_{32}^{(3)}$ \\
					156 & 1 & 174 & $\setT_{33}^{(3)}$ \\\hline
					157 & 0 & 14 & $\setT_{30}^{(0)}$ \\
					157 & 1 & 135 & $\setT_{33}^{(3)}$ \\\hline
					158 & 0 & 82 & $\setT_{30}^{(3)}$ \\
					158 & 1 & 118 & $\setT_{33}^{(3)}$ \\\hline
					159 & 0 & 119 & $\setT_{33}^{(3)}$ \\
					159 & 1 & 162 & $\setT_{33}^{(2)}$ \\\hline
					\hline
				\end{tabular}\,\,
				\begin{tabular}{||c c c | c||}
					\hline\hline
					SC's	& trellis	  	& SC's 		&	\\[-5pt]
					start 	& code		  	& end	 	&   $\setT_{ij}^{(\ell)}$ \\[-3pt]
					state 	& input		    & state 	&	\\[-1pt]\hline\hline
					160 & 0 & 70 & $\setT_{30}^{(0)}$ \\
					160 & 1 & 127 & $\setT_{33}^{(3)}$ \\\hline
					161 & 0 & 169 & $\setT_{33}^{(3)}$ \\
					161 & 1 & 12 & $\setT_{30}^{(0)}$ \\\hline
					162 & 0 & 59 & $\setT_{30}^{(0)}$ \\
					162 & 1 & 190 & $\setT_{33}^{(3)}$ \\\hline
					163 & 0 & 3 & $\setT_{30}^{(0)}$ \\
					163 & 1 & 145 & $\setT_{33}^{(3)}$ \\\hline
					164 & 0 & 38 & $\setT_{30}^{(3)}$ \\
					164 & 1 & 170 & $\setT_{33}^{(3)}$ \\\hline
					165 & 0 & 79 & $\setT_{30}^{(3)}$ \\
					165 & 1 & 128 & $\setT_{33}^{(3)}$ \\\hline
					166 & 0 & 127 & $\setT_{33}^{(3)}$ \\
					166 & 1 & 194 & $\setT_{33}^{(1)}$ \\\hline
					167 & 0 & 39 & $\setT_{30}^{(0)}$ \\
					167 & 1 & 164 & $\setT_{33}^{(3)}$ \\\hline
					168 & 0 & 20 & $\setT_{30}^{(2)}$ \\
					168 & 1 & 142 & $\setT_{33}^{(3)}$ \\\hline
					169 & 0 & 20 & $\setT_{30}^{(0)}$ \\
					169 & 1 & 189 & $\setT_{33}^{(3)}$ \\\hline
					170 & 0 & 9 & $\setT_{30}^{(2)}$ \\
					170 & 1 & 117 & $\setT_{33}^{(3)}$ \\\hline
					171 & 0 & 109 & $\setT_{32}^{(3)}$ \\
					171 & 1 & 148 & $\setT_{33}^{(3)}$ \\\hline
					172 & 0 & 119 & $\setT_{33}^{(3)}$ \\
					172 & 1 & 110 & $\setT_{32}^{(3)}$ \\\hline
					173 & 0 & 86 & $\setT_{31}^{(0)}$ \\
					173 & 1 & 66 & $\setT_{30}^{(2)}$ \\\hline
					174 & 0 & 139 & $\setT_{33}^{(3)}$ \\
					174 & 1 & 98 & $\setT_{32}^{(3)}$ \\\hline
					175 & 0 & 163 & $\setT_{33}^{(3)}$ \\
					175 & 1 & 101 & $\setT_{32}^{(3)}$ \\\hline
					176 & 0 & 137 & $\setT_{33}^{(3)}$ \\
					176 & 1 & 26 & $\setT_{30}^{(0)}$ \\\hline
					177 & 0 & 6 & $\setT_{30}^{(0)}$ \\
					177 & 1 & 175 & $\setT_{33}^{(3)}$ \\\hline
					178 & 0 & 23 & $\setT_{30}^{(2)}$ \\
					178 & 1 & 126 & $\setT_{33}^{(3)}$ \\\hline
					179 & 0 & 80 & $\setT_{30}^{(3)}$ \\
					179 & 1 & 157 & $\setT_{33}^{(3)}$ \\\hline
					180 & 0 & 104 & $\setT_{32}^{(3)}$ \\
					180 & 1 & 152 & $\setT_{33}^{(3)}$ \\\hline
					181 & 0 & 163 & $\setT_{33}^{(3)}$ \\
					181 & 1 & 3 & $\setT_{30}^{(3)}$ \\\hline
					182 & 0 & 83 & $\setT_{30}^{(2)}$ \\
					182 & 1 & 141 & $\setT_{33}^{(3)}$ \\\hline
					183 & 0 & 1 & $\setT_{30}^{(2)}$ \\
					183 & 1 & 192 & $\setT_{33}^{(3)}$ \\\hline
					184 & 0 & 136 & $\setT_{33}^{(3)}$ \\
					184 & 1 & 11 & $\setT_{30}^{(0)}$ \\\hline
					185 & 0 & 153 & $\setT_{33}^{(3)}$ \\
					185 & 1 & 106 & $\setT_{32}^{(3)}$ \\\hline
					186 & 0 & 148 & $\setT_{33}^{(3)}$ \\
					186 & 1 & 116 & $\setT_{33}^{(0)}$ \\\hline
					187 & 0 & 21 & $\setT_{30}^{(0)}$ \\
					187 & 1 & 121 & $\setT_{33}^{(3)}$ \\\hline
					188 & 0 & 191 & $\setT_{33}^{(3)}$ \\
					188 & 1 & 40 & $\setT_{30}^{(2)}$ \\\hline
					189 & 0 & 96 & $\setT_{32}^{(3)}$ \\
					189 & 1 & 178 & $\setT_{33}^{(3)}$ \\\hline
					190 & 0 & 91 & $\setT_{32}^{(3)}$ \\
					190 & 1 & 120 & $\setT_{33}^{(3)}$ \\\hline
					191 & 0 & 10 & $\setT_{30}^{(0)}$ \\
					191 & 1 & 134 & $\setT_{33}^{(3)}$ \\\hline
					192 & 0 & 55 & $\setT_{30}^{(0)}$ \\
					192 & 1 & 140 & $\setT_{33}^{(3)}$ \\\hline
					193 & 0 & 27 & $\setT_{30}^{(3)}$ \\
					193 & 1 & 151 & $\setT_{33}^{(3)}$ \\\hline
					194 & 0 & 0 & $\setT_{30}^{(3)}$ \\
					194 & 1 & 156 & $\setT_{33}^{(3)}$ \\\hline
					195 & 0 & 123 & $\setT_{33}^{(3)}$ \\
					195 & 1 & 95 & $\setT_{32}^{(3)}$ \\\hline
					196 & 0 & 50 & $\setT_{30}^{(0)}$ \\
					196 & 1 & 193 & $\setT_{33}^{(3)}$ \\\hline
					197 & 0 & 5 & $\setT_{30}^{(0)}$ \\
					197 & 1 & 131 & $\setT_{33}^{(3)}$ \\\hline
					198 & 0 & 138 & $\setT_{33}^{(3)}$ \\
					198 & 1 & 92 & $\setT_{32}^{(3)}$ \\\hline
					199 & 0 & 167 & $\setT_{33}^{(3)}$ \\
					199 & 1 & 184 & $\setT_{33}^{(2)}$ \\\hline
					\hline
		\end{tabular}}}
	\end{figure*}
	
	\setstretch{1.02}
	\bibliographystyle{IEEEtran}
	\bibliography{citation}

\begin{thebibliography}{10}
\providecommand{\url}[1]{#1}
\csname url@samestyle\endcsname
\providecommand{\newblock}{\relax}
\providecommand{\bibinfo}[2]{#2}
\providecommand{\BIBentrySTDinterwordspacing}{\spaceskip=0pt\relax}
\providecommand{\BIBentryALTinterwordstretchfactor}{4}
\providecommand{\BIBentryALTinterwordspacing}{\spaceskip=\fontdimen2\font plus
\BIBentryALTinterwordstretchfactor\fontdimen3\font minus
  \fontdimen4\font\relax}
\providecommand{\BIBforeignlanguage}[2]{{%
\expandafter\ifx\csname l@#1\endcsname\relax
\typeout{** WARNING: IEEEtran.bst: No hyphenation pattern has been}%
\typeout{** loaded for the language `#1'. Using the pattern for}%
\typeout{** the default language instead.}%
\else
\language=\csname l@#1\endcsname
\fi
#2}}
\providecommand{\BIBdecl}{\relax}
\BIBdecl

\bibitem{9834578}
A.~Nouri and R.~Asvadi, ``Matched information rate codes for binary-input
  intersymbol interference wiretap channels,'' in \emph{Proc. IEEE Int.\ Symp.\
  Inf.\ Theory}, Espoo, Finland, June 2022, pp. 1052--1057.

\bibitem{Gidney_2021}
C.~Gidney and M.~Ekerå, ``How to factor 2048 bit \textmd{RSA} integers in 8
  hours using 20 million noisy qubits,'' \emph{Quantum}, vol.~5, p. 433, Apr.
  2021.

\bibitem{Pirandola:20}
S.~Pirandola, U.~L. Andersen, L.~Banchi, M.~Berta, D.~Bunandar, R.~Colbeck,
  D.~Englund, T.~Gehring, C.~Lupo, C.~Ottaviani, J.~L. Pereira, M.~Razavi,
  J.~S. Shaari, M.~Tomamichel, V.~C. Usenko, G.~Vallone, P.~Villoresi, and
  P.~Wallden, ``Advances in quantum cryptography,'' \emph{Adv. Opt. Photon.},
  vol.~12, no.~4, pp. 1012--1236, Dec. 2020.

\bibitem{Wehnereaam9288}
S.~Wehner, D.~Elkouss, and R.~Hanson, ``Quantum internet: \text{A} vision for
  the road ahead,'' \emph{Science}, vol. 362, no. 6412, Oct. 2018.

\bibitem{9023997}
A.~S. {Cacciapuoti}, M.~{Caleffi}, R.~{Van Meter}, and L.~{Hanzo}, ``When
  entanglement meets classical communications: \text{Quantum} teleportation for
  the quantum internet,'' \emph{IEEE Trans.\ Commun.}, vol.~68, no.~6, pp.
  3808--3833, June 2020.

\bibitem{PhysRevLett.78.3414}
D.~Mayers, ``Unconditionally secure quantum bit commitment is impossible,''
  \emph{Phys. Rev. Lett.}, vol.~78, pp. 3414--3417, Apr. 1997.

\bibitem{6772207}
A.~D. Wyner, ``The wire-tap channel,'' \emph{Bell Syst. Tech. J.}, vol.~54,
  no.~8, pp. 1355--1387, Oct. 1975.

\bibitem{9380147}
M.~{Bloch}, O.~{Günlü}, A.~{Yener}, F.~{Oggier}, H.~V. {Poor}, L.~{Sankar},
  and R.~F. {Schaefer}, ``An overview of information-theoretic security and
  privacy: \text{Metrics}, limits and applications,'' \emph{IEEE J.\ Sel.\
  Areas Inf.\ Theory}, vol.~2, no.~1, pp. 5--22, Mar. 2021.

\bibitem{10336902}
M.~Mitev, T.~M. Pham, A.~Chorti, A.~N. Barreto, and G.~Fettweis, ``Physical
  layer security—from theory to practice,'' \emph{IEEE \text{BITS} Inf.\
  Theory Mag.}, vol.~3, no.~2, pp. 67--79, Dec. 2023.

\bibitem{8698792}
J.~Choi, ``Single-carrier index modulation for \textmd{IoT} uplink,''
  \emph{IEEE J. Sel. Top. Signal Process.}, vol.~13, no.~6, pp. 1237--1248,
  Oct. 2019.

\bibitem{1569979}
M.~Sahin and H.~Arslan, ``Inter-symbol interference in high data rate
  \textmd{UWB} communications using energy detector receivers,'' in \emph{Proc.
  IEEE Int.\ Conf.\ on Ultra-Wideband}, Zurich, Switzerland, Sept. 2005, pp.
  176--179.

\bibitem{9324767}
Z.~Zhang, Y.~Yin, and G.~M. Rebeiz, ``Intersymbol interference and equalization
  for large \text{5G} phased arrays with wide scan angles,'' \emph{IEEE Trans.\
  Microw.\ Theory Techn.}, vol.~69, no.~3, pp. 1955--1964, Mar. 2021.

\bibitem{10061469}
Y.~Wu, C.~Han, and Z.~Chen, ``\text{DFT}-spread orthogonal time frequency space
  system with superimposed pilots for terahertz integrated sensing and
  communication,'' \emph{IEEE Trans.\ Wirel.\ Commun.}, vol.~22, no.~11, pp.
  7361--7376, Nov. 2023.

\bibitem{1421929}
S.~H. Han and J.~H. Lee, ``An overview of peak-to-average power ratio reduction
  techniques for multicarrier transmission,'' \emph{IEEE Wirel. Commun.},
  vol.~12, no.~2, pp. 56--65, Apr. 2005.

\bibitem{10530524}
Y.~Qin, A.~Li, Y.~Lyu, X.~Liao, and C.~Masouros, ``Symbol-level precoding for
  \text{PAPR} reduction in multi-user \text{MISO-OFDM} systems,'' \emph{IEEE
  Trans.\ Wirel.\ Commun.}, vol.~23, no.~9, pp. 12\,484--12\,498, Sept. 2024.

\bibitem{3GPP_TS_38_101_1}
{3rd Generation Partnership Project (3GPP)}, ``\text{TS 38.101-1 (V18.9.0):}
  \text{5G;} \text{NR;} \text{User} equipment \text{(UE)} radio transmission
  and reception; \text{(Release 18)},'' 3GPP, Technical Specification 38.101-1,
  Apr. 2025.

\bibitem{3GPP_TR_38_821}
------, ``\text{TR 38.821 (V16.0.0):} \text{Technical} specification group
  radio access network; \text{Solutions} for \text{NR} to support
  non-terrestrial networks \text{(NTN)} \text{(Release 16)},'' 3GPP, Technical
  Specification 38.821, Dec. 2019.

\bibitem{3gpp-ntn-overview}
------, ``{Non-Terrestrial Networks \text{(NTN)} Overview},''
  \url{https://www.3gpp.org/technologies/ntn-overview}, {Apr.} {2025}.

\bibitem{10068266}
A.~Nouri, R.~Asvadi, J.~Chen, and P.~O. Vontobel, ``Constrained secrecy
  capacity of finite-input intersymbol interference wiretap channels,''
  \emph{IEEE Trans.\ Commun.}, vol.~71, no.~6, pp. 3301--3316, June 2023.

\bibitem{1054829}
G.~Forney, ``Maximum-likelihood sequence estimation of digital sequences in the
  presence of intersymbol interference,'' \emph{IEEE Trans.\ Inf.\ Theory},
  vol.~18, no.~3, pp. 363--378, May 1972.

\bibitem{382016}
Y.~Li, B.~Vucetic, and Y.~Sato, ``Optimum soft-output detection for channels
  with intersymbol interference,'' \emph{IEEE Trans.\ Inf.\ Theory}, vol.~41,
  no.~3, pp. 704--713, May 1995.

\bibitem{4460060506}
C.~Douillard, M.~Jézéquel, C.~Berrou, D.~Electronique, A.~Picart, P.~Didier,
  and A.~Glavieux, ``Iterative correction of intersymbol interference:
  \text{Turbo}-equalization,'' \emph{Eur.\ Trans.\ Telecommun.}, vol.~6, no.~5,
  pp. 507--511, Sept. 1995.

\bibitem{864167}
T.~V. {Souvignier}, M.~{Oberg}, P.~H. {Siegel}, R.~E. {Swanson}, and J.~K.
  {Wolf}, ``Turbo decoding for partial response channels,'' \emph{IEEE Trans.\
  Commun.}, vol.~48, no.~8, pp. 1297--1308, Aug. 2000.

\bibitem{911451}
M.~{Oberg} and P.~H. {Siegel}, ``Performance analysis of turbo-equalized
  partial response channels,'' \emph{IEEE Trans.\ Commun.}, vol.~49, no.~3, pp.
  436--444, Mar. 2001.

\bibitem{1003830}
B.~M. Kurkoski, P.~H. Siegel, and J.~K. Wolf, ``Joint message-passing decoding
  of \textmd{LDPC} codes and partial-response channels,'' \emph{IEEE Trans.\
  Inf.\ Theory}, vol.~48, no.~6, pp. 1410--1422, June 2002.

\bibitem{1207365}
A.~Kavčić, X.~Ma, and M.~Mitzenmacher, ``Binary intersymbol interference
  channels: \text{Gallager} codes, density evolution, and code performance
  bounds,'' \emph{IEEE Trans.\ Inf.\ Theory}, vol.~49, no.~7, pp. 1636--1652,
  July 2003.

\bibitem{1431126}
G.~Colavolpe and G.~Germi, ``On the application of factor graphs and the
  sum-product algorithm to \textmd{ISI} channels,'' \emph{IEEE Trans.\
  Commun.}, vol.~53, no.~5, pp. 818--825, May 2005.

\bibitem{4282114}
A.~Anastasopoulos, K.~M. Chugg, G.~Colavolpe, G.~Ferrari, and R.~Raheli,
  ``Iterative detection for channels with memory,'' \emph{Proc. IEEE}, vol.~95,
  no.~6, pp. 1272--1294, June 2007.

\bibitem{9279252}
S.~Zheng, Y.~Liu, and P.~H. Siegel, ``\textmd{PR-NN}: \textmd{RNN}-based
  detection for coded partial-response channels,'' \emph{IEEE J. Sel. Areas
  Commun.}, vol.~39, no.~7, pp. 1967--1982, July 2021.

\bibitem{397441}
C.~Berrou, A.~Glavieux, and P.~Thitimajshima, ``Near \text{Shannon} limit
  error-correcting coding and decoding: \text{Turbo}-codes. 1,'' in \emph{Proc.
  IEEE Int.\ Conf.\ Commun.}, vol.~2, Geneva, Switzerland, May 1993, pp.
  1064--1070 vol.2.

\bibitem{10.1214/aoms/1177699147}
L.~E. Baum and T.~Petrie, ``Statistical inference for probabilistic functions
  of finite state \textmd{Markov} chains,'' \emph{Ann. Math. Stat.}, vol.~37,
  no.~6, pp. 1554--1563, Dec. 1966.

\bibitem{welch}
P.~L. McAdam, L.~R. Welch, and C.~L. Weber, ``\textmd{M.A.P.} bit decoding of
  convolutional codes,'' in \emph{Proc. IEEE Int.\ Symp.\ Inf.\ Theory},
  Asilomar, CA, USA, Dec. 1972, p.~91.

\bibitem{1055186}
L.~Bahl, J.~Cocke, F.~Jelinek, and J.~Raviv, ``Optimal decoding of linear codes
  for minimizing symbol error rate,'' \emph{IEEE Trans.\ Inf.\ Theory},
  vol.~20, no.~2, pp. 284--287, Mar. 1974.

\bibitem{910578}
T.~J. Richardson, M.~A. Shokrollahi, and R.~L. Urbanke, ``Design of
  capacity-approaching irregular low-density parity-check codes,'' \emph{IEEE
  Trans.\ Inf.\ Theory}, vol.~47, no.~2, pp. 619--637, Feb. 2001.

\bibitem{1195499}
N.~Varnica and A.~Kavčić, ``Optimized low-density parity-check codes for
  partial response channels,'' \emph{IEEE Commun. Lett.}, vol.~7, no.~4, pp.
  168--170, Apr. 2003.

\bibitem{8840908}
X.~Jiao, J.~Mu, Y.~C. He, and W.~Xu, ``Linear-complexity \textmd{ADMM} updates
  for decoding \textmd{LDPC} codes in partial response channels,'' \emph{IEEE
  Commun. Lett.}, vol.~23, no.~12, pp. 2200--2204, Dec. 2019.

\bibitem{9384303}
X.~Jiao, H.~Liu, J.~Mu, and Y.~C. He, ``$l_2$-box \textmd{ADMM} decoding for
  \textmd{LDPC} codes over \textmd{ISI} channels,'' \emph{IEEE Trans.\ Veh.\
  Technol.}, vol.~70, no.~4, pp. 3966--3971, Apr. 2021.

\bibitem{5571870}
T.~Wadayama, ``Interior point decoding for linear vector channels based on
  convex optimization,'' \emph{IEEE Trans.\ Inf.\ Theory}, vol.~56, no.~10, pp.
  4905--4921, Oct. 2010.

\bibitem{ADMM}
S.~Boyd, N.~Parikh, E.~Chu, B.~Peleato, and J.~Eckstein, ``Distributed
  optimization and statistical learning via the alternating direction method of
  multipliers,'' \emph{Found. Trends Mach. Learn.}, vol.~3, no.~1, pp. 1--122,
  Jan. 2011.

\bibitem{4137898}
J.~B. {Soriaga}, H.~D. {Pfister}, and P.~H. {Siegel}, ``Determining and
  approaching achievable rates of binary intersymbol interference channels
  using multistage decoding,'' \emph{IEEE Trans.\ Inf.\ Theory}, vol.~53,
  no.~4, pp. 1416--1429, Apr. 2007.

\bibitem{1397934}
A.~{Kavčić}, X.~{Ma}, and N.~{Varnica}, ``Matched information rate codes for
  partial response channels,'' \emph{IEEE Trans.\ Inf.\ Theory}, vol.~51,
  no.~3, pp. 973--989, Mar. 2005.

\bibitem{1055917}
S.~{Leung-Yan-Cheong} and M.~{Hellman}, ``The \textmd{Gaussian} wire-tap
  channel,'' \emph{IEEE Trans.\ Inf.\ Theory}, vol.~24, no.~4, pp. 451--456,
  July 1978.

\bibitem{5592833}
M.~{Baldi}, M.~{Bianchi}, and F.~{Chiaraluce}, ``Non-systematic codes for
  physical layer security,'' in \emph{Proc. IEEE Inf.\ Theory Workshop},
  Dublin, Ireland, Sept. 2010, pp. 1--5.

\bibitem{5740591}
D.~Klinc, J.~Ha, S.~W. McLaughlin, J.~Barros, and B.~J. Kwak, ``\textmd{LDPC}
  codes for the \textmd{Gaussian} wiretap channel,'' \emph{IEEE Trans.\ Inf.\
  Forensics Secur.}, vol.~6, no.~3, pp. 532--540, Sept. 2011.

\bibitem{6162586}
C.~W. Wong, T.~F. Wong, and J.~M. Shea, ``\textmd{LDPC} code design for the
  \textmd{BPSK}-constrained \textmd{Gaussian} wiretap channel,'' in \emph{Proc.
  IEEE Glob.\ Commun.\ Conf.}, San Antonio, TX, USA, Dec. 2011, pp. 898--902.

\bibitem{6151133}
M.~{Baldi}, M.~{Bianchi}, and F.~{Chiaraluce}, ``Coding with scrambling,
  concatenation, and \textmd{HARQ} for the \textmd{AWGN} wire-tap channel:
  \text{A} security gap analysis,'' \emph{IEEE Trans.\ Inf.\ Forensics Secur.},
  vol.~7, no.~3, pp. 883--894, June 2012.

\bibitem{7247218}
M.~Baldi, G.~Ricciutelli, N.~Maturo, and F.~Chiaraluce, ``Performance
  assessment and design of finite length \textmd{LDPC} codes for the
  \textmd{Gaussian} wiretap channel,'' in \emph{Proc. IEEE Int.\ Conf.\
  Commun.}, London, UK, June 2015, pp. 435--440.

\bibitem{9093879}
A.~{Nooraiepour}, S.~R. {Aghdam}, and T.~M. {Duman}, ``On secure communications
  over \textmd{Gaussian} wiretap channels via finite-length codes,'' \emph{IEEE
  Commun. Lett.}, vol.~24, no.~9, pp. 1904--1908, Sept. 2020.

\bibitem{7774989}
B.~Dai, Z.~Ma, and Y.~Luo, ``Finite state {Markov} wiretap channel with delayed
  feedback,'' \emph{IEEE Trans.\ Inf.\ Forensics Secur.}, vol.~12, no.~3, pp.
  746--760, Mar. 2017.

\bibitem{8737786}
J.~{de Dieu Mutangana} and R.~Tandon, ``Blind \text{MIMO} cooperative jamming:
  \text{Secrecy} via \text{ISI} heterogeneity without \text{CSIT},'' \emph{IEEE
  Trans.\ Inf.\ Forensics Secur.}, vol.~15, pp. 447--461, June 2020.

\bibitem{4276938}
A.~Thangaraj, S.~Dihidar, A.~R. Calderbank, S.~W. McLaughlin, and J.~M.
  Merolla, ``Applications of \textmd{LDPC} codes to the wiretap channel,''
  \emph{IEEE Trans.\ Inf.\ Theory}, vol.~53, no.~8, pp. 2933--2945, Aug. 2007.

\bibitem{TTC2015}
C.~B. Schlegel and L.~C. Perez, \emph{Trellis and Turbo Coding:
  \text{Iterative} and Graph-Based Error Control Coding, 2nd Edition}.\hskip
  1em plus 0.5em minus 0.4em\relax Hoboken, NJ, USA: John Wiley \& Sons, 2015.

\bibitem{1096620}
J.~K. {Wolf} and G.~{Ungerboeck}, ``Trellis coding for partial-response
  channels,'' \emph{IEEE Trans.\ Commun.}, vol.~34, no.~8, pp. 765--773, Aug.
  1986.

\bibitem{910573}
G.~Forney, ``Codes on graphs: \text{Normal} realizations,'' \emph{IEEE Trans.\
  Inf.\ Theory}, vol.~47, no.~2, pp. 520--548, Feb. 2001.

\bibitem{4494705}
P.~O. Vontobel, A.~Kavčić, D.~M. Arnold, and H.~A. Loeliger, ``A
  generalization of the \textmd{Blahut-Arimoto} algorithm to finite-state
  channels,'' \emph{IEEE Trans.\ Inf.\ Theory}, vol.~54, no.~5, pp. 1887--1918,
  May 2008.

\bibitem{Gallager:1968:ITR:578869}
R.~G. Gallager, \emph{Information Theory and Reliable Communication}.\hskip 1em
  plus 0.5em minus 0.4em\relax New~York, NY, USA: John Wiley \& Sons, 1968.

\bibitem{910579}
T.~Richardson and R.~Urbanke, ``Efficient encoding of low-density parity-check
  codes,'' \emph{IEEE Trans.\ Inf.\ Theory}, vol.~47, no.~2, pp. 638--656, Feb.
  2001.

\bibitem{1661831}
D.~M. Arnold, H.~A. Loeliger, P.~O. Vontobel, A.~Kavčić, and W.~Zeng,
  ``Simulation-based computation of information rates for channels with
  memory,'' \emph{IEEE Trans.\ Inf.\ Theory}, vol.~52, no.~8, pp. 3498--3508,
  Aug. 2006.

\bibitem{Maurer2000}
U.~Maurer and S.~Wolf, ``Information-theoretic key agreement: From weak to
  strong secrecy for free,'' in \emph{Proc. EUROCRYPT 2000}, ser. \emph{Lect.
  Notes Comput. Sci.}, B.~Preneel, Ed., vol. 1807.\hskip 1em plus 0.5em minus
  0.4em\relax Berlin, Germany: Springer, May 2000, pp. 351--368.

\bibitem{490940}
S.~{Benedetto}, D.~{Divsalar}, G.~{Montorsi}, and F.~{Pollara}, ``Algorithm for
  continuous decoding of turbo codes,'' \emph{Electron. Lett.}, vol.~32, no.~4,
  pp. 314--315, Feb. 1996.

\bibitem{910577}
T.~J. Richardson and R.~L. Urbanke, ``The capacity of low-density parity-check
  codes under message-passing decoding,'' \emph{IEEE Trans.\ Inf.\ Theory},
  vol.~47, no.~2, pp. 599--618, Feb. 2001.

\bibitem{905935}
S.~Y. Chung, G.~D. Forney, T.~J. Richardson, and R.~Urbanke, ``On the design of
  low-density parity-check codes within 0.0045 \textmd{dB} of the
  \textmd{Shannon} limit,'' \emph{IEEE Commun. Lett.}, vol.~5, no.~2, pp.
  58--60, Feb. 2001.

\bibitem{866615}
T.~Richardson and R.~Urbanke, ``Thresholds for turbo codes,'' in \emph{Proc.
  IEEE Int.\ Symp.\ Inf.\ Theory}, Sorrento, Italy, June 2000.

\bibitem{965976}
H.~Pfister, J.~Soriaga, and P.~Siegel, ``On the achievable information rates of
  finite state \text{ISI} channels,'' in \emph{Proc. IEEE Glob.\ Commun.\
  Conf.}, San Antonio, TX, USA, Nov. 2001, pp. 2992--2996.

\bibitem{HO1999169}
Y.~C. Ho, ``An explanation of ordinal optimization: \text{Soft} computing for
  hard problems,'' \emph{Inf.\ Sci.}, vol. 113, no.~3, pp. 169--192, Mar. 1999.

\bibitem{1055718}
H.~Imai and S.~Hirakawa, ``A new multilevel coding method using
  error-correcting codes,'' \emph{IEEE Trans.\ Inf.\ Theory}, vol.~23, no.~3,
  pp. 371--377, May 1977.

\bibitem{3GPP_TS_36_212}
{3rd Generation Partnership Project (3GPP)}, ``\text{TS 36.212 (V18.1.0):}
  \text{LTE}; \text{Evolved} universal terrestrial radio access
  \text{(E-UTRA)}; \text{Multiplexing} and channel coding \text{(Release
  18)},'' 3GPP, Technical Specification 36.212, Jan. 2025.

\bibitem{3GPP_TS_38_212}
------, ``\text{TS 38.212 (V18.6.0):} \text{5G;} \text{NR;} \text{Multiplexing}
  and channel coding \text{(Release 18)},'' 3GPP, Technical Specification
  38.212, Apr. 2025.

\bibitem{ISIWTC_CD}
A.~Nouri, ``{ISI Wiretap Channels [SIMULATION\_FILES] Code Design},''
  \url{https://doi.org/10.5281/zenodo.17267799}, Oct. 2025, version 1.0, Zenodo
  Repository.

\bibitem{Ha2004RCP}
J.~Ha, J.~Kim, and S.~W. McLaughlin, ``Rate-compatible puncturing of
  low-density parity-check codes,'' \emph{IEEE Trans.\ Inf.\ Theory}, vol.~50,
  no.~11, pp. 2824--2836, Nov. 2004.

\bibitem{proakis2008digital}
J.~G. Proakis and M.~Salehi, \emph{Digital Communications}, 5th~ed.\hskip 1em
  plus 0.5em minus 0.4em\relax New York, NY: McGraw-Hill Education, 2008.

\bibitem{8017512}
T.~S. Han, H.~Endo, and M.~Sasaki, ``Wiretap channels with one-time state
  information: \text{Strong} secrecy,'' \emph{IEEE Trans.\ Inf.\ Forensics
  Secur.}, vol.~13, no.~1, pp. 224--236, Jan. 2018.

\end{thebibliography}
	
	\setstretch{1}
		\begin{IEEEbiography}[{\includegraphics[width=1in,height=1.25in,clip,keepaspectratio]{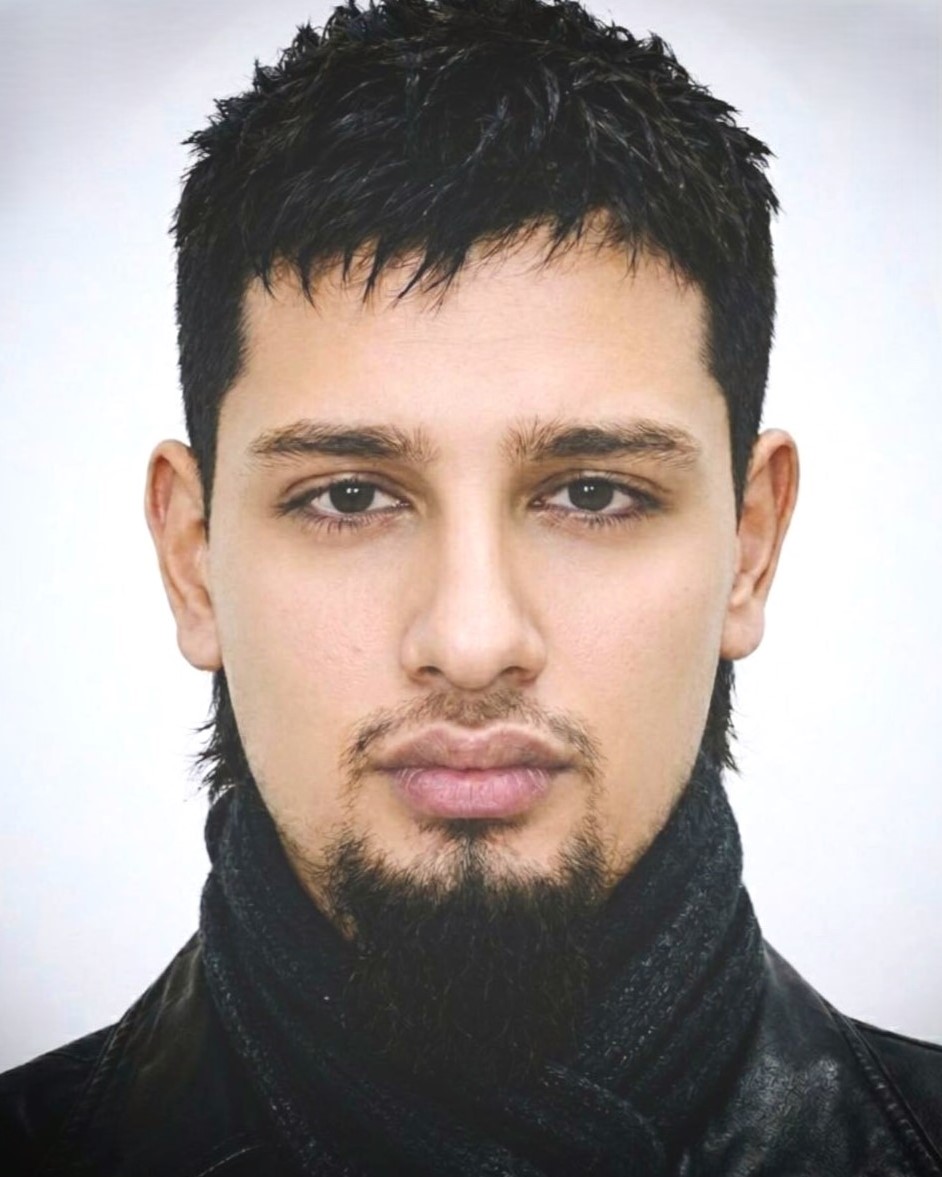}}]{Aria Nouri}
				(Graduate Student Member, IEEE) received the M.Sc.\ degree in Electrical Engineering (Communications) from Shahid Beheshti University, Tehran, Iran, in 2021. He was a Research Assistant with the Faculty of Electrical Engineering at Shahid Beheshti University from 2017 to 2023, and has been a Graduate Researcher there since 2023. His research interests lie in the theoretical aspects of digital communications, with a particular focus on coding theory, information theory, and wireless networks.
		\end{IEEEbiography}

		\begin{IEEEbiography}[{\includegraphics[width=1in,height=1.25in,clip,keepaspectratio]{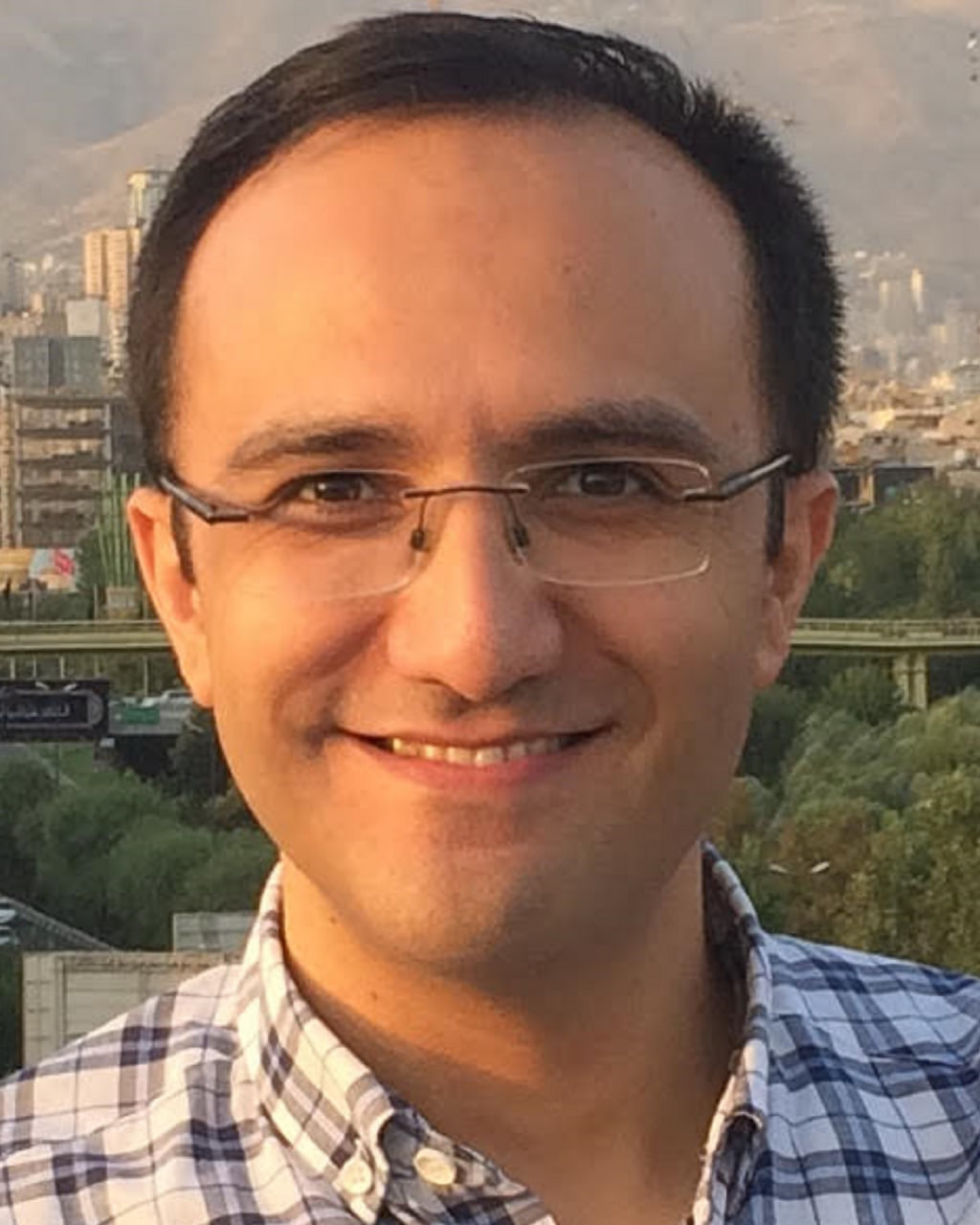}}]{Reza Asvadi}
				(Senior Member, IEEE) received the B.Sc.\ degree (with highest honors) from K.\ N.\ Toosi University of Technology, Tehran, Iran, in 2001, the M.Sc.\ degree from Sharif University of Technology, Tehran, in 2003, and the Ph.D.\ degree from K.\ N.\ Toosi University of Technology in 2011, all in electrical engineering.
				
				He has been an Assistant Professor at Shahid Beheshti University, Tehran, since 2016. From June 2023 to August 2024, he was a Visiting Scholar at Bilkent University, Ankara, Türkiye, working on insertion–deletion channels and DNA data storage. Previously, he was a Postdoctoral Researcher at the University of Oulu, Finland, from 2012 to 2014, where he contributed to several Academy of Finland and European Union (FP7) projects on iterative algorithms and information-theoretic bounds for emerging wireless networks. His research interests include coding and information theory and signal processing for wireless communications.
				
				Dr.\ Asvadi has received multiple postdoctoral research fellowships, including those from the University of Alberta (2011–2012) and Carleton University (2014–2016). He currently serves as an Associate Editor for IEEE Communications Letters and IEEE Transactions on Communications journals. He was recognized as an Exemplary Reviewer for IEEE Transactions on Communications in 2022 and as an Exemplary Editor for IEEE Communications Letters in 2025.
		\end{IEEEbiography}
		
		\begin{IEEEbiography}[{\includegraphics[width=1in,height=1.25in,clip,keepaspectratio]{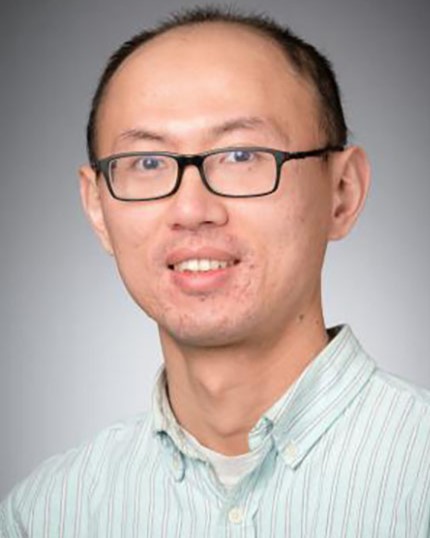}}]{Jun Chen}
				(Senior Member, IEEE) received the B.E.\ degree in communication engineering from Shanghai Jiao Tong University, Shanghai, China, in 2001, and the M.S.\ and Ph.D.\ degrees in electrical and computer engineering from Cornell University, Ithaca, NY, USA, in 2004 and 2006, respectively.
				
				From September 2005 to July 2006, he was a Post-Doctoral Research Associate with the Coordinated Science Laboratory, University of Illinois at Urbana-Champaign, Urbana, IL, USA, and a Post-Doctoral Fellow with the IBM Thomas J.\ Watson Research Center, Yorktown Heights, NY, USA, from July 2006 to August 2007. Since September 2007, he has been with the Department of Electrical and Computer Engineering, McMaster University, Hamilton, ON, Canada, where he is currently a Professor. His research interests include information theory, machine learning, wireless communications, and signal processing.
				
				Dr.\ Chen was a recipient of the Josef Raviv Memorial Postdoctoral Fellowship in 2006, the Early Researcher Award from the Province of Ontario in 2010, the IBM Faculty Award in 2010, the ICC Best Paper Award in 2020, and the JSPS Invitational Fellowship in 2021. He held the title of the Barber-Gennum Chair in Information Technology from 2008 to 2013 and the title of the Joseph Ip Distinguished Engineering Fellow from 2016 to 2018. He was an Associate Editor of IEEE Transactions on Information Theory from 2014 to 2016, an Editor of IEEE Transactions on Green Communications and Networking from 2020 to 2021, and an Associate Editor of IEEE Transactions on Communications from 2023 to 2026. In addition, he has served as a Guest Editor for several Special Issues of IEEE Journal on Selected Areas in Communications, IEEE Journal on Selected Areas in Information Theory, and IEEE BITS the Information Theory Magazine. He is currently an Associate Editor of IEEE Transactions on Information Theory and a Distinguished Lecturer of the IEEE Information Theory Society.
		\end{IEEEbiography}
	
\end{document}